\documentclass[11pt,a4paper]{article}

\usepackage{hyperref}  

\usepackage[T1]{fontenc}    
\usepackage{lmodern}
\usepackage{times}

\usepackage{pgf}
\usepackage{tikz}

\usepackage{doi}
\usepackage{cite}

\usepackage[fleqn]{amsmath}
\usepackage{amsfonts,amsthm,amssymb}

\numberwithin{equation}{section}

\newtheorem{theorem}{Theorem}[section]
\newtheorem{proposition}{Proposition}[section]
\newtheorem{lemma}{Lemma}[section]

\newtheorem{identity}{Identity}
\newtheorem{cor}{Corollary}[section]

{\theoremstyle{remark}

\DeclareMathSymbol{\Alpha}{\mathalpha}{operators}{"41}
\DeclareMathSymbol{\Beta}{\mathalpha}{operators}{"42}
\DeclareMathSymbol{\Epsilon}{\mathalpha}{operators}{"45}
\DeclareMathSymbol{\Zeta}{\mathalpha}{operators}{"5A}
\DeclareMathSymbol{\Eta}{\mathalpha}{operators}{"48}
\DeclareMathSymbol{\Iota}{\mathalpha}{operators}{"49}
\DeclareMathSymbol{\Kappa}{\mathalpha}{operators}{"4B}
\DeclareMathSymbol{\Mu}{\mathalpha}{operators}{"4D}
\DeclareMathSymbol{\Nu}{\mathalpha}{operators}{"4E}
\DeclareMathSymbol{\Omicron}{\mathalpha}{operators}{"4F}
\DeclareMathSymbol{\Rho}{\mathalpha}{operators}{"50}
\DeclareMathSymbol{\Tau}{\mathalpha}{operators}{"54}
\DeclareMathSymbol{\Chi}{\mathalpha}{operators}{"58}
\DeclareMathSymbol{\omicron}{\mathord}{letters}{"6F}

\begin{document}

LPENSL-TH-05/21

\bigskip

\bigskip

\begin{center}
\textbf{\Large Correlation functions for open XXX spin 1/2 quantum chains
with unparallel boundary magnetic fields} 


\vspace{45pt}
\end{center}

{\large \textbf{G. Niccoli}\footnote{{\large Univ Lyon, Ens de Lyon, Univ
Claude Bernard, CNRS, Laboratoire de Physique, F-69342 Lyon, France;
giuliano.niccoli@ens-lyon.fr}} }

\begin{center}
\vspace{45pt}

\begin{abstract}
In this first paper, we start the analysis of correlation functions of
quantum spin chains with general integrable boundary conditions. We initiate
these computations for the open XXX spin 1/2 quantum chains with some
unparallel magnetic fields allowing for a spectrum characterization in terms
of homogeneous Baxter like $TQ$-equations, in the framework of the quantum
separation of variables (SoV). Previous SoV analysis leads to the formula
for the scalar products of the so-called separate states. Here, we solve the
remaining fundamental steps allowing for the computation of correlation
functions. In particular, we rederive the ground state density in the
thermodynamic limit thanks to SoV approach, we compute the so-called
boundary-bulk decomposition of boundary separate states and the action of
local operators on these separate states in the case of unparallel boundary
magnetic fields. These findings allow us to derive multiple integral
formulae for these correlation functions similar to those previously known
for the open XXX quantum spin chain with parallel magnetic fields.
\end{abstract}

\today \vspace{45pt}
\end{center}

\newpage

\tableofcontents

\newpage

\section{Introduction}

The open integrable quantum spin chains with magnetic fields located at the
boundaries \cite{AlcBBBQ87} have attracted large scientific attention \cite%
{AlcBBBQ87,Skl88,GhoZ94,JimKKKM95,JimKKMW95,FanHSY96,Nep02,Nep04,CaoLSW03,YanZ07,Bas06,BasK07,KitKMNST07,KitKMNST08,CraRS10,CraRS11,FraSW08,FraGSW11,Nic12,CaoYSW13b,FalN14,FalKN14,KitMN14,BelC13,Bel15,BelP15,BelP15b,AvaBGP15,BelP16}%
. They have been used in connection to the studies of classical stochastic
models, as asymmetric simple exclusion models \cite{deGE05}, but also to
modelling numerous applications in quantum condensed matter physics, as
out-of-equilibrium and transport properties in the spin chains \cite{Pro11}.
Sklyanin \cite{Skl88} has extended to them the quantum inverse scattering
method (QISM) \cite%
{FadS78,FadST79,FadT79,Skl79,Skl79a,FadT81,Skl82,Fad82,Fad96,BogIK93L}, by
using the reflection equation introduced by Cherednik \cite{Che84}, in this
way introducing the natural algebraic framework to handle these open spin
chains. In this framework, the so-called boundary matrices, i.e. the scalar
solutions of the reflection equation \cite{Che84,deVG93,GhoZ94}, allow to
parametrize the magnetic fields at the boundaries of the quantum spin chain.
In particular, parallel boundary magnetic fields along the z-direction
correspond to diagonal boundary matrices while the unparallel cases
correspond to non-simultaneously diagonalizable boundary matrices. The open
spin chains associated to both diagonal boundary matrices have been first
analyzed by means of coordinate Bethe ansatz \cite{AlcBBBQ87} while Sklyanin
has generalized the algebraic Bethe ansatz (ABA) approach to these boundary
cases in its fundamental work \cite{Skl88}.

The cases with non-diagonal boundary matrices have proven themselves to be
more involved to analyze and, until recently, their spectrum has long
remained a very challenging problem in quantum integrability. Let us recall
that in \cite{Nep02} a first description of the spectrum of these open XXZ
spin 1/2 chains with non-diagonal boundary matrices has been obtained using
the fusion procedure \cite{KulRS81}, under a special constrain relating the
parameters of the two boundary matrices. There, the transfer matrix spectrum
has been described in terms of polynomial solutions of ordinary $TQ$%
-equation of Baxter's type \cite{Bax82L}, for the roots of unity points and
later in \cite{Nep04} for general value of the inhomogeneity parameter. This
constraint\footnote{%
It has emerged also in several other approaches leading to description of
the spectrum in terms of polynomial solutions of ordinary $TQ$-equations, as
by coordinate Bethe ansatz with elements of matrix product ansatz \cite%
{CraRS10,CraRS11}, $q$-Onsager algebra \cite{Bas06,BasK07} etc.} is required
also for generalized ABA-like constructions of eigenstates \cite%
{FanHSY96,CaoLSW03,YanZ07}, which use Baxter's gauge transformations \cite%
{Bax73a,Bax82L} to simplify the form of the boundary matrices.

Only more recently in \cite{CaoYSW13b}, on the basis of analytic properties
and functional relations satisfied by the transfer matrices a description of
the spectrum of open chains for the unconstrained cases has been proposed in
terms of polynomial solutions of inhomogeneous $TQ$-equations, i.e.
admitting some extra term. These inhomogeneous $TQ$ -equations have also
emerged in the framework of the so-called modified algebraic Bethe ansatz 
\cite{BelC13,Bel15,BelP15,AvaBGP15} to deal with unconstrained boundary
conditions.

The quantum Separation of Variables (SoV), pioneered by Sklyanin \cite%
{Skl85,Skl85a,Skl90,Skl92,Skl95,Skl96} in the QISM framework, was introduced
as an alternative to ABA approach for solving models in which a reference
state cannot be identified, as the Toda model. Then, it has been shown to be
applicable to a large class of models \cite%
{BabBS96,Smi98a,Smi01,DerKM01,DerKM03,DerKM03b,BytT06,vonGIPS06,FraSW08,AmiFOW10,NicT10,Nic10a,Nic11,FraGSW11,GroMN12,GroN12,Nic12,Nic13,Nic13a,Nic13b,GroMN14,FalN14,FalKN14,KitMN14,NicT15,LevNT16,NicT16,KitMNT16,JiaKKS16,KitMNT17,MaiNP17,KitMNT18}
and more recently reintroduced in \cite{MaiN18} on the pure basis of the
integrable structure of the models and widely extended even to higher rank
cases in \cite{MaiN18,RyaV19,MaiN19,MaiN19d,MaiN19b,MaiN19c,MaiNV20,RyaV20},
see also \cite{Skl96,Smi01,MarS16,GroLS17} for previous developments. The
SoV approaches in their different presentations have the built-in advantage
of the completeness of the characterization of eigenvalues and eigenstates
of the models. This method has been used, in particular, for open spin
chains with the most general unconstrained non-diagonal boundary matrices 
\cite{FraSW08,FraGSW11,Nic12,FalKN14,FalN14}. In \cite{KitMN14}, it has been
proven that the complete SoV characterization of the spectrum of these open
chains can be reformulated in terms of polynomial solutions to functional $%
TQ $-equations of Baxter type that have the aforementioned inhomogeneous
extra term, for the most general unconstrained boundary matrices. Another
advantage of the SoV approaches, of particular relevance for the correlation
functions analysis, is the natural and universal emergence of determinant
formulae for the scalar products of \textit{separate states} \cite%
{GroMN12,Nic12,Nic13,Nic13a,Nic13b,GroMN14,FalN14,FalKN14,LevNT16,KitMNT16,KitMNT17,MaiNP17,KitMNT18}%
\footnote{%
This is surely the case for the rank one models and in \cite{MaiNV20b} we
have proven it for the higher rank gl(3) case under a special choice of the
conserved charges generating the SoV bases. See also the interesting and
recent papers \cite{CavGL19,GroLRV20} dealing with the computations of
higher rank scalar products in a related SoV framework.}. This is a class of
states with factorized wave-functions in the SoV basis which includes the
eigenstates of the transfer matrix.

The short recall here presented clarify the progresses achieved about the
knowledge of the spectrum of these open spin 1/2 quantum chains for general
integrable boundary conditions. One should, however, remark that the study
of the ground state in the thermodynamic limit of these models presents
still some open problems for non-diagonal boundary matrices and, in
particular, in the unconstrained cases\footnote{%
See, however, \cite{BelF18} for a numerical analysis of the Bethe roots of
these inhomogeneous equations.}.

Let us now focus our attention on the study of correlation functions which
is the main subject of the current paper. In this area of research, one
should stress that main results for interacting models are available for the
XXX/XXZ spin 1/2 quantum chains or the quantum non-linear Schr\"odinger
model with periodic boundary conditions. The current state of the art is
quite unsatisfactory in connection to the integrable boundary conditions for
which only few results are so far available.

Indeed, these correlation functions have been computed first for the case of
closed chains with periodic boundary conditions\footnote{%
Interesting results on correlation functions related to the hidden Grassmann
structure have been derived in \cite%
{BooJMST05,BooJMST06,BooJMST06a,BooJMST06b,BooJMST06c,BooJMST07,BooJMST09,JimMS09,JimMS11,MesP14,Poz17}%
.}, for the zero-temperature cases \cite{JimM95L,JimMMN92,JimM96} and
non-zero global magnetic fields \cite%
{KitMT99,MaiT00,KitMT00,KitMST02a,KitMST05a,KitMST05b,KitKMST07}, and
subsequently for the temperature cases \cite%
{GohKS04,GohKS05,BooGKS07,GohS10,DugGKS15,GohKKKS17}. Always in this
periodic setting, further developments have led to the analytical study in
the thermodynamic limit of long distances two-point and multi-point
correlation functions \cite%
{KitKMST09a,KitKMST09b,KitKMST09c,KozMS11a,KozMS11b,KozT11,KitKMST11a,KitKMST11b,KitKMST12,DugGK13,KitKMT14}
and the numerical study of the dynamical structure factors \cite%
{CauHM05,CauM05,PerSCHMWA06}, accessible experimentally through neutron
scattering \cite{KenCTVHST02}.

So far, the main exception to periodic boundary conditions for a closed
chain is our computation of correlation functions in the case of
antiperiodic boundary conditions in the SoV framework \cite{NicPT20}. While
for XXX and XXZ quantum spin 1/2 chains with open boundaries the only
available results on correlation functions are those derived for
zero-temperature in the case of parallel magnetic fields on the boundary for
the XXX chains and of parallel magnetic fields along the Z-direction for the
XXZ chains \cite{JimKKKM95,JimKKMW95,KitKMNST07,KitKMNST08}.

Here, we start the analysis of correlation functions of quantum spin chains
with general integrable boundary conditions beyond those so far analyzed.
The aim of this paper is to derive correlation functions for the open XXX
spin 1/2 quantum chains under unparallel boundary magnetic fields. This
achievement represents on one side a first access to correlation functions
for these more general boundary conditions and on the other hand is
instrumental to introduce some technical ingenuity which will be then used
also in more involved models like open XXZ/XYZ quantum spin 1/2 chains whose
correlation functions will be derived in our forthcoming papers. We develop
our analysis in the framework of the quantum separation of variables (SoV)
considering a generic magnetic field in the first site 1 of the XXX spin 1/2
quantum chain. While, we adjust the magnetic field in the last site $N$ of
the chain such that it isn't parallel to the one in site 1 but it allows for
an SoV complete description of the transfer matrix spectrum in terms of
homogeneous Baxter like $TQ$-equations.

Once the transfer matrix spectral problem is characterized by SoV approach,
we derive the following four main steps to compute correlation functions in
the SoV framework: i. A decomposition formula, \textit{boundary-bulk
decomposition}, expressing the so-called boundary separate states (a class
of states containing the transfer matrix eigenstates of the open chain) in
term of analogous states generated by bulk operators, associated to the
closed chain. ii. A decomposition formula over the boundary separate states
for the action of local operators on a generic boundary separate state. iii.
Simple determinant formulae defining the scalar products between left and
right boundary separate states and their particularization when one of the
states is a transfer matrix eigenstate \cite{KitMNT17}. iv. The density
distribution of the ground state Bethe's roots in the thermodynamic limit.

In our current analysis the built-in feature of completeness of the SoV
method, both in its original Sklyanin's like generalization to open chain 
\cite{Nic12,FalKN14,FalN14} as well as in our new SoV approach \cite%
{MaiN18,MaiN19}, plays a fundamental role in the characterization of the
ground state in the thermodynamic limit. Indeed, by using the complete SoV
characterization of the transfer matrix spectrum, we can prove the
isospectrality of the transfer matrices under consideration with specific
ones with parallel boundary magnetic fields, in this way recovering the
traditional thermodynamic results for the ground state. On the other hand,
in the SoV framework, we have at our disposal the scalar products of the
boundary separate states \cite{KitMNT16,KitMNT17,KitMNT18}.

It is however worth mentioning that, the recent and interesting results on
scalar products for open chains with general boundary conditions\footnote{%
See \cite{BelS19a} and also \cite{SlaZZ20} for the original idea developed
first for periodic chains and see also \cite{BelP16} for a first conjecture
on these determinant formulae and \cite{FilK11,YanCFHHSZ11,DuvP15} for
previous determinant representations under special boundary constraints.} 
\cite{BelPS21} put a basis for the computation of correlation functions in a
generalized/modified Bethe Ansatz framework\footnote{%
One should notice that a priori the scalar products analyzed in \cite%
{BelPS21} are between left $C$-gauged Bethe like states and right $B$-gauged
Bethe like states so a priori different w.r.t. to those of our papers \cite%
{KitMNT16,KitMNT17,KitMNT18}, which are computed between left and right
separate states. However, the completeness of the spectrum description by
SoV approach \cite{MaiN19} and the therein proven simplicity of the transfer
matrix spectrum can be used to relate these scalar products once the left
C-gauged Bethe like state is a transfer matrix eigenstate.}.

The paper is arranged in the following sections. In Section 2, we make a
brief introduction to the XXX quantum spin 1/2 chain and to the reflection
algebra. In Section 3, we first recall and rework in the new SoV framework 
\cite{MaiN18,MaiN19} the characterization of the transfer matrix spectrum
for general boundary conditions. Then, we setup the one constraint boundary conditions for
the XXX model that we will use to determine correlation functions. More
precisely, by the SoV approach, we show that we can leave the magnetic field
on site 1 arbitrary while adjusting the one in site N such that they are
kept unparallel, the transfer matrix is proven to be isospectral to one
associated to parallel boundary magnetic field and their spectrum is
completely characterized by a homogeneous Baxter's like $TQ$-equation. In
Section 4, we derive the boundary-bulk decomposition of the boundary
separate states and the action on them of local operators. In Section 5, we
recall and rearrange the known results on the scalar products of separate
states \cite{KitMNT17}. In Section 6, we compute the correlation functions
in terms of multiple integral representations. The final section contain
some conclusion and outlooks.

\section{The open XXX quantum spin 1/2 chain}

The Hamiltonian of the open XXX quantum spin-1/2  chain with the most
general boundary magnetic fields reads: 
\begin{align}
H& =\sum_{i=1}^{N-1}\left[ \sigma _{i}^{x}\sigma _{i+1}^{x}+\sigma
_{i}^{y}\sigma _{i+1}^{y}+\sigma _{i}^{z}\sigma _{i+1}^{z}\right] +\frac{%
\eta }{\zeta _{-}}\left[ \sigma _{1}^{z}+2\kappa _{-}\left( e^{\tau
_{-}}\sigma _{1}^{+}+e^{-\tau _{-}}\sigma _{1}^{-}\right) \right]  \notag \\
& +\frac{\eta }{\zeta _{+}}\left[ \sigma _{N}^{z}+2\kappa _{+}\left( e^{\tau
_{+}}\sigma _{N}^{+}+e^{-\tau _{+}}\sigma _{N}^{-}\right) \right] .
\label{H|XXX-ND}
\end{align}%
The local spin-$1/2$ operators (Pauli matrices) $\sigma _{i}^{\alpha }$, for 
$\alpha =x,y,z$, acts on the local quantum space $\mathcal{H}_{i}\simeq 
\mathbb{C}^{2}$ at site $i$, $\eta $ is a fixed arbitrary parameter, and the
six complex boundary parameters $\zeta _{\pm }$, $\kappa _{\pm }$ and $\tau
_{\pm }$ parametrize the coupling of the spin operators at site $1$ and $N$
with two arbitrary boundary magnetic fields.

Following the seminal Sklyanin's paper \cite{Skl88}, this Hamiltonian can be
obtained as the following derivative:%
\begin{equation}
H=\frac{2\,\eta ^{1-2N}}{\text{tr}\{K_{+}(\eta /2)\}\,\text{tr}\{K_{-}(\eta
/2)\}}\,\frac{d}{d\lambda }\mathcal{T}(\lambda )_{\,\vrule %
height13ptdepth1pt\>{\lambda =\eta /2}\!}+\text{constant,}  \label{H|Txxx}
\end{equation}%
of the one-parameter family of commuting boundary transfer matrices%
\begin{align}
\mathcal{T}(\lambda )& =\text{tr}_{0}\{K_{0,+}(\lambda )\,M_{0}(\lambda
)\,K_{0,-}(\lambda )\,\hat{M}_{0}(\lambda )\}  \notag \\
& =\text{tr}_{0}\left\{ K_{+}(\lambda )\,\mathcal{U}_{-}(\lambda )\right\} =%
\text{tr}_{0}\left\{ K_{-}(\lambda )\,\mathcal{U}_{+}(\lambda )\right\} \in 
\text{End}(\mathcal{H}),  \label{NDtransfer}
\end{align}%
on the 2$^{N}$-dimensional linear space $\mathcal{H}=\otimes _{n=1}^{N}%
\mathcal{H}_{n}$, the physical space of states of the Hamiltonian %
\eqref{H|XXX-ND}, where we have defined the boundary monodromy matrices 
\begin{align}
& \mathcal{U}_{-}(\lambda )=M_{0}(\lambda )\,K_{-}(\lambda )\,\hat{M}%
_{0}(\lambda )=\left( 
\begin{array}{cc}
\mathcal{A}_{-}(\lambda ) & \mathcal{B}_{-}(\lambda ) \\ 
\mathcal{C}_{-}(\lambda ) & \mathcal{D}_{-}(\lambda )%
\end{array}%
\right) \in \text{End}(\mathcal{H}_{0}\otimes \mathcal{H}),  \label{U-Def} \\
& \mathcal{U}_{+}^{t_{0}}(\lambda )=M_{0}^{t_{0}}(\lambda
)\,K_{+}^{t_{0}}(\lambda )\,\hat{M}_{0}^{t_{0}}(\lambda )=\left( 
\begin{array}{cc}
\mathcal{A}_{+}(\lambda ) & \mathcal{C}_{+}(\lambda ) \\ 
\mathcal{B}_{+}(\lambda ) & \mathcal{D}_{+}(\lambda )%
\end{array}%
\right) \in \text{End}(\mathcal{H}_{0}\otimes \mathcal{H}).  \label{U+Def}
\end{align}%
and the bulk monodromy matrices $M_{0}(\lambda )\in $End$(\mathcal{H}%
_{0}\otimes \mathcal{H})$ by%
\begin{align}
M_{0}(\lambda )& =R_{0N}(\lambda -\xi _{N}-\eta /2)\dots R_{01}(\lambda -\xi
_{1}-\eta /2)=\left( 
\begin{array}{cc}
A(\lambda ) & B(\lambda ) \\ 
C(\lambda ) & D(\lambda )%
\end{array}%
\right) ,  \label{mon-Def} \\
\hat{M}_{0}(\lambda )& =(-1)^{N}\,\sigma _{0}^{y}\,M_{0}^{t_{0}}(-\lambda
)\,\sigma _{0}^{y}=(-1)^{N}\left( 
\begin{array}{cc}
D(-\lambda ) & B(-\lambda ) \\ 
C(-\lambda ) & A(-\lambda )%
\end{array}%
\right) \in \text{End}(\mathcal{H}_{0}\otimes \mathcal{H}),
\end{align}%
for arbitrary complex parameters $\xi _{n}$, $1\leq n\leq N$, the so-called
inhomogeneities. Moreover, the $R$-matrix of the model, 
\begin{equation}
R(\lambda )=\left( 
\begin{array}{cccc}
\lambda +\eta & 0 & 0 & 0 \\ 
0 & \lambda & \eta & 0 \\ 
0 & \eta & \lambda & 0 \\ 
0 & 0 & 0 & \lambda +\eta%
\end{array}%
\right) \ \in \text{End}(\mathbb{C}^{2}\otimes \mathbb{C}^{2}),
\label{Rmatrix-Def}
\end{equation}%
is the 6-vertex polynomial solution of the Yang-Baxter equation and the two
boundary $K$-matrices%
\begin{equation}
K_{-}(\lambda )=K(\lambda -\eta /2;\zeta _{-},\kappa _{-},\tau _{-}),\qquad
K_{+}(\lambda )=K(\lambda +\eta /2;\zeta _{+},\kappa _{+},\tau _{+}),
\label{Kpm-Def}
\end{equation}%
with the boundary parameters appearing $\zeta _{\pm },\kappa _{\pm },\tau
_{\pm }$, coinciding with those of \eqref{H|XXX-ND}, are defined in terms of 
\cite{deVG93,GhoZ94} 
\begin{equation}
K(\lambda ;\zeta ,\kappa ,\tau )=\frac{1}{\zeta }\left( 
\begin{array}{cc}
\zeta +\lambda & 2\kappa e^{\tau }\lambda \\ 
2\kappa e^{-\tau }\lambda & \zeta -\lambda%
\end{array}%
\right) ,  \label{Kxxx-Def}
\end{equation}%
the most general non-diagonal scalar solution $K(\lambda )\in \text{End}(%
\mathbb{C}^{2})$ of the reflection equation \cite{Che84}: 
\begin{equation}
R_{ab}(\lambda -\mu )\,K_{a}(\lambda )\,R_{ab}(\lambda +\mu )\,K_{b}(\mu
)=K_{b}(\mu )\,R_{ab}(\lambda +\mu )\,K_{a}(\lambda )\,R_{ab}(\lambda -\mu
)\in \text{End}(\mathcal{H}_{a}\otimes \mathcal{H}_{b}).
\end{equation}%
Then, the bulk monodromy matrix is solution of the Yang-Baxter equation:%
\begin{equation}
R_{ab}(\lambda -\mu )\,M_{a}(\lambda )\,M_{b}(\mu )=M_{b}(\mu
)\,M_{a}(\lambda )\,R_{ab}(\lambda -\mu )\in \text{End}(\mathcal{H}%
_{a}\otimes \mathcal{H}_{b}\otimes \mathcal{H}),  \label{RTT-Def}
\end{equation}%
while the two boundary monodromy matrices $\mathcal{V}_{-}(\lambda )=%
\mathcal{U}_{-}(\lambda +\eta /2)$ and $\mathcal{V}_{+}(\lambda )=\mathcal{U}%
_{+}^{t_{0}}(-\lambda -\eta /2)$ are solutions of the reflection equation: 
\begin{equation}
R_{ab}(\lambda -\mu )\,\mathcal{U}_{a}(\lambda )\,R_{ab}(\lambda +\mu )\,%
\mathcal{U}_{b}(\mu )=\mathcal{U}_{b}(\mu )\,R_{ab}(\lambda +\mu )\,\mathcal{%
U}_{a}(\lambda )\,R_{ab}(\lambda -\mu )\in \text{End}(\mathcal{H}_{a}\otimes 
\mathcal{H}_{b}\otimes \mathcal{H}).  \label{refl-eq-Def}
\end{equation}%
As shown in \cite{Skl88}, the commutativity of the boundary transfer
matrices is implied by these reflection equations as well as the following
inversion relation for the boundary monodromy matrix $\mathcal{U}_{\pm
}(\lambda )$, 
\begin{equation}
\mathcal{U}_{\pm }(\lambda +\eta /2)\ \mathcal{U}_{\pm }(-\lambda +\eta /2)=%
\frac{\mathrm{det}_{q}\,\mathcal{U}_{\pm }(\lambda )}{2(\lambda \pm \eta )},
\label{invert-U-Def}
\end{equation}%
where $\mathrm{det}_{q}\,\mathcal{U}_{\pm }(\lambda )$ are the quantum
determinants, which are central elements of the corresponding boundary
algebra:%
\begin{equation}
\big[\mathrm{det}_{q}\,\mathcal{U}_{\pm }(\lambda ),\mathcal{U}_{\pm }(\mu )%
\big]=0.
\end{equation}%
They can be expressed as 
\begin{equation}
\mathrm{det}_{q}\,\mathcal{U}_{\pm }(\lambda )=\mathrm{det}_{q}M(\lambda )\,%
\mathrm{det}_{q}M(-\lambda )\,\mathrm{det}_{q}K_{\pm }(\lambda ),
\label{detqU-Def}
\end{equation}%
where 
\begin{equation}
\mathrm{det}_{q}M(\lambda )=a(\lambda +\eta /2)\,d(\lambda -\eta /2),
\label{detqM-Def}
\end{equation}%
is the bulk quantum determinant, as well a central element, with 
\begin{equation}
a(\lambda )\equiv \prod_{n=1}^{N}(\lambda -\xi _{n}+\eta /2),\qquad
d(\lambda )\equiv \prod_{n=1}^{N}(\lambda -\xi _{n}-\eta /2),
\label{a-d-Def}
\end{equation}%
and%
\begin{equation}
\mathrm{det}_{q}K_{\pm }(\lambda )=\pm 2(\lambda \pm \eta )\left( \frac{%
\lambda ^{2}}{\bar{\zeta}_{\pm }^{2}}\,-1\right) ,  \label{detqK-Def}
\end{equation}%
is the quantum determinant of the scalar boundary matrix $K_{\pm }(\lambda )$%
, where we have defined:%
\begin{equation}
\bar{\zeta}_{\pm }=\epsilon _{\pm }\frac{\zeta _{\pm }}{\sqrt{1+4\kappa
_{\pm }^{2}}},  \label{Zi-bar-Def}
\end{equation}%
with $\epsilon _{\pm }$ plus or minus one at will. The following fundamental
identities relate the transfer matrix at special values to the quantum
determinant:%
\begin{equation}
\mathcal{T}(\xi _{a}-\eta /2)\mathcal{T}(\xi _{a}+\eta /2)=\frac{%
\det_{q}K_{\pm }(\xi _{a})\,\det_{q}\mathcal{U}_{\mp }(\xi _{a})}{\eta
^{2}-4\xi _{a}^{2}},
\end{equation}%
here we also introduced the notation%
\begin{equation}
k_{n}=(\xi _{n}+\eta )/(\xi _{n}-\eta )
\end{equation}%
and the function%
\begin{equation}
\mathsf{A}_{\bar{\zeta}_{+},\bar{\zeta}_{-}}(\lambda )\equiv (-1)^{N}\frac{%
2\lambda +\eta }{2\lambda }\,\frac{(\lambda -\frac{\eta }{2}+\bar{\zeta}%
_{+})(\lambda -\frac{\eta }{2}+\bar{\zeta}_{-})}{\bar{\zeta}_{+}\,\bar{\zeta}%
_{-}}\,a(\lambda )\,d(-\lambda ),  \label{ATot_pm-Def}
\end{equation}%
which allows for the following explicit writing of the quantum determinant 
\begin{equation}
\frac{\det_{q}K_{\pm }(\lambda )\,\det_{q}\mathcal{U}_{\mp }(\lambda )}{\eta
^{2}-4\lambda ^{2}}=\mathsf{A}_{\bar{\zeta}_{+},\bar{\zeta}_{-}}(\lambda
+\eta /2)\,\mathsf{A}_{\bar{\zeta}_{+},\bar{\zeta}_{-}}(-\lambda +\eta /2),
\end{equation}%
that will be used in the following.

It is import to remark that the transfer matrix $\mathcal{T}(\lambda )$ is a
polynomial function of degree $N+1$ in the variable $\lambda ^{2}$ and its
leading coefficient is given by, 
\begin{equation}
t_{N+1}\,\lambda ^{2(N+1)}\,\mathrm{Id},\quad \text{with}\quad t_{N+1}=\frac{%
2}{\zeta _{+}\zeta _{-}}[1+4\kappa _{+}\kappa _{-}\cosh (\tau _{+}-\tau
_{-})],
\end{equation}%
and that its value in $\lambda =\pm \eta /2$ is central:%
\begin{equation}
\mathcal{T}(\pm \eta /2)=2(-1)^{N}\mathrm{det}_{q}M(0).
\end{equation}%
Then, we have the following interpolation formula%
\begin{equation}
\mathcal{T}(\lambda )=t_{N+1}u_{\mathbf{h}}(\lambda )+\mathcal{T}(\eta /2)s_{%
\mathbf{h}}(\lambda )+\sum_{a=1}^{N}r_{a,\mathbf{h}}(\lambda )\mathcal{T}%
(\xi _{a}^{\left( h_{a}\right) }),
\end{equation}%
where%
\begin{align}
t_{N+1}& =\frac{2}{\zeta _{+}\zeta _{-}}[1+4\kappa _{+}\kappa _{-}\cosh
(\tau _{+}-\tau _{-})], \\
r_{a,\mathbf{h}}(\lambda )& =\frac{\lambda ^{2}-\left( \eta /2\right) ^{2}}{%
(\xi _{a}^{(h_{a})})^{2}-\left( \eta /2\right) ^{2}}\prod_{b\neq a,b=1}^{N}%
\frac{\lambda ^{2}-(\xi _{b}^{(h_{b})})^{2}}{(\xi _{a}^{(h_{a})})^{2}-(\xi
_{b}^{(h_{b})})^{2}}\ , \\
s_{\mathbf{h}}(\lambda )& =\prod_{b=1}^{N}\frac{\lambda ^{2}-(\xi
_{b}^{(h_{b})})^{2}}{(\eta /2)^{2}-(\xi _{b}^{(h_{b})})^{2}}\ , \\
u_{\mathbf{h}}(\lambda )& =(\lambda ^{2}-(\eta /2)^{2})\prod_{b=1}^{\mathsf{N%
}}(\lambda ^{2}-(\xi _{b}^{(h_{b})})^{2})\ ,
\end{align}%
and%
\begin{equation}
\xi _{b}^{(h_{b})}=\xi _{b}^{(h_{b})}+(1-2h_{b})\eta /2.
\end{equation}

\section{Separation of variable spectrum characterization}

In this section, we further develops known results about the transfer matrix
spectrum characterization of the open XXX spin 1/2 quantum chains, with
generic integrable boundaries, in the framework of the recent formulation of
the quantum separation of variables (SoV) \cite{MaiN18,MaiN19}. This
analysis allows us to give a uniform description of the transfer matrix
spectrum independently from the boundary conditions. More in detail, we
write explicitly the left and right SoV basis and the left and right
transfer matrix eigenstates without the need to distinguish between the
cases of parallel or unparallel magnetic fields. Distinction which is
instead essential in Sklyanin's like SoV framework \cite{Skl85,Skl90,Skl92},
holding only in the case of unparallel fields\footnote{%
Observation which may lead to the false perception of a dichotomy unparallel
case solvable by SoV method and parallel case by ordinary Algebraic Bethe
Ansatz.}. These SoV results allow us to prove that the transfer matrices
associated to unparallel boundary magnetic fields satisfying one specific
boundary condition are isospectral to those associated to the parallel case.
This isospectrality holds up to the relative Hamiltonians and it is an
important fact for the computation of the thermodynamic limit of the ground
state.

\subsection{Discrete SoV spectrum characterization}

\subsubsection{Covector and vector SoV bases}

Here, we complete the construction of the left and right SoV basis in our
new SoV approach further deriving their left/right couplings, i.e. the SoV
measure. As a corollary of the Theorem 2.1 of our paper \cite{MaiN19} the
following proposition holds:

\begin{proposition}
Let us suppose that the inhomogeneity parameters are generic: 
\begin{equation}
\xi _{j},\xi _{j}\pm \xi _{k}\notin \{0,-\eta ,\eta \},\quad \forall j,k\in
\{1,\ldots ,N\},\ j\neq k,  \label{cond-inh}
\end{equation}%
and that the boundary matrices $K_{-}(\lambda )$ and $K_{+}(\lambda )$ are
not both proportianl to the identity, then%
\begin{equation}
\langle \mathbf{h}|\equiv \langle S|\prod_{n=1}^{N}\left( \frac{\mathcal{T}%
(\xi _{n}-\eta /2)}{\mathsf{A}_{\bar{\zeta}_{+},\bar{\zeta}_{-}}(\eta /2-\xi
_{n})}\right) ^{1-h_{n}},\ \ \mathbf{h}\equiv (h_{1},\ldots ,h_{N})\in
\{0,1\}^{N}\text{, }  \label{SoV-Basis-6v-Open}
\end{equation}%
is a co-vector basis of $\mathcal{H}$ for almost any choice of the co-vector 
$\langle S|$, of the inhomogeneity parameters satisfying $(\ref{cond-inh})$
and of the boundary parameters. So, denoted with $|R\rangle $ the unique
vector satisfying the following orthogonality conditions:%
\begin{equation}
\langle h_{1},...,h_{N}|R\rangle =\delta _{\mathbf{h},\mathbf{0}}\,\frac{%
\widehat{V}(\xi _{1},\ldots ,\xi _{N})}{\widehat{V}(\xi _{1}^{(0)},\ldots
,\xi _{N}^{(0)})},
\end{equation}%
we have that the following set of vectors:%
\begin{equation}
|\mathbf{h}\rangle \equiv \prod_{n=1}^{N}\left( \frac{\mathcal{T}(\xi
_{n}+\eta /2)}{k_{n}\mathsf{A}_{\bar{\zeta}_{+},\bar{\zeta}_{-}}(\eta /2-\xi
_{n})}\right) ^{h_{n}}|R\rangle ,\ \ \ \mathbf{h}\in \{0,1\}^{N}\text{, }
\end{equation}%
is a vector basis of $\mathcal{H}$ and the two basis are orthogonal: 
\begin{equation}
\langle \,\mathbf{h}^{\prime }\mid \mathbf{h}\,\rangle =\delta _{\mathbf{h},%
\mathbf{h}^{\prime }}\,\frac{N_{\boldsymbol{\xi}}}{\widehat{V}(\xi
_{1}^{(h_{1})},\ldots ,\xi _{N}^{(h_{N})})}\,,  \label{Ortho-norm}
\end{equation}%
with 
\begin{equation}
N_{\boldsymbol{\xi}}=\widehat{V}(\xi _{1},\ldots ,\xi _{N})\,\frac{\widehat{V%
}(\xi _{1}^{(0)},\ldots ,\xi _{N}^{(0)})}{\widehat{V}(\xi _{1}^{(1)},\ldots
,\xi _{N}^{(1)})}\,,  \label{norm-factor-F}
\end{equation}
and 
\begin{equation}
\widehat{V}(x_{1},\ldots ,x_{N})=\det_{1\leq i,j\leq N}\left[ x_{i}^{2(j-1)}%
\right] =\prod\limits_{j<k}(x_{k}^{2}-x_{j}^{2}),  \label{VDM-Def}
\end{equation}%
for the Vandermonde determinant of a $N$-tuple of square variables $%
(x_{1}^{2},\ldots ,x_{N}^{2})$.
\end{proposition}

\begin{proof}
The proof that the set of covectors is a basis is detailed in Theorem 2.1 of 
\cite{MaiN19}. We prove the rest of the proposition just proving the
orthogonality conditions $\left( \ref{Ortho-norm}\right) $. These conditions
are satisfied by definition by the vector $|R\rangle $ so we can prove them
by induction. That is, we assume that they hold for a fixed $\mathbf{h}$
such that%
\begin{equation}
\sum_{n=1}^{N}h_{n}=m,
\end{equation}%
and then we prove them for the generic $\mathbf{\bar{h}}$ obtained changing
one of the elements of $\mathbf{h}$ from zero to 1, i.e. the $m+1$ case.
Then, there exists a permutation of $\{1,...,N\}$ such that:%
\begin{equation}
h_{\pi (n)}=\{1\text{ for }n\leq m\text{ \ and \ }0\text{ for }m<n\}
\end{equation}%
and we want to prove that for any $a\in \{\pi (m+1),...,\pi (N)\}$ it holds:%
\begin{equation}
\langle \,\mathbf{h}^{\prime }\mid \mathbf{\bar{h}}\,\rangle =\delta _{%
\mathbf{\bar{h}},\mathbf{h}^{\prime }}\,\frac{N_{\boldsymbol{\xi}}}{\widehat{%
V}(\xi _{1}^{(\bar{h}_{1})},\ldots ,\xi _{N}^{(\bar{h}_{N})})}\,,
\end{equation}%
where:%
\begin{equation}
\bar{h}_{b}=h_{b}\text{ }\forall b\neq a\text{ \ and }\bar{h}_{a}=1.
\end{equation}%
Let us start proving the orthogonality condition for $h_{a}^{\prime }=0$, we
have that it holds:%
\begin{align}
\langle \,\mathbf{h}^{\prime }|\mathbf{\bar{h}}\,\rangle & =\langle
h_{1}^{\prime },...,h_{a}^{\prime \prime }=1,...,h_{N}^{\prime }|\frac{%
\mathcal{T}(\xi _{a}-\eta /2)\mathcal{T}(\xi _{a}+\eta /2)}{k_{a}\mathsf{A}_{%
\bar{\zeta}_{+},\bar{\zeta}_{-}}(\eta /2-\xi _{a})\mathsf{A}_{\bar{\zeta}%
_{+},\bar{\zeta}_{-}}(\eta /2-\xi _{a})}|\mathbf{h}\rangle  \notag \\
& =\frac{\mathsf{A}_{\bar{\zeta}_{+},\bar{\zeta}_{-}}(\eta /2+\xi _{a})}{%
k_{a}\mathsf{A}_{\bar{\zeta}_{+},\bar{\zeta}_{-}}(\eta /2-\xi _{a})}\langle
h_{1}^{\prime },...,h_{a}^{\prime \prime }=1,...,h_{N}^{\prime }|\mathbf{h}%
\rangle =0
\end{align}%
where we have used the quantum determinant relation and we get zero by the
induction being $\mathbf{h}$ such that it holds $h_{a}=0$. Let us now show
the orthogonality condition for $h_{a}^{\prime }=1$ and $\mathbf{h}^{\prime
}\neq \mathbf{\bar{h}}$, this is the case if it exists a $b\neq a$ such that 
$h_{b}^{\prime }=1-h_{b}$.

Then, to compute the action of $\mathcal{T}(\xi _{a}+\eta /2)$ on $\langle \,%
\mathbf{h}^{\prime }|$, we use the following interpolation formula:%
\begin{equation}
\mathcal{T}(\lambda )=t_{N+1}u_{\mathbf{h}^{\prime }}(\lambda )+\mathcal{T}%
(\eta /2)s_{\mathbf{h}^{\prime }}(\lambda )+\sum_{a=1}^{N}r_{a,\mathbf{h}%
^{\prime }}(\lambda )\mathcal{T}(\xi _{a}^{\left( h_{a}^{\prime }\right) }),
\end{equation}%
to get:%
\begin{align}
\langle \,\mathbf{h}^{\prime }|\mathbf{\bar{h}}\,\rangle & =(t_{N+1}u_{%
\mathbf{h}^{\prime }}(\xi _{a}+\eta /2)+\mathcal{T}(\eta /2)s_{\mathbf{h}%
^{\prime }}(\xi _{a}+\eta /2))\langle \,\mathbf{h}^{\prime }|\mathbf{h}%
\,\rangle +\sum_{c=1}^{N}r_{c,\mathbf{h}^{\prime }}(\xi _{a}+\eta /2)  \notag
\\
& \times \frac{\mathsf{A}_{\bar{\zeta}_{+},\bar{\zeta}_{-}}(\eta /2-\xi _{c})%
}{k_{a}\mathsf{A}_{\bar{\zeta}_{+},\bar{\zeta}_{-}}(\eta /2-\xi _{a})}%
\langle \,h_{1}^{\prime },...,h_{c}^{\prime \prime }\left. =\right.
1-h_{c}^{\prime },...,h_{N}^{\prime }|\mathbf{h}\,\rangle ,
\end{align}%
where we have used the simple identity:%
\begin{equation}
\langle \,\,\mathbf{h}^{\prime }|\frac{\mathcal{T}(\xi _{c}^{\left(
h_{c}^{\prime }\right) })}{k_{a}\mathsf{A}_{\bar{\zeta}_{+},\bar{\zeta}%
_{-}}(\eta /2-\xi _{a})}=\frac{\mathsf{A}_{\bar{\zeta}_{+},\bar{\zeta}%
_{-}}(\eta /2-\xi _{c})}{k_{a}\mathsf{A}_{\bar{\zeta}_{+},\bar{\zeta}%
_{-}}(\eta /2-\xi _{a})}\langle \,h_{1}^{\prime },...,h_{c}^{\prime \prime
}\left. =\right. h_{c}^{\prime }-1,...,h_{N}^{\prime }|.
\end{equation}%
Then, by the induction we get:%
\begin{equation}
\langle \,\mathbf{h}^{\prime }|\mathbf{\bar{h}}\,\rangle =0
\end{equation}%
being%
\begin{equation}
\langle \,\mathbf{h}^{\prime }|\mathbf{h}\,\rangle =0\text{ \ and \ }\langle
\,h_{1}^{\prime },...,h_{c}^{\prime \prime }\left. =\right. 1-h_{c}^{\prime
},...,h_{N}^{\prime }|\mathbf{h}\,\rangle =0\text{ \ }\forall c\neq a\text{
\ }
\end{equation}%
as by definition $h_{a}^{\prime }=1$ and $h_{a}=0$ and%
\begin{equation}
\langle \,h_{1}^{\prime },...,h_{a}^{\prime \prime }\left. =\right.
0,...,h_{N}^{\prime }|\mathbf{h}\,\rangle =0,
\end{equation}%
being by definition $h_{b}^{\prime }=1-h_{b}$. Let us finally compute the
last coupling:%
\begin{equation}
\langle \,\mathbf{\bar{h}}|\mathbf{\bar{h}}\,\rangle =\langle \mathbf{\bar{h}%
}|\frac{\mathcal{T}(\xi _{a}+\eta /2)}{k_{a}\mathsf{A}_{\bar{\zeta}_{+},\bar{%
\zeta}_{-}}(\eta /2-\xi _{a})}|\mathbf{h}\rangle
\end{equation}%
using once again the interpolation formula for $\mathcal{T}(\xi _{a}-\eta
/2) $ for $\mathbf{\bar{h}}$ we get:%
\begin{equation}
\langle \,\mathbf{\bar{h}}|\mathbf{\bar{h}}\,\rangle =\frac{r_{a,\mathbf{%
\bar{h}}}(\xi _{a}^{\left( 0\right) })}{k_{a}}\langle \mathbf{h}|\mathbf{h}%
\rangle ,
\end{equation}%
as all the others contributions are zero as one can prove following the same
steps described above. From which the formula for the normalization follows.
\end{proof}

\subsubsection{Transfer matrix spectrum and their isospectrality}

The previous proposition on the SoV bases directly implies the following
complete characterization of the transfer matrix spectrum which represents a
completion from the wave-functions to the eigenstates of the Theorem 2.2 of 
\cite{MaiN19}.

\begin{theorem}
Let the inhomogeneity parameters be generic \eqref{cond-inh} and let the
boundary matrices $K_{-}(\lambda )$ and $K_{+}(\lambda )$ not be both
proportional to the identity, then, for almost any choice of the boundary
parameters, the eigenvalue spectrum $\Sigma _{\mathcal{T}}$ of $\mathcal{T}%
(\lambda )$ is simple and it coincides with the set of polynomials%
\begin{eqnarray}
t(\lambda ) &=&t_{N+1}\left( \lambda ^{2}-(\eta /2)^{2}\right)
\prod_{b=1}^{N}(\lambda ^{2}-\xi _{b}^{2})+2(-1)^{N}\mathrm{det}%
_{q}M(0)\prod_{b=1}^{N}\frac{\lambda ^{2}-\xi _{b}^{2}}{(\eta /2)^{2}-\xi
_{b}^{2}}  \notag \\
&&+\sum_{a=1}^{N}\frac{4\lambda ^{2}-\eta ^{2}}{4\xi _{a}^{2}-\eta ^{2}}%
\,\prod_{\substack{ b=1  \\ b\neq a}}^{N}\frac{\lambda ^{2}-\xi _{b}^{2}}{%
\xi _{a}^{2}-\xi _{b}^{2}}\,t(\xi _{a}),  \label{form-t}
\end{eqnarray}%
satisfying the following discrete system of equations 
\begin{equation}
\det \left( 
\begin{array}{cc}
t(\xi _{n}^{(0)}) & -\mathsf{A}_{{\bar{\zeta}}_{+},{\bar{\zeta}}_{-}}(\xi
_{n}^{(0)}) \\ 
-\mathsf{A}_{{\bar{\zeta}}_{+},{\bar{\zeta}}_{-}}(-\xi _{n}^{(1)}) & t(\xi
_{n}^{(1)})%
\end{array}%
\right) =0,\quad \forall n\in \{1,\ldots ,N\}.  \label{ARXFI-Functional-eq}
\end{equation}%
The following vector and co-vectors 
\begin{align}
& |\,\Psi _{t}\,\rangle =\sum_{\mathbf{h}\in
\{0,1\}^{N}}\prod_{n=1}^{N}Q_{t}(\xi _{n}^{(h_{n})})\ \widehat{V}(\xi
_{1}^{(h_{1})},\ldots ,\xi _{N}^{(h_{N})})\,\,|\,\mathbf{h}\,\rangle ,
\label{eigenT-right} \\
& \langle \,\Psi _{t}\,|=\sum_{\mathbf{h}\in \{0,1\}^{N}}\prod_{n=1}^{N}%
\left[ \left( \frac{\xi _{n}-\eta }{\xi _{n}+\eta }\frac{\mathsf{A}_{\bar{%
\zeta}_{+},\bar{\zeta}_{-}}(\xi _{n}^{(0)})}{\mathsf{A}_{\bar{\zeta}_{+},%
\bar{\zeta}_{-}}(-\xi _{n}^{(1)})}\right) ^{\!h_{n}}Q_{t}(\xi _{n}^{(h_{n})})%
\right] \widehat{V}(\xi _{1}^{(h_{1})},\ldots ,\xi _{N}^{(h_{N})})\,\langle
\,\mathbf{h}|\,,  \label{eigenT-left}
\end{align}%
generate respectively the one-dimensional right and left $\mathcal{T}%
(\lambda )$-eigenspaces associated with the eigenvalue $t(\lambda )\in
\Sigma _{\mathcal{T}}$, where the $Q_{t}$ is defined on the discrete set of
values $\xi _{n}^{(h_{n})},\ n\in \{1,\ldots ,N\},\ h_{n}\in \{0,1\}$ by%
\begin{equation}
\frac{Q_{t}(\xi _{n}^{(1)})}{Q_{t}(\xi _{n}^{(0)})}=\frac{t(\xi _{n}^{(0)})}{%
\mathsf{A}_{\bar{\zeta}_{+},\bar{\zeta}_{-}}(\xi _{n}^{(0)})}=\frac{\mathsf{A%
}_{\bar{\zeta}_{+},\bar{\zeta}_{-}}(-\xi _{n}^{(1)})}{t(\xi _{n}^{(1)})}%
,\qquad n=1,\ldots ,N.  \label{Q-dis}
\end{equation}
\end{theorem}

\begin{proof}
This theorem follows from Theorem 2.2 of \cite{MaiN19} and the decomposition
of the identity induced from the previous proposition.
\end{proof}

In the following we will use also the following notations: 
\begin{equation}
g_{n}\equiv g_{\bar{\zeta}_{+},\bar{\zeta}_{-}}(\xi _{n})=\frac{(\xi _{n}+%
\bar{\zeta}_{+})(\xi _{n}+\bar{\zeta}_{-})}{(\xi _{n}-\bar{\zeta}_{+})(\xi
_{n}-\bar{\zeta}_{-})},
\end{equation}%
and 
\begin{equation}
f_{n}\equiv f(\xi _{n},\{\xi \})=-\prod_{\substack{ a=1  \\ a\neq n}}^{N}%
\frac{(\xi _{n}-\xi _{a}+\eta )(\xi _{n}+\xi _{a}+\eta )}{(\xi _{n}-\xi
_{a}-\eta )(\xi _{n}+\xi _{a}-\eta )},  \label{fn}
\end{equation}%
from which 
\begin{equation}
\frac{\xi _{n}-\eta }{\xi _{n}+\eta }\frac{\mathsf{A}_{\bar{\zeta}_{+},\bar{%
\zeta}_{-}}(\xi _{n}^{(0)})}{\mathsf{A}_{\bar{\zeta}_{+},\bar{\zeta}%
_{-}}(-\xi _{n}^{(1)})}=f_{n}\,g_{n}.
\end{equation}
The previous theorem allows us to state the following corollary on the
isospectrality of transfer matrices associated to different boundary
conditions:

\begin{cor}
\label{Iso-Cor}Let the inhomogeneity parameters be generic \eqref{cond-inh}
let us consider two different sets of boundary parameters:%
\begin{equation}
(\zeta _{1,\pm },\kappa _{1,\pm },\tau _{1,\pm })\neq (\zeta _{2,\pm
},\kappa _{2,\pm },\tau _{2,\pm })
\end{equation}%
with for both of them the associated boundary matrices are not both
proportianl to the identity. Then, if the following conditions holds: 
\begin{align}
\frac{1+4\kappa _{1,+}\kappa _{1,-}\cosh (\tau _{1,+}-\tau _{1,-})}{\zeta
_{1,+}\zeta _{1,-}}& =\frac{1+4\kappa _{2,+}\kappa _{2,-}\cosh (\tau
_{2,+}-\tau _{2,-})}{\zeta _{2,+}\zeta _{2,-}},  \label{Iso-1} \\
\bar{\zeta}_{1,\pm }^{2}& =\bar{\zeta}_{2,\pm \epsilon }^{2}\text{ \ for a
given }\epsilon =\{+,-\},  \label{Iso-2}
\end{align}%
the associated two transfer matrices are isospectral, i.e. there exists an
invertible $\Gamma _{12}\in \text{End}\mathcal{H}$ such that: 
\begin{equation}
\mathcal{T}(\lambda |\zeta _{1,\pm },\kappa _{1,\pm },\tau _{1,\pm })=\Gamma
_{12}^{-1}\mathcal{T}(\lambda |\zeta _{2,\pm },\kappa _{2,\pm },\tau _{2,\pm
})\Gamma _{12}\text{.}
\end{equation}%
Moreover, taken the set of boundary parameters $(\zeta _{\pm },\kappa _{\pm
},\tau _{\pm })$ satisfying the condition: 
\begin{equation}
\omega _{\epsilon }(\kappa _{\pm },\tau _{\pm })\equiv 1+4\kappa _{+}\kappa
_{-}\cosh (\tau _{+}-\tau _{-})-\epsilon \sqrt{(1+4\kappa
_{+}^{2})(1+4\kappa _{-}^{2})}=0,  \label{Iso-parallel}
\end{equation}%
\ for a given $\epsilon =\{+,-\},$ then the following isospectrality holds:%
\begin{equation}
\mathcal{T}(\lambda |\zeta _{\pm },\kappa _{\pm },\tau _{\pm })=\Gamma ^{-1}%
\mathcal{T}(\lambda |\bar{\zeta}_{\pm },0,0)\Gamma \text{,}
\label{Non-Diag-Diag-iso}
\end{equation}%
where $\mathcal{T}(\lambda |\bar{\zeta}_{\pm },0,0)$ is the transfer matrix
associated to diagonal boundary matrices with parameters $\bar{\zeta}_{\pm }$%
.
\end{cor}

\begin{proof}
Here, we have just to remark that the SoV characterization of the transfer
matrix spectrum of the previous theorem implies that the spectrum depends by
the boundary parameters only by $\bar{\zeta}_{\pm }^{2}$ and by $t_{N+1}$,
so that the identities $\left( \ref{Iso-1}\right) $-$\left( \ref{Iso-2}%
\right) $ assure that two different sets of boundary parameters $(\zeta
_{1,\pm },\kappa _{1,\pm },\tau _{1,\pm })$ and $(\zeta _{2,\pm },\kappa
_{2,\pm },\tau _{2,\pm })$ share the same values of $\bar{\zeta}_{\pm }^{2}$
and by $t_{N+1}$. The completeness of the transfer matrix spectrum
description implies then the isospectrality. Finally, taking diagonal
boundary matrices with parameters $\bar{\zeta}_{\pm }$, by definition it
holds:%
\begin{equation}
t_{N+1}(\bar{\zeta}_{\pm },0,0)=\frac{2}{\bar{\zeta}_{+}\bar{\zeta}_{-}}%
=t_{N+1}(\zeta _{\pm },\kappa _{\pm },\tau _{\pm }),
\end{equation}%
under the condition $\left( \ref{Iso-parallel}\right) $ and the choice $%
\epsilon =\epsilon _{-}\epsilon _{+}$, where $\epsilon _{\pm }$ are the
signs that we chose in $\left( \ref{Zi-bar-Def}\right) $, so that the
isospectrality statement follows.
\end{proof}

It is important to stress that this isospectrality goes far beyond that
associated with the $GL(2)$ symmetry of the model, which just implies the
isospectrality in the case one can go from the set of boundary matrices
associated to $(\zeta _{1,\pm },\kappa _{1,\pm },\tau _{1,\pm })$ to the set
associated to $(\zeta _{2,\pm },\kappa _{2,\pm },\tau _{2,\pm })$ by a
similarity transformation. One simple example is the isospectrality of the
transfer matrix associated to $(\zeta _{-},\kappa _{-},\tau _{-})$ in the
site $1$ and to $(\zeta _{+},\kappa _{+},\tau _{+})$ in the site $N$ and the
transfer matrix associated to $(\zeta _{+},\kappa _{+},\tau _{+})$ in the
site $1$ and to $(\zeta _{-},\kappa _{-},\tau _{-})$ in the site $N$. These
two transfer matrices are isospectral but for general values of the boundary
parameters one cannot pass from one set of boundary matrices to the other by
a similarity transformation.

The main example of interest for us now of this beyond $GL(2)$
isospectrality is the case of unparallel boundary magnetic fields which
satisfy the condition $\left( \ref{Iso-parallel}\right) $. The unparallel
boundary magnetic fields case is equivalent to ask that the boundary
matrices are non simultaneously diagonalizable, we will see that this
condition is compatible with $\left( \ref{Iso-parallel}\right) $. So that we
establish the isospectrality of these transfer matrices with those with
parallel magnetic fields according to $\left( \ref{Non-Diag-Diag-iso}\right) 
$, which we will use in the computation of the correlation functions. Note
that being the parallel boundary magnetic fields case equivalent to ask that
the two boundary matrices are simultaneously diagonalizable, then our
statement of beyond $GL(2)$ isospectrality follows as well as the fact that
in general the similarity $\Gamma $ in $\left( \ref{Non-Diag-Diag-iso}%
\right) $ is not easy to derive and it is not of tensor product type.

\subsection{Functional $TQ$-equation spectrum characterization}

The transfer matrix eigenvalues and eigenstates are characterized in the SoV
framework in terms of the $Q_{t}$ defined on the discrete set of shifted
inhomogeneity parameters only. It is possible to show that this discrete
characterization can be reformulated by a functional equation for a $Q_{t}$
function defined on the whole complex plane $\mathbb{C}$.

In the case of the open XXZ spin-1/2 chain with generic integrable boundary
conditions, this was first proven in \cite{KitMN14} under the quite general
boundary conditions allowing for the introduction of the Sklyanin's like SoV
approach. Here, we generalized the results derived for the XXX case in \cite%
{KitMNT17} thanks to the wider SoV approach derived in \cite{MaiN19} and
further detailed in the previous section.

\begin{theorem}
Let the inhomogeneity parameters be generic \eqref{cond-inh} and let the
boundary matrices $K_{-}(\lambda )$ and $K_{+}(\lambda )$ not be both
proportional to the identity, then, for almost any choice of the boundary
parameters, defined 
\begin{equation}
F(\lambda )=\frac{2\omega _{\epsilon }(\kappa _{\pm },\tau _{\pm })}{\bar{%
\zeta}_{-}\bar{\zeta}_{+}}\left( \lambda ^{2}-\left( \eta /2\right)
^{2}\right) \,\prod_{b=1}^{N}\prod_{h=0}^{1}\left( \lambda ^{2}-(\xi
_{b}^{(h)})^{2}\right) ,  \label{DEF-F}
\end{equation}%
with $\epsilon =\epsilon _{-}\epsilon _{+}$ and $\epsilon _{\pm }$ the signs
in $\left( \ref{Zi-bar-Def}\right) $, $t(\lambda )\in \Sigma _{\mathcal{T}}$
iff there exists a unique polynomial $Q_{t}(\lambda )$ of the form%
\begin{equation}
Q_{t}(\lambda )=\prod_{b=1}^{q}\left( \lambda ^{2}-\lambda _{b}^{2}\right)
,\qquad \lambda _{1},\ldots ,\lambda _{q}\in \mathbb{C}\setminus \big\{\pm
\xi _{1}^{(0)},\ldots ,\pm \xi _{N}^{(0)}\big\},  \label{Q-form}
\end{equation}%
satisfying 
\begin{equation}
t(\lambda )\,Q_{t}(\lambda )=\mathsf{A}_{\bar{\zeta}_{+},\bar{\zeta}%
_{-}}(\lambda )\,Q_{t}(\lambda -\eta )+\mathsf{A}_{\bar{\zeta}_{+},\bar{\zeta%
}_{-}}(-\lambda )\,Q_{t}(\lambda +\eta )+F(\lambda ),  \label{Inhom-BAX-eq}
\end{equation}%
or equivalently $t(\lambda )\in \Sigma _{\mathcal{T}}$ iff there exists a
unique polynomial $P_{t}(\lambda )$ of the form 
\begin{equation}
P_{t}(\lambda )=\prod_{b=1}^{p}\left( \lambda ^{2}-\mu _{b}^{2}\right)
,\qquad \mu _{1},\ldots ,\mu _{p}\in \mathbb{C}\setminus \big\{\pm \xi
_{1}^{(0)},\ldots ,\pm \xi _{N}^{(0)}\big\},  \label{P-form}
\end{equation}%
such that 
\begin{equation}
t(\lambda )\,P_{t}(\lambda )=\mathsf{A}_{-\bar{\zeta}_{+},-\bar{\zeta}%
_{-}}(\lambda )\,P_{t}(\lambda -\eta )+\mathsf{A}_{-\bar{\zeta}_{+},-\bar{%
\zeta}_{-}}(-\lambda )\,P_{t}(\lambda +\eta )+F(\lambda ),
\label{Inhom-BAX-eq-bis}
\end{equation}%
where:%
\begin{equation}
p=q=N\text{ \ \ if \ }\omega _{\epsilon }(\kappa _{\pm },\tau _{\pm })\neq 0,%
\text{ \ \ }p+q=N\text{ \ \ if \ }\omega _{\epsilon }(\kappa _{\pm },\tau
_{\pm })=0.\text{ \ \ }
\end{equation}
\end{theorem}

\begin{proof}[Proof]
Starting from the SoV discrete characterization derived in the previous
section, we can prove the current theorem just as done in \cite{KitMNT17},
for the Theorem 3.2, 3.3 and Proposition 3.1.
\end{proof}

It is important to remark that differently from the case of the Sklyanin's
like SoV characterization\footnote{%
See for example Theorem 3.1 of our paper \cite{KitMNT17} for the open XXX
spin chain.}, here, on the one hand, we are doing a characterization holding
in general both for parallel and unparallel boundary magnetic fields and, on
the other hand, we do not need to implement any similarity transformation to
bring the system in an appropriate form to use Sklyanin's like SoV
characterization. These remarks are very important for the case of the XXZ
and XYZ spin chains where these similarity transformations, required to make
Sklyanin's like SoV applicable, are not of simple tensor product form but
have a non-local form intrinsic of the Baxter's gauge transformations.

\subsection{Bethe ansatz form of separate state and boundary-bulk
decomposition}

Here, we recall the rewriting in Bethe ansatz form of the SoV
characterization of the transfer matrix eigenstates and separate states of
the open XXX quantum chain following \cite{KitMNT17}. Then, we present the
main result of the section, i.e. the boundary-bulk decomposition of these
separate states.

\subsubsection{Similarity transformation to triangular cases}

\label{Similarity-def-Sec}

We can define the following similarity transformed boundary monodromy
matrices:%
\begin{equation}
\mathcal{\bar{U}}_{\mp }(\lambda )=W_{0}\,\Gamma _{W}\,\mathcal{U}_{\mp
}(\lambda )\,\Gamma _{W}^{-1}\,W_{0}^{-1}=\left( 
\begin{array}{cc}
\mathcal{\bar{A}}_{\mp }(\lambda ) & \mathcal{\bar{B}}_{\mp }(\lambda ) \\ 
\mathcal{\bar{C}}_{\mp }(\lambda ) & \mathcal{\bar{D}}_{\mp }(\lambda )%
\end{array}%
\right) ,
\end{equation}%
where $W_{0}\in GL(2,\mathbb{C})$ acts on the auxiliary space, whereas $%
\Gamma _{W}\equiv \otimes _{n=1}^{N}W_{n}$ acts on the quantum space of
states. By the $GL(2,\mathbb{C})$ invariance of the R-matrices: 
\begin{equation}
R_{12}(\lambda )\,W_{1}\,W_{2}=W_{2}\,W_{1}\,R_{12}(\lambda ),
\end{equation}%
we have that it holds: 
\begin{align}
\mathcal{\bar{U}}_{-}(\lambda )& =M(\lambda )\,\bar{K}_{-}(\lambda )\,\hat{M}%
(\lambda ), \\
\mathcal{\bar{U}}_{+}^{t_{0}}(\lambda )& =M^{t_{0}}(\lambda )\,\bar{K}%
_{+}^{t_{0}}(\lambda )\,\hat{M}^{t_{0}}(\lambda ),
\end{align}%
where we have defined:%
\begin{equation}
\bar{K}_{\mp }(\lambda )=W_{0}\,K_{\mp }(\lambda )\,W_{0}^{-1}=\left( 
\begin{array}{cc}
\bar{a}_{\mp }(\lambda ) & \bar{b}_{\mp }(\lambda ) \\ 
\bar{c}_{\mp }(\lambda ) & \bar{d}_{\mp }(\lambda )%
\end{array}%
\right) .  \label{K-XXX-Simil}
\end{equation}%
The following similarity transformation holds:%
\begin{equation}
\mathcal{\bar{T}}(\lambda )=\Gamma _{W}\,\mathcal{T}(\lambda )\,\Gamma
_{W}^{-1}
\end{equation}%
where: 
\begin{equation}
\mathcal{\bar{T}}(\lambda )=\text{tr}_{0}\left\{ \bar{K}_{+}(\lambda
)\,M(\lambda )\,\bar{K}_{-}(\lambda )\,\hat{M}(\lambda )\right\} .
\end{equation}%
Then, under the following choice:%
\begin{equation}
W\equiv W_{\epsilon _{+},\epsilon _{-}}=\left( 
\begin{array}{cc}
1 & \frac{-1+\epsilon _{-}\sqrt{1+4\kappa _{-}^{2}}}{2\kappa _{-}e^{-\tau
_{-}}} \\ 
\frac{1-\epsilon _{+}\sqrt{1+4\kappa _{+}^{2}}}{2\kappa _{+}e^{\tau _{+}}} & 
1%
\end{array}%
\right) ,  \label{W}
\end{equation}%
for a given $(\epsilon _{+},\epsilon _{-})\in \{-1,1\}^{2}$, the boundary
matrices take the triangular form 
\begin{equation}
\bar{K}_{-}(\lambda )=\mathrm{I}+\frac{\lambda +\eta /2}{\bar{\zeta}_{-}}%
(\sigma ^{z}+\bar{\mathsf{c}}_{-}\sigma ^{-}),\text{ }\bar{K}_{+}(\lambda )=%
\mathrm{I}+\frac{\lambda -\eta /2}{\bar{\zeta}_{+}}(\sigma ^{z}+\bar{\mathsf{%
b}}_{+}\sigma ^{+}),
\end{equation}%
with 
\begin{align}
& \bar{\mathsf{c}}_{-}=\,\frac{2\epsilon _{-}\kappa _{-}e^{-\tau _{-}}}{%
\sqrt{1+4\kappa _{-}^{2}}}\left[ 1+\frac{(1+\epsilon _{-}\sqrt{1+4\kappa
_{-}^{2}})(1-\epsilon _{+}\sqrt{1+4\kappa _{+}^{2}})}{4\kappa _{+}\kappa
_{-}e^{\tau _{+}-\tau _{-}}}\right] , \\
& \bar{\mathsf{b}}_{+}=\,\frac{2\epsilon _{+}\kappa _{+}e^{\tau _{+}}}{\sqrt{%
1+4\kappa _{+}^{2}}}\left[ 1+\frac{(1-\epsilon _{-}\sqrt{1+4\kappa _{-}^{2}}%
)(1+\epsilon _{+}\sqrt{1+4\kappa _{+}^{2}})}{4\kappa _{+}\kappa _{-}e^{\tau
_{+}-\tau _{-}}}\right] .
\end{align}%
The interest in the above similarity transformation is that for $\bar{%
\mathsf{b}}_{+}\neq 0$ we can implement the Sklyanin's like SoV approach
using as generator of the left and right SoV bases the monodromy matrix
entry $\mathcal{\bar{B}}_{+}(\lambda )$ (diagonalizable and with simple
spectrum) to solve the spectral problem of the transfer matrix $\mathcal{%
\bar{T}}(\lambda )$. Then, the eigenstates ($|\,t\rangle $,$\langle \,t|$)
of the original transfer matrix $\mathcal{T}(\lambda )$ can be therefore
expressed in terms of those of the new triangular one ($|\bar{t}\rangle $,$%
\langle \bar{t}\,|$) using the tensor product similarity transformation $%
\Gamma _{W}$, i.e. it holds: 
\begin{equation}
|t\rangle =\Gamma _{W}^{-1}\,|\bar{t}\rangle ,\qquad \langle t|=\langle \bar{%
t}|\,\Gamma _{W}.
\end{equation}%
Finally, let us remark that the Hamiltonian%
\begin{equation}
\bar{H}=\Gamma _{W}\,H\Gamma _{W}^{-1},
\end{equation}%
associated to the transfer matrix $\mathcal{\bar{T}}(\lambda )$ reads:%
\begin{align}
\bar{H} =\sum_{i=1}^{N-1}\Big[\sigma _{i}^{x}\sigma _{i+1}^{x}+\sigma
_{i}^{y}\sigma _{i+1}^{y}+\sigma _{i}^{z}\sigma _{i+1}^{z}\Big]+\frac{\eta }{%
\bar{\zeta}_{-}}\Big[\sigma _{1}^{z}+\bar{\mathsf{c}}_{-}\sigma _{1}^{-}\Big]%
+\frac{\eta }{\bar{\zeta}_{+}}\Big[\sigma _{N}^{z}+\bar{\mathsf{b}}%
_{+}\sigma _{N}^{+}\Big].
\end{align}

\subsubsection{Unparallel cases isospectral to the parallel ones}

Here, we characterize the cases in which the original transfer matrix $%
\mathcal{T}(\lambda )$, associated to unparallel boundary magnetic fields,
is isospectral to the transfer matrix $\mathcal{\hat{T}}(\lambda )$ $\equiv 
\mathcal{T}(\lambda |\bar{\zeta}_{\pm },0,0)$, associated to parallel along
the z-direction boundary magnetic fields, i.e. 
\begin{equation}
\mathcal{\hat{T}}(\lambda )=\text{tr}_{0}\left\{ \hat{K}_{+}(\lambda
)\,M(\lambda )\,\hat{K}_{-}(\lambda )\,\hat{M}(\lambda )\right\} ,
\label{Def-T-Hat}
\end{equation}%
with%
\begin{equation}
\hat{K}_{-}(\lambda )=\mathrm{I}+\frac{\lambda +\eta /2}{\bar{\zeta}_{-}}%
\sigma ^{z},\text{ }\hat{K}_{+}(\lambda )=\mathrm{I}+\frac{\lambda -\eta /2}{%
\bar{\zeta}_{+}}\sigma ^{z}.  \label{Def-Both-Diag-K}
\end{equation}%
Let us recall that by construction of the similarity transformation, $%
\mathcal{T}(\lambda )$ and $\mathcal{\bar{T}}(\lambda )$ are clearly
isospectral. Then, let us now impose the following boundary conditions:%
\begin{equation}
\bar{\mathsf{c}}_{-}=0,\text{ \ }\bar{\mathsf{b}}_{+}\neq 0,
\label{Triang-Diag}
\end{equation}%
which keep $\mathcal{\bar{T}}(\lambda )$ associated to a properly triangular
boundary matrix in site $N$, i.e. $\mathcal{T}(\lambda )$ is properly
associated to unparallel boundary magnetic fields. Then, the following
identity: 
\begin{equation}
\bar{\mathsf{c}}_{-}\bar{\mathsf{b}}_{+}=\,\frac{\epsilon _{+}\epsilon
_{-}e^{\tau _{-}-\tau _{+}}\omega _{\epsilon _{+}\epsilon _{-}}(\kappa _{\pm
},\tau _{\pm })}{4\kappa _{+}\kappa _{-}\sqrt{(1+4\kappa _{+}^{2})(1+4\kappa
_{-}^{2})}},
\end{equation}%
implies that the condition $\left( \ref{Triang-Diag}\right) $ is equivalent
to the following one%
\begin{equation}
\omega _{\epsilon _{+}\epsilon _{-}}(\kappa _{\pm },\tau _{\pm })=0,\text{ \ 
}\bar{\mathsf{b}}_{+}\neq 0,  \label{Non-par-Iso-par}
\end{equation}%
which by Corollary \ref{Iso-Cor} implies the isospectrality of our original
transfer matrix $\mathcal{T}(\lambda )$ (with unparallel boundary magnetic
fields) to $\mathcal{\hat{T}}(\lambda )$ (with parallel boundary magnetic
fields). We can state the following:

\begin{lemma}
\label{Lem-iso-cond}Let us fix a couple $(\epsilon _{+},\epsilon _{-})\in
\{-1,1\}^{2}$ and let us impose the following boundary conditions:%
\begin{equation}
e^{\tau _{+}}=e^{\tau _{-}}\frac{(\epsilon _{-}\sqrt{1+4\kappa _{-}^{2}}%
+1)(\epsilon _{+}\sqrt{1+4\kappa _{+}^{2}}-1)}{4\kappa _{+}\kappa _{-}},
\end{equation}%
then for any choice of the boundary parameters such that:%
\begin{equation}
\kappa _{+}\neq \pm \kappa _{-},
\end{equation}%
our original transfer matrix $\mathcal{T}(\lambda )$ has unparallel boundary
magnetic fields and it is isospectral to $\mathcal{\hat{T}}(\lambda )$ with
parallel z-oriented boundary magnetic fields. Moreover, taken $i\eta \in 
\mathbb{R}$, then for any choice of the boundary parameters $(\zeta _{\pm
},\kappa _{\pm })$ such that:%
\begin{equation}
i\bar{\zeta}_{\pm }\in \mathbb{R}\text{,}
\end{equation}%
the transfer matrix $\mathcal{\hat{T}}(\lambda )$ is self-adjoint. So, the
ground state distribution of our original Hamiltonian $H$ (associated to $%
\mathcal{T}(\lambda )$) coincides with that of the Hamiltonian%
\begin{equation}
\hat{H}=\sum_{i=1}^{N-1}\left[ \sigma _{i}^{x}\sigma _{i+1}^{x}+\sigma
_{i}^{y}\sigma _{i+1}^{y}+\sigma _{i}^{z}\sigma _{i+1}^{z}\right] +\frac{i}{%
\bar{\zeta}_{-}}\sigma _{1}^{z}+\frac{i}{\bar{\zeta}_{+}}\sigma _{N}^{z},
\end{equation}%
(associated to $\mathcal{\hat{T}}(\lambda )$) and in the thermodynamic limit 
$N\rightarrow \infty $ it is known \cite{YanY66,YanY66a,Koz18} to be characterized by the following
distribution on the positive real axis:\footnote{Here, we are restricting ourself to the case without boundary roots, they were analyzed in \cite{KapS96,GriDT19} for the case of the open XXZ chain. Anyhow, in the following, we will argue how their presence can be handled without altering the main features of our results on correlation functions.}
\begin{equation}
\rho (\lambda )=\frac{1}{\,\cosh \pi \lambda },
\end{equation}%
once we fix $\eta =-i$.
\end{lemma}

\subsubsection{The Bethe ansatz form of eigenstates and separate states}

Let us here present the Bethe ansatz form of the transfer matrix eigenstates
following from our SoV characterization of the spectrum

\begin{proposition}[\protect\cite{KitMNT17}]
Let the inhomogeneities $\xi _{1},\ldots ,\xi _{N}$ be generic %
\eqref{cond-inh} and under the condition $\bar{\mathsf{b}}_{+}\neq 0$, then,
for any $t(\lambda )\in \Sigma _{\mathcal{T}}$, the corresponding (unique up
to normalization) right and left $\mathcal{T}(\lambda )$-eigenstates reads:%
\begin{equation}
\langle \,t|=\langle \,\underline{0}\,|\prod_{a=1}^{p}\mathcal{\bar{B}}%
_{+}(\mu _{a})\,\Gamma _{W}\,\text{, \ \ }|\,t\rangle =\Gamma
_{W}^{-1}\,\prod_{a=1}^{q}\mathcal{\bar{B}}_{+}(\lambda _{a})\,|\,0\rangle ,
\label{LR-Bethe-vector}
\end{equation}%
where $\lambda _{1},\ldots ,\lambda _{q}$ are the roots of the polynomial $%
Q_{t}(\lambda )$ and $\mu _{1},\ldots ,\mu _{p}$ are those of the polynomial 
$P_{t}(\lambda )$, solutions of $\left( \ref{Inhom-BAX-eq}\right) $ and $%
\left( \ref{Inhom-BAX-eq-bis}\right) $, respectively.
\end{proposition}

\begin{proof}[Proof]
The Proposition 3.3 and 3.4 of \cite{KitMNT17} show that the SoV
characterization of the eigenvectors is equivalent to the above Bethe ansatz
like formulation.
\end{proof}

Let us observe now that under the condition $\bar{\mathsf{b}}_{+}\neq 0$,
the results of Section 2.3 of the paper \cite{MaiN19}, implies the
identities:%
\begin{equation}
\,\langle \,\mathbf{h}|=\mathsf{l\,}\langle +,\,\mathbf{h}|\Gamma _{W}\text{%
, \ \ }|\mathbf{h}\rangle =\mathsf{r\,}\Gamma _{W}^{-1}|\,\mathbf{h,}%
+\rangle ,
\end{equation}%
where $\mathsf{l\,}$ and $\mathsf{r}$\ are some computable normalization
constant and $\langle +,\,\mathbf{h}|$ and $|\,\mathbf{h,}+\rangle $ are the
left and right eigenstates of $\bar{\mathcal{B}}_{+}(\lambda )$. So that we
can equivalently use one or the other SoV basis leading to the same above
results.

In separation of variable the so-called \emph{separate states} read:%
\begin{equation}
|\,\gamma \,\rangle =\sum_{\mathbf{h}\in \{0,1\}^{N}}\prod_{n=1}^{N}\gamma
_{n}^{(h_{n})}\ \widehat{V}(\xi _{1}^{(h_{1})},\ldots ,\xi
_{N}^{(h_{N})})\,|\,\mathbf{h}\,\rangle ,  \label{separate-st}
\end{equation}%
for some coefficients $\gamma _{n}^{(h_{n})}$, $n\in \{1,\ldots ,N\}$, $%
h_{n}\in \{0,1\}$. They are a set of states containing as special elements
the transfer matrix eigenvectors and playing a fundamental role in the
computation of correlation functions. Indeed, scalar products of the
separate states universally admit \cite{GroMN12,Nic13} determinat
representations, shown in \cite{KitMNT17} for the open XXX case to be
equivalent to determinants generalizing the Slavnov's determinants \cite%
{Sla89}, results previously known only in the case of parallel boundary
magnetic fields \cite{Wan02,KitKMNST07}. Moreover, as we will show in the
next section, the action of local operators on separate states can be
efficiently rewritten in terms of linear combinations of separate states in
this way allowing to implement the calculation of correlation functions.

Here, we recall that along the same lines of the previous proposition in 
\cite{KitMNT17} we have shown that also the separate states naturally admit
Bethe ansatz representations. In the following we will use the following
representation of separate states: 
\begin{equation}
|\,\beta \,\rangle =\Gamma _{W}^{-1}\prod_{a=1}^{n_{\beta }}\mathcal{\bar{B}}%
_{+}(b_{a})\,|\,0\,\rangle ,  \label{Bethe1}
\end{equation}%
where%
\begin{equation}
\beta (\lambda )=\prod_{m=1}^{n_{\beta }}(\lambda ^{2}-b_{m}^{2}),
\label{poly-sep-state}
\end{equation}%
which coincides with the separate state $|\,\gamma \,\rangle $ under the
identification 
\begin{equation}
\gamma _{n}^{(h_{n})}=\left( -1\right) ^{N}\beta (\xi _{n}^{(h_{n})})\bar{b}_{+}(\xi _{n}^{(h_{n})}).
\end{equation}

\subsubsection{Boundary-bulk decomposition of separate states}

Here, we compute the boundary-bulk decomposition for the separate states. We
observe that it holds:%
\begin{equation}
\mathcal{\bar{B}}_{+}\left( \lambda \right) =\mathcal{\hat{B}}_{+}\left(
\lambda \right) +D\left( \lambda \right) D\left( -\lambda \right) \bar{b}%
_{+}\left( \lambda \right)
\end{equation}%
where $\mathcal{\hat{B}}_{+}\left( \lambda \right) $ is the operator
associated to the diagonal part of $\bar{K}_{+}\left( \lambda \right) $:%
\begin{equation}
\mathcal{\bar{B}}_{+}\left( \lambda \right) =B\left( \lambda \right) D\left(
-\lambda \right) \bar{a}_{+}\left( \lambda \right) -D\left( \lambda \right)
B\left( -\lambda \right) \bar{d}_{+}\left( \lambda \right) ,
\end{equation}
then the next proposition follows:

\begin{proposition}
The following boundary-bulk decomposition of separate states holds: 
\begin{align}
\prod_{j=1}^{M}\mathcal{B}_{+}(\lambda _{j})|\,0\,\rangle &
=\sum_{a=0}^{M}\sum_{\substack{ \mathcal{X}\cup \mathcal{Y}=I_{M}  \\ 
\mathcal{X}\cap \mathcal{Y}=\emptyset ,|\mathcal{Y}|=a}}\frac{\bar{b}%
_{+}\left( \lambda _{\mathcal{X}}\right) d(\lambda _{\mathcal{X}})d(-\lambda
_{\mathcal{X}})}{\lambda _{\mathcal{X}}^{2}-\lambda _{\mathcal{Y}}^{2}} 
\notag \\
& \times \sum_{\sigma _{\mathcal{Y}}=\pm }[\lambda _{\mathcal{X}%
}^{2}-(\lambda _{\mathcal{Y}}^{(\sigma )}-\eta )^{2}]H_{(\sigma _{\mathcal{Y}%
})}^{\mathcal{B}_{+}}(\lambda _{\mathcal{Y}})B(\lambda _{\mathcal{Y}%
}^{(\sigma )})|\,0\,\rangle ,
\end{align}%
where we have used the short notations $I_{M}=\{1,...,M\}$ and%
\begin{equation}
B(\lambda _{\mathcal{Y}}^{(\sigma )})=\prod_{j\in \mathcal{Y}}B(\sigma
_{j}\lambda _{j}),\text{ }\lambda _{\mathcal{X}}^{2}-\lambda _{\mathcal{Y}%
}^{2}=\prod_{i\in \mathcal{X}}\prod_{j\in \mathcal{Y}}(\lambda
_{i}^{2}-\lambda _{j}^{2}),\text{ etc}  \label{Multi-Implicit}
\end{equation}%
and 
\begin{align}
H_{(\sigma _{1},...,\sigma _{R})}^{\mathcal{B}_{+}}(\lambda _{1},...,\lambda
_{R})& =\prod_{j=1}^{R}d(-\sigma _{j}\lambda _{j})\frac{\sinh (2\lambda
_{j}+\eta )}{\sinh (2\lambda _{j})}\sinh (\lambda _{j}+\sigma _{j}(\bar{\zeta%
}_{+}-\eta /2))  \notag \\
\times & \prod_{1\leq r<s\leq R}\frac{\sinh (\sigma _{s}\lambda _{s}+\sigma
_{r}\lambda _{r}-\eta )}{\sinh (\sigma _{s}\lambda _{s}+\sigma _{r}\lambda
_{r})}\text{.}  \label{HB+}
\end{align}
\end{proposition}

\begin{proof}
In the following, we use the commutativity of the following three sets of
operator families, $\mathcal{\bar{B}}_{+}\left( \lambda \right) $, $\mathcal{%
\hat{B}}_{+}\left( \lambda \right) $ and $D\left( \lambda \right) $. We
clearly have the following identity:%
\begin{equation}
\prod_{j=1}^{M}\mathcal{\bar{B}}_{+}(\lambda _{j})|\,0\,\rangle
=\sum_{a=0}^{M}\sum_{\substack{ \mathcal{X}\cup \mathcal{Y}=I_{M}  \\ 
\mathcal{X}\cap \mathcal{Y}=\emptyset ,|\mathcal{Y}|=a}}\sum_{\sigma _{%
\mathcal{Y}}=\pm }C_{\mathcal{X}}{}_{,\mathcal{Y}}^{(\sigma )}(\lambda
)B(\lambda _{\mathcal{Y}}^{(\sigma )})|\,0\,\rangle ,
\label{first-Development}
\end{equation}%
as a consequence of the boundary-bulk decomposition of the $\mathcal{\bar{B}}%
_{+}$-operator family, the Yang-Baxter commutation relations between $%
D\left( \lambda \right) $ and $B\left( \lambda \right) $ and the following
identity:%
\begin{equation}
D\left( \lambda \right) |\,0\,\rangle =|\,0\,\rangle d(\lambda ).
\label{D-eigenV}
\end{equation}%
So we are left with the proof that the coefficients take the above defined
form. Let us fix a couple of sets $\mathcal{X}\cup \mathcal{Y}=I_{M}$ then
by using the commutativity we can take the following rewriting:%
\begin{equation}
\prod_{j=1}^{M}\mathcal{\bar{B}}_{+}(\lambda _{j})|\,0\,\rangle =\mathcal{%
\bar{B}}_{+}(\lambda _{\mathcal{X}})\mathcal{\bar{B}}_{+}(\lambda _{\mathcal{%
Y}})|\,0\,\rangle ,
\end{equation}%
then it is easy to understand that the terms%
\begin{equation}
\sum_{\sigma _{\mathcal{Y}}=\pm }C_{\mathcal{X}}{}_{,\mathcal{Y}}^{(\sigma
)}(\lambda )B(\lambda _{\mathcal{Y}}^{(\sigma )})|\,0\,\rangle
\end{equation}%
in the state $\left( \ref{first-Development}\right) $ can be generated only
by the following term%
\begin{equation}
\bar{b}_{+}\left( \lambda _{\mathcal{X}}\right) D\left( \lambda _{\mathcal{X}%
}\right) D\left( -\lambda _{\mathcal{X}}\right) \mathcal{\hat{B}}%
_{+}(\lambda _{\mathcal{Y}})|\,0\,\rangle
\end{equation}%
and only by taking the direct action of the operators $D(\pm \lambda _{%
\mathcal{X}})$ on the state: 
\begin{equation}
\mathcal{\hat{B}}_{+}(\lambda _{\mathcal{Y}})|\,0\,\rangle =\sum_{\sigma _{%
\mathcal{Y}}=\pm }H_{(\sigma _{\mathcal{Y}})}^{\mathcal{B}_{+}}(\lambda _{%
\mathcal{Y}})B(\lambda _{\mathcal{Y}}^{(\sigma )})|\,0\,\rangle ,
\end{equation}%
where the above boundary-bulk decomposition has been derived in Proposition
3.4 of our paper \cite{KitKMNST07} and it holds being $\mathcal{\hat{B}}%
_{+}\left( \lambda \right) $ associated to the diagonal part of $\bar{K}%
_{+}\left( \lambda \right) $. So all we have to do to compute the
coefficients $C_{\mathcal{X}}{}_{,\mathcal{Y}}^{(\sigma )}(\lambda )$ is to
compute this direct action on the above state, which, by the Yang-Baxter
commutation relations and the property $\left( \ref{D-eigenV}\right) $,
clearly give:%
\begin{align}
\bar{b}_{+}\left( \lambda _{\mathcal{X}}\right)\left[ D\left( \lambda _{%
\mathcal{X}}\right) D\left( -\lambda _{\mathcal{X}}\right) \right] _{\text{%
Direct action}}&\mathcal{\hat{B}}_{+}(\lambda _{\mathcal{Y}})|\,0\,\rangle =%
\bar{b}_{+}\left( \lambda _{\mathcal{X}}\right) d(\lambda _{\mathcal{X}%
})d(-\lambda _{\mathcal{X}})  \notag \\
& \times \sum_{\sigma _{\mathcal{Y}}=\pm }\frac{\lambda _{\mathcal{X}%
}^{2}-(\lambda _{\mathcal{Y}}^{(\sigma )}-\eta )^{2}}{\lambda _{\mathcal{X}%
}^{2}-\lambda _{\mathcal{Y}}^{2}}H_{(\sigma _{\mathcal{Y}})}^{\mathcal{B}%
_{+}}(\lambda _{\mathcal{Y}})B(\lambda _{\mathcal{Y}}^{(\sigma
)})|\,0\,\rangle ,
\end{align}%
in this way completing our proof.
\end{proof}

\section{Action of local operators on boundary separate states}

In order to compute correlation functions we have to be able to compute the
action of local operators on transfer matrix eigenstates and so on separate
states. The first fundamental step in these computations is the
reconstruction of local operators \cite{KitMT99} in terms of the bulk
generators of the Yang-Baxter algebra and, in particular, their simplified
version derived in \cite{KitKMNST07}. Indeed, thanks to these
reconstructions and the previously derived boundary-bulk decomposition of
separate states, these actions can be computed by using the known
Yang-Baxter algebra or SoV representations.

Here, we compute the action of the local and quasi-local operators on
separate states associated to the transfer matrix $\mathcal{\bar{T}}(\lambda
)$, i.e. the one with diagonal $\bar{K}_{-}(\lambda )$ and properly
triangular ($\bar{\mathsf{b}}_{+}\neq 0$) $\bar{K}_{+}(\lambda )$ boundary
matrices. These results then translate directly in action of local and
quasi-local operators on separate states associated to the original transfer
matrix $\mathcal{T}(\lambda )$. Indeed, the similarity transformation
relating the two transfer matrices is of pure tensor product type and so it
transforms local operators in local operators.

\subsection{Action of the quasi-local operator $Q_{m}(\protect\kappa )$}

Let us recall that the first nontrivial correlation functions for the open
quantum spin chains are the one point functions, i.e. the ground state
average of a local spin operators, which measure the correlation with the
boundaries. Following, the presentation of \cite{KitKMNST08}, we first
compute the action of the quasi-local operator:%
\begin{equation}
Q_{m}(\kappa )=\prod_{a=1}^{m}(E_{a}^{11}+\kappa
E_{a}^{22})=\prod_{a=1}^{m}(A(\xi _{a}^{\left( 0\right) })+\kappa D(\xi
_{a}^{\left( 0\right) }))\prod_{a=1}^{m}(A(\xi _{a}^{\left( 1\right)
})+D(\xi _{a}^{\left( 1\right) })),
\end{equation}%
which gives access to the one-point function of the local operator:%
\begin{equation}
\sigma _{m}^{z}=1+2\left[ \left( \partial _{\kappa }Q_{m}(\kappa )\right)
_{\kappa =1}-\left( \partial _{\kappa }Q_{m+1}(\kappa )\right) _{\kappa =1}%
\right] .
\end{equation}%
Here, we use the implicit multiplication notation introduced in $\left( \ref%
{Multi-Implicit}\right) $ and the following further notations: 
\begin{equation}
R_{n}^{\kappa }(\xi _{\gamma _{+}}|\xi _{\gamma _{-}}|\mu _{\alpha _{+}}|\mu
_{\alpha _{-}})=R(\xi _{\gamma _{+}}|\xi _{\gamma _{-}}|\mu _{\alpha
_{+}}|\mu _{\alpha _{-}})\bar S_{n}^{\kappa }(\xi _{\gamma _{+}}|\mu
_{\alpha _{+}}),
\end{equation}%
where:%
\begin{eqnarray}
&&R(\xi _{\gamma _{+}}|\xi _{\gamma _{-}}|\mu _{\alpha _{+}}|\mu _{\alpha
_{-}})\left. =\right. \frac{\tau (\mu _{\alpha _{+}}|\mu _{\alpha _{-}})\tau
(\xi _{\gamma _{-}}^{\left( 0\right) }|\xi _{\gamma _{+}}^{\left( 0\right) })%
}{\tau (\xi _{\gamma _{+}\cup \gamma _{-}}^{\left( 0\right) }|\mu _{\alpha
_{+}})f(\mu _{\alpha _{-}}|\xi _{\gamma _{+}}^{\left( 0\right) })}, \\
&&\tau (x|y)\left. =\right. a(x)f(y|x),\text{ \ }f(x|y)\left. =\right. \frac{%
x-y+\eta }{x-y},
\end{eqnarray}%
and $\bar S_{n}^{\kappa }(\xi _{\gamma _{+}}|\mu _{\alpha _{+}})$ is defined
by 
\begin{equation}
\bar S_{n}^{\kappa }(\xi _{\gamma _{+}}|\mu _{\alpha _{+}})=\frac{\xi
_{\gamma _{+}}^{\left( 0\right) }-\mu _{\alpha _{+}}+\eta }{\prod\limits
_{\substack{ a>b  \\ a,b\in \gamma _{+}}}\xi _{ab}\prod\limits_{\substack{ %
a>b  \\ a,b\in \alpha _{+}}}\mu _{ba}}\text{ \ \ }\underset{k\in \gamma
_{+},j\in \alpha _{+}}{\text{det}}\bar M_{jk}^{\kappa }\text{ ,}
\end{equation}%
with the $n\times n$ matrix $M^{\kappa }$ defined by: 
\begin{equation}
\bar M_{jk}^{\kappa }=r(\xi _{k}^{(0)}|\mu _{j})-\kappa r(\mu _{j}|\xi
_{k}^{(0)})\frac{f(\mu _{\alpha _{+}-\{j\}}|\mu _{j})}{f(\mu _{j}|\mu
_{\alpha _{+}-\{j\}})}\frac{f(\mu _{j}|\xi _{\gamma _{+}}^{(0)})}{f(\xi
_{\gamma _{+}}^{(0)}|\mu _{j})},
\end{equation}%
and%
\begin{equation}
r(x|y)=\frac{\eta }{(x-y)(x-y+\eta )}.  \label{r-instead-t}
\end{equation}%
Then, the following proposition holds:

\begin{proposition}
\label{Action-Q_m}The action of the following local operators on boundary
separate states read: 
\begin{equation}
Q_{m}^{\kappa }\mathcal{\bar{B}}_{+}(\mu _{I_{M}})|\,0\,\rangle
=\sum_{n=0}^{l(m,M)}\sum_{\substack{ I_{M}=\alpha _{+}\cup \alpha
_{-},\alpha _{+}\cap \alpha _{-}=\emptyset  \\ I_{m}=\gamma _{+}\cup \gamma
_{-},\gamma _{+}\cap \gamma _{-}=\emptyset  \\ |\gamma _{+}|=|\alpha _{+}|=n 
}}\mathcal{R}_{n}^{\kappa }(\xi _{\gamma _{+}}|\xi _{\gamma _{-}}|\mu
_{\alpha _{+}}|\mu _{\alpha _{-}})\mathcal{\bar{B}}_{+}(\mu _{\alpha
_{-}}\cup \xi _{\gamma _{+}}^{(0)})|\,0\,\rangle ,
\end{equation}%
where we have defined $l(m,R)=\min (m,R)$,%
\begin{equation}
\mathcal{R}_{n}^{\kappa }(\xi _{\gamma _{+}}|\xi _{\gamma _{-}}|\mu _{\alpha
_{+}}|\mu _{\alpha _{-}})=\sum_{\sigma _{\alpha _{+}}=\pm }\frac{H_{\sigma
_{\alpha _{+}}}^{\mathcal{B}_{+}}(\mu _{\alpha _{+}})}{H^{\mathcal{B}%
_{+}}(\xi _{\gamma _{+}}^{\left( 0\right) })}R_{n}^{\kappa }(\xi _{\gamma
_{+}}|\xi _{\gamma _{-}}|\mu _{\alpha _{+}}^{\sigma }|\pm \mu _{\alpha
_{-}}),
\end{equation}%
and $H^{\mathcal{B}_{+}}(\xi _{\gamma _{+}}^{\left( 0\right) })$ stays for
the coefficient\footnote{%
Note that these are the only nonzero coefficients in such arguments.} (\ref%
{HB+}) with $\sigma _{\gamma _{+}}=(1,..,1)$ while we have defined $\mu
_{\alpha _{+}}^{\sigma }\equiv \{\sigma _{i}\mu _{i}$ $\forall i\in \alpha
_{+}\}$ and $\pm \mu _{\alpha _{-}}\equiv \mu _{\alpha _{-}}\cup (-\mu
_{\alpha _{-}})$.
\end{proposition}

\begin{proof}
We use first the boundary-bulk decomposition of the separate states given in
the previous section, so that it holds:%
\begin{align}
Q_{m}^{\kappa }\mathcal{\bar{B}}_{+}(\mu _{I_{M}})|\,0\,\rangle &
=\sum_{a=0}^{M}\sum_{\substack{ \mathcal{X}\cup \mathcal{Y}=I_{M} \notag  \\ 
\mathcal{X}\cap \mathcal{Y}=\emptyset ,|\mathcal{Y}|=a}}\frac{\bar{b}%
_{+}\left( \mu _{\mathcal{X}}\right) d(\mu _{\mathcal{X}})d(-\mu _{\mathcal{X%
}})}{\mu _{\mathcal{X}}^{2}-\mu _{\mathcal{Y}}^{2}}\sum_{\sigma _{\mathcal{Y}%
}=\pm }[\mu _{\mathcal{X}}^{2}-(\mu _{\mathcal{Y}}^{(\sigma )}-\eta )^{2}] \\
& \times H_{(\sigma _{\mathcal{Y}})}^{\mathcal{B}_{+}}(\mu _{\mathcal{Y}%
})\left( Q_{m}^{\kappa }B(\mu _{\mathcal{Y}}^{(\sigma )})|\,0\,\rangle
\right) ,
\end{align}%
Now, we use the known action of the operator $Q_{m}^{\kappa }$ on a generic
bulk state, as derived in Proposition 4.1 of \cite{KitKMNST08}, to write: 
\begin{equation}
Q_{m}^{\kappa }B(\mu _{\mathcal{Y}}^{(\sigma )})|\,0\,\rangle
=\sum_{n=0}^{l(m,a)}\sum_{\substack{ \mathcal{Y}=\mathcal{Y}_{+}\cup 
\mathcal{Y}_{-}  \\ \xi _{I_{m}}=\xi _{\gamma _{+}}\cup \xi _{\gamma _{-}} 
\\ |\gamma _{+}|=|\mathcal{Y}_{+}|=n}}R_{n}^{\kappa }(\xi _{\gamma _{+}}|\xi
_{\gamma _{-}}|\mu _{\mathcal{Y}_{+}}^{(\sigma )}|\mu _{\mathcal{Y}%
_{-}}^{(\sigma )})B(\mu _{\mathcal{Y}_{-}}^{(\sigma )}\cup \xi _{\gamma
_{+}})|\,0\,\rangle ,
\end{equation}%
then by expanding the coefficients, we get:%
\begin{align}
& [\mu _{\mathcal{X}}^{2}-(\mu _{\mathcal{Y}}^{(\sigma )}-\eta
)^{2}]H_{(\sigma _{\mathcal{Y}})}^{\mathcal{B}_{+}}(\mu _{\mathcal{Y}%
})R_{n}^{\kappa }(\xi _{\gamma _{+}}|\xi _{\gamma _{-}}|\mu _{\mathcal{Y}%
_{+}}^{(\sigma )}|\mu _{\mathcal{Y}_{-}}^{(\sigma )})\times \left( \frac{\xi
_{\gamma _{+}}^{\left( 0\right) }+\mu _{\mathcal{Y}_{-}}^{(\sigma )}-\eta }{%
\xi _{\gamma _{+}}^{\left( 0\right) }+\mu _{\mathcal{Y}_{-}}^{(\sigma )}}%
\right) ^{\pm 1} \\
& =\bar S_{n}^{\kappa }(\xi _{\gamma _{+}}|\mu _{\mathcal{Y}_{+}}^{(\sigma
)})\frac{a(\mu _{\mathcal{Y}_{+}}^{(\sigma )})}{a(\xi _{\gamma _{+}})}\frac{%
\mu _{\mathcal{Y}_{-}}^{2}-(\mu _{\mathcal{Y}_{+}}^{(\sigma )}-\eta )^{2}}{%
\mu _{\mathcal{Y}_{-}}^{2}-\mu _{\mathcal{Y}_{+}}^{2}}\frac{\mu _{\mathcal{Y}%
_{-}}^{2}-\xi _{\gamma _{+}}^{\left( 0\right) 2}}{\mu _{\mathcal{Y}%
_{-}}^{2}-\xi _{\gamma _{+}}^{\left( 1\right) }{}^{2}}\frac{\xi _{\gamma
_{+}}-\xi _{\gamma _{-}}+\eta }{\xi _{\gamma _{+}}-\xi _{\gamma _{-}}}\frac{%
H_{(\sigma _{_{\mathcal{Y}_{+}}})}^{\mathcal{B}_{+}}(\mu _{\mathcal{Y}_{+}})%
}{H^{\mathcal{B}_{+}}(\xi _{\gamma _{+}}^{\left( 0\right) })}  \notag \\
& \times \frac{\mu _{\mathcal{Y}_{+}}^{(\sigma )}-\xi _{\gamma _{+}}^{\left(
0\right) }}{\mu _{\mathcal{Y}_{+}}^{(\sigma )}-\xi _{\gamma _{+}}^{\left(
1\right) }}[\mu _{\mathcal{X}}^{2}-(\mu _{\mathcal{Y}_{+}}^{(\sigma )}-\eta
)^{2}][\mu _{\mathcal{X}}^{2}-(\mu _{\mathcal{Y}_{-}}^{(\sigma )}-\eta
)^{2}]H_{(1_{\gamma _{+}}\cup \sigma _{_{\mathcal{Y}_{-}}})}^{\mathcal{B}%
_{+}}(\xi _{\gamma _{+}}^{\left( 0\right) }\cup \mu _{\mathcal{Y}_{-}}).
\end{align}%
We are then free to split the sum over $\sigma _{\mathcal{Y}}=\pm $ in sum
over $\sigma _{\mathcal{Y}_{+}}=\pm $ and $\sigma _{\mathcal{Y}_{-}}=\pm $
and reverse the order of the sum in the following way:%
\begin{align}
& Q_{m}^{\kappa }\mathcal{\bar{B}}_{+}(\mu _{I_{M}})|\,0\,\rangle \left.
=\right. \sum_{a=0}^{M}\sum_{\substack{ \mathcal{X}\cup \mathcal{Y}=I_{M} 
\\ \mathcal{X}\cap \mathcal{Y}=\emptyset ,|\mathcal{Y}|=a}}\frac{\bar{b}%
_{+}\left( \mu _{\mathcal{X}}\right) d(\mu _{\mathcal{X}})d(-\mu _{\mathcal{X%
}})}{\mu _{\mathcal{X}}^{2}-\mu _{\mathcal{Y}}^{2}}\sum_{n=0}^{l(m,a)}\sum 
_{\substack{ \mathcal{Y}=\mathcal{Y}_{+}\cup \mathcal{Y}_{-}  \\ \xi
_{I_{m}}=\xi _{\gamma _{+}}\cup \xi _{\gamma _{-}}  \\ |\gamma _{+}|=|%
\mathcal{Y}_{+}|=n}}  \notag \\
& \left\{ \sum_{\sigma _{\mathcal{Y}_{+}}=\pm }[\mu _{\mathcal{X}}^{2}-(\mu
_{\mathcal{Y}_{+}}^{(\sigma )}-\eta )^{2}]\bar S_{n}^{\kappa }(\xi _{\gamma
_{+}}|\mu _{\mathcal{Y}_{+}}^{(\sigma )})\frac{a(\mu _{\mathcal{Y}%
_{+}}^{(\sigma )})}{a(\xi _{\gamma _{+}})}\frac{H_{(\sigma _{_{\mathcal{Y}%
_{+}}})}^{\mathcal{B}_{+}}(\mu _{\mathcal{Y}_{+}})}{H^{\mathcal{B}_{+}}(\xi
_{\gamma _{+}})}\right.  \notag \\
& \times \left. \frac{\mu _{\mathcal{Y}_{-}}^{2}-(\mu _{\mathcal{Y}%
_{+}}^{(\sigma )}-\eta )^{2}}{\mu _{\mathcal{Y}_{-}}^{2}-\mu _{\mathcal{Y}%
_{+}}^{2}}\frac{\xi _{\gamma _{+}}-\xi _{\gamma _{-}}+\eta }{\xi _{\gamma
_{+}}-\xi _{\gamma _{-}}}\frac{\mu _{\mathcal{Y}_{+}}^{(\sigma )}-\xi
_{\gamma _{+}}^{\left( 0\right) }}{\mu _{\mathcal{Y}_{+}}^{(\sigma )}-\xi
_{\gamma _{+}}^{\left( 1\right) }}\right\} \frac{\mu _{\mathcal{Y}%
_{-}}^{2}-\xi _{\gamma _{+}}^{\left( 0\right) 2}}{\mu _{\mathcal{Y}%
_{-}}^{2}-\xi _{\gamma _{+}}^{\left( 1\right) }{}^{2}}  \notag \\
& \times \left\{ \sum_{\sigma _{\mathcal{Y}_{-}}=\pm }[\mu _{\mathcal{X}%
}^{2}-(\mu _{\mathcal{Y}_{-}}^{(\sigma )}-\eta )^{2}]H_{(1_{\gamma _{+}}\cup
\sigma _{_{\mathcal{Y}_{-}}})}^{\mathcal{B}_{+}}(\xi _{\gamma _{+}}^{\left(
0\right) }\cup \mu _{\mathcal{Y}_{-}})B(\xi _{\gamma _{+}}^{\left( 0\right)
}\cup \mu _{\mathcal{Y}_{-}}^{(\sigma )})|\,0\,\rangle \right\} ,
\end{align}%
and we can, moreover, multiply each term for:%
\begin{equation}
1=\frac{\mu _{\mathcal{X}}^{2}-\xi _{\gamma _{+}}^{\left( 0\right) 2}}{\mu _{%
\mathcal{X}}^{2}-\xi _{\gamma _{+}}^{\left( 1\right) 2}}\frac{\mu _{\mathcal{%
X}}^{2}-\xi _{\gamma _{+}}^{\left( 1\right) 2}}{\mu _{\mathcal{X}}^{2}-\xi
_{\gamma _{+}}^{\left( 0\right) 2}},
\end{equation}%
so that the previous expansion take the following form:%
\begin{align}
& Q_{m}^{\kappa }\mathcal{\bar{B}}_{+}(\mu _{I_{M}})|\,0\,\rangle \left.
=\right. \sum_{a=0}^{M}\sum_{\substack{ \mathcal{X}\cup \mathcal{Y}=I_{M} 
\\ \mathcal{X}\cap \mathcal{Y}=\emptyset ,|\mathcal{Y}|=a}}\frac{\bar{b}%
_{+}\left( \mu _{\mathcal{X}}\right) d(\mu _{\mathcal{X}})d(-\mu _{\mathcal{X%
}})}{\mu _{\mathcal{X}}^{2}-(\xi _{\gamma _{+}}^{\left( 0\right) }\cup \mu _{%
\mathcal{Y}_{-}})^{2}}\sum_{n=0}^{l(m,a)}\sum_{\substack{ \mathcal{Y}=%
\mathcal{Y}_{+}\cup \mathcal{Y}_{-}  \\ \xi _{I_{m}}=\xi _{\gamma _{+}}\cup
\xi _{\gamma _{-}}  \\ |\gamma _{+}|=|\mathcal{Y}_{+}|=n}}  \notag \\
& \left\{ \sum_{\sigma _{\mathcal{Y}_{-}}=\pm }[\mu _{\mathcal{X}}^{2}-(\xi
_{\gamma _{+}}^{\left( 1\right) }\cup (\mu _{\mathcal{Y}_{-}}^{(\sigma
)}-\eta ))^{2}]H_{(1_{\gamma _{+}}\cup \sigma _{_{\mathcal{Y}_{-}}})}^{%
\mathcal{B}_{+}}(\xi _{\gamma _{+}}^{\left( 0\right) }\cup \mu _{\mathcal{Y}%
_{-}})B(\xi _{\gamma _{+}}^{\left( 0\right) }\cup \mu _{\mathcal{Y}%
_{-}}^{(\sigma )})|\,0\,\rangle \right\}  \notag \\
& \left\{ \sum_{\sigma _{\mathcal{Y}_{+}}=\pm }\frac{(\mu _{\mathcal{X}}\cup
\mu _{\mathcal{Y}_{-}})^{2}-(\mu _{\mathcal{Y}_{+}}^{(\sigma )}-\eta )^{2}}{%
(\mu _{\mathcal{X}}\cup \mu _{\mathcal{Y}_{-}})^{2}-\mu _{\mathcal{Y}%
_{+}}^{2}}\bar S_{n}^{\kappa }(\xi _{\gamma _{+}}|\mu _{\mathcal{Y}%
_{+}}^{(\sigma )})\frac{a(\mu _{\mathcal{Y}_{+}}^{(\sigma )})}{a(\xi
_{\gamma _{+}})}\frac{H_{(\sigma _{_{\mathcal{Y}_{+}}})}^{\mathcal{B}%
_{+}}(\mu _{\mathcal{Y}_{+}})}{H^{\mathcal{B}_{+}}(\xi _{\gamma _{+}})}%
\right.  \notag \\
& \times \left. \frac{\xi _{\gamma _{+}}-\xi _{\gamma _{-}}+\eta }{\xi
_{\gamma _{+}}-\xi _{\gamma _{-}}}\frac{\mu _{\mathcal{Y}_{+}}^{(\sigma
)}-\xi _{\gamma _{+}}^{\left( 0\right) }}{\mu _{\mathcal{Y}_{+}}^{(\sigma
)}-\xi _{\gamma _{+}}^{\left( 1\right) }}\frac{(\mu _{\mathcal{X}}\cup \mu _{%
\mathcal{Y}_{-}})^{2}-\xi _{\gamma _{+}}^{\left( 0\right) 2}}{(\mu _{%
\mathcal{X}}\cup \mu _{\mathcal{Y}_{-}})^{2}-\xi _{\gamma _{+}}^{\left(
1\right) 2}}\right\} .
\end{align}%
Let us now remark that the factor associated to the sum over $\sigma _{%
\mathcal{Y}_{+}}$ depends only by the variable on the sets $\mathcal{Y}_{+}$%
, $\mathcal{\bar{Y}}_{-}=\mathcal{Y}_{-}\cup \mathcal{X}$, $\gamma _{-}$\
and $\gamma _{+}$, i.e. it does not distinguish between $\mu _{\mathcal{X}}$
and $\mu _{\mathcal{Y}_{-}}$ so that we are free to rewrite our result as it
follows:%
\begin{align}
& Q_{m}^{\kappa }\mathcal{\bar{B}}_{+}(\mu _{I_{M}})|\,0\,\rangle \left.
=\right. \sum_{n=0}^{l(m,M)}\sum_{\substack{ \mathcal{Y}_{+}\cup \mathcal{%
\bar{Y}}_{-}=I_{M}  \\ \xi _{I_{m}}=\xi _{\gamma _{+}}\cup \xi _{\gamma
_{-}}  \\ |\gamma _{+}|=|\mathcal{Y}_{+}|=n}}\left\{ \sum_{\sigma _{\mathcal{%
Y}_{+}}=\pm }\frac{\mu _{\mathcal{\bar{Y}}_{-}}^{2}-(\mu _{\mathcal{Y}%
_{+}}^{(\sigma )}-\eta )^{2}}{\mu _{\mathcal{\bar{Y}}_{-}}^{2}-\mu _{%
\mathcal{Y}_{+}}^{2}}\bar S_{n}^{\kappa }(\xi _{\gamma _{+}}|\mu _{\mathcal{Y%
}_{+}}^{(\sigma )})\frac{a(\mu _{\mathcal{Y}_{+}}^{(\sigma )})}{a(\xi
_{\gamma _{+}}^{\left( 0\right) })}\right.  \notag \\
& \times \left. \frac{H_{(\sigma _{_{\mathcal{Y}_{+}}})}^{\mathcal{B}%
_{+}}(\mu _{\mathcal{Y}_{+}})}{H^{\mathcal{B}_{+}}(\xi _{\gamma
_{+}}^{\left( 0\right) })}\frac{\xi _{\gamma _{+}}-\xi _{\gamma _{-}}+\eta }{%
\xi _{\gamma _{+}}-\xi _{\gamma _{-}}}\frac{\mu _{\mathcal{Y}_{+}}^{(\sigma
)}-\xi _{\gamma _{+}}^{\left( 0\right) }}{\mu _{\mathcal{Y}_{+}}^{(\sigma
)}-\xi _{\gamma _{+}}^{\left( 1\right) }}\frac{\mu _{\mathcal{\bar{Y}}%
_{-}}^{2}-\xi _{\gamma _{+}}^{\left( 0\right) 2}}{\mu _{\mathcal{\bar{Y}}%
_{-}}^{2}-\xi _{\gamma _{+}}^{\left( 1\right) 2}}\right\} \left\{
\sum_{a=0}^{M-n}\right. \sum_{\substack{ \mathcal{X}\cup \mathcal{Y}_{-}=%
\mathcal{\bar{Y}}_{-}  \\ \mathcal{X}\cap \mathcal{Y}_{-}=\emptyset ,|%
\mathcal{Y}_{-}|=a}}  \notag \\
& \times \frac{\bar{b}_{+}\left( \mu _{\mathcal{X}}\right) d(\mu _{\mathcal{X%
}})d(-\mu _{\mathcal{X}})}{\mu _{\mathcal{X}}^{2}-(\xi _{\gamma
_{+}}^{\left( 0\right) }\cup \mu _{\mathcal{Y}_{-}})^{2}}\sum_{\sigma _{%
\mathcal{Y}_{-}}=\pm }[\mu _{\mathcal{X}}^{2}-(\xi _{\gamma _{+}}^{\left(
1\right) }\cup (\mu _{\mathcal{Y}_{-}}^{(\sigma )}-\eta ))^{2}]H_{(1_{\gamma
_{+}}\cup \sigma _{_{\mathcal{Y}_{-}}})}^{\mathcal{B}_{+}}(\xi _{\gamma
_{+}}^{\left( 0\right) }\cup \mu _{\mathcal{Y}_{-}})  \notag \\
& \left. B(\xi _{\gamma _{+}}^{\left( 0\right) }\cup \mu _{\mathcal{Y}%
_{-}}^{(\sigma )})|\,0\,\rangle \right\} .
\end{align}%
From which we get our result just remarking the identities:%
\begin{align}
\mathcal{R}_{n}^{\kappa }(\xi _{\gamma _{+}}|\xi _{\gamma _{-}}|\mu _{%
\mathcal{Y}_{+}}|\mu _{\mathcal{\bar{Y}}_{-}})& =\sum_{\sigma _{\mathcal{Y}%
_{+}}=\pm }\frac{\mu _{\mathcal{\bar{Y}}_{-}}^{2}-(\mu _{\mathcal{Y}%
_{+}}^{(\sigma )}-\eta )^{2}}{\mu _{\mathcal{\bar{Y}}_{-}}^{2}-\mu _{%
\mathcal{Y}_{+}}^{2}}\bar S_{n}^{\kappa }(\xi _{\gamma _{+}}|\mu _{\mathcal{Y%
}_{+}}^{(\sigma )})\frac{a(\mu _{\mathcal{Y}_{+}}^{(\sigma )})}{a(\xi
_{\gamma _{+}}^{\left( 0\right) })}\frac{H_{(\sigma _{_{\mathcal{Y}_{+}}})}^{%
\mathcal{B}_{+}}(\mu _{\mathcal{Y}_{+}})}{H^{\mathcal{B}_{+}}(\xi _{\gamma
_{+}}^{\left( 0\right) })}  \notag \\
& \times \frac{\xi _{\gamma _{+}}-\xi _{\gamma _{-}}+\eta }{\xi _{\gamma
_{+}}-\xi _{\gamma _{-}}}\frac{\mu _{\mathcal{Y}_{+}}^{(\sigma )}-\xi
_{\gamma _{+}}^{\left( 0\right) }}{\mu _{\mathcal{Y}_{+}}^{(\sigma )}-\xi
_{\gamma _{+}}^{\left( 1\right) }}\frac{\mu _{\mathcal{\bar{Y}}_{-}}^{2}-\xi
_{\gamma _{+}}^{\left( 0\right) 2}}{\mu _{\mathcal{\bar{Y}}_{-}}^{2}-\xi
_{\gamma _{+}}^{\left( 1\right) 2}}
\end{align}%
and%
\begin{align}
& \mathcal{\bar{B}}_{+}(\xi _{\gamma _{+}}\cup \mu _{\mathcal{\bar{Y}}%
_{-}})|\,0\,\rangle =\sum_{a=0}^{M-n}\sum_{\substack{ \mathcal{X}\cup 
\mathcal{Y}_{-}=\mathcal{\bar{Y}}_{-}  \\ \mathcal{X}\cap \mathcal{Y}%
_{-}=\emptyset ,|\mathcal{Y}_{-}|=a}}\frac{\bar{b}_{+}\left( \mu _{\mathcal{X%
}}\right) d(\mu _{\mathcal{X}})d(-\mu _{\mathcal{X}})}{\mu _{\mathcal{X}%
}^{2}-(\xi _{\gamma _{+}}^{\left( 0\right) }\cup \mu _{\mathcal{Y}_{-}})^{2}}
\notag \\
& \times \sum_{\sigma _{\mathcal{Y}_{-}}=\pm }[\mu _{\mathcal{X}}^{2}-(\xi
_{\gamma _{+}}^{\left( 1\right) }\cup (\mu _{\mathcal{Y}_{-}}^{(\sigma
)}-\eta ))^{2}]H_{(1_{\gamma _{+}}\cup \sigma _{_{\mathcal{Y}_{-}}})}^{%
\mathcal{B}_{+}}(\xi _{\gamma _{+}}^{\left( 0\right) }\cup \mu _{\mathcal{Y}%
_{-}})B(\xi _{\gamma _{+}}^{\left( 0\right) }\cup \mu _{\mathcal{Y}%
_{-}}^{(\sigma )})|\,0\,\rangle ,
\end{align}%
as all the terms which correspond to a choice of $\bar{\mu}_{\mathcal{X}%
}\subset \{\pm \xi _{\gamma _{+}}^{\left( 0\right) }\cup \pm \mu _{\mathcal{%
\bar{Y}}_{-}}\}$ with $\bar{\mu}_{\mathcal{X}}\cap \xi _{\gamma
_{+}}^{\left( 0\right) }\neq \emptyset $ are zero being $d(\bar{\mu}_{%
\mathcal{X}})d(-\bar{\mu}_{\mathcal{X}})=0$.
\end{proof}

\subsection{Action of local operators}

Here, we present the action of a basis of local operators at the generic
site $m$ of the chain on the boundary separate states associated to the
transfer matrix $\mathcal{\bar{T}}(\lambda )$. In order, to do so we
introduce some further notations:%
\begin{align}
& \mathcal{R}_{n}^{2,2}(\xi _{\gamma _{+}}|\xi _{\gamma _{-}}|\mu _{\alpha
_{+}}|\mu _{a}|\mu _{\alpha _{-}})\left. =\right. \sum_{\sigma _{\alpha
_{+}},\sigma _{a}=\pm }\frac{\eta \tau (\mu _{a}^{\sigma }|\pm \mu
_{I_{N}-\{a\}}\cup \xi _{I_{m}}^{\left( 1\right) })}{(\mu _{a}^{\sigma }-\xi
_{m+1}^{\left( 1\right) })\tau (\xi _{m+1}^{\left( 0\right) }|\pm \mu
_{I_{N}-\{a\}}\cup \xi _{I_{m}}^{\left( 1\right) })}  \notag \\
& \times \frac{H_{\sigma _{\alpha _{+}}}^{\mathcal{B}_{+}}(\mu _{\alpha
_{+}})H_{\sigma _{a}}^{\mathcal{B}_{+}}(\mu _{a})}{H^{\mathcal{B}_{+}}(\xi
_{\gamma _{+}}^{\left( 0\right) })H^{\mathcal{B}_{+}}(\xi _{m+1}^{\left(
0\right) })}R(\xi _{\gamma _{+}}|\xi _{\gamma _{-}}|\mu _{\alpha
_{+}}^{\sigma }|\pm \mu _{\alpha _{-}})S_{n}^{2,2}(\xi _{\gamma _{+}}|\mu
_{\alpha _{+}}^{\sigma }|\mu _{a}^{\sigma }), \\
& \mathcal{R}_{n}^{-}(\xi _{\gamma _{+}}|\xi _{\gamma _{-}}|\mu _{\alpha
_{+}}|\mu _{\alpha _{-}})\left. =\right. \frac{\sum_{\sigma _{\alpha
_{+}}=\pm }H_{\sigma _{\alpha _{+}}}^{\mathcal{B}_{+}}(\mu _{\alpha
_{+}})R(\xi _{\gamma _{+}}|\xi _{\gamma _{-}}|\mu _{\alpha _{+}}^{\sigma
}|\pm \mu _{\alpha _{-}})S_{n}^{-}(\xi _{\gamma _{+}}|\mu _{\alpha
_{+}}^{\sigma })}{H^{\mathcal{B}_{+}}(\xi _{\gamma _{+}}^{\left( 0\right)
})H^{\mathcal{B}_{+}}(\xi _{m+1}^{\left( 0\right) })\tau (\xi _{m+1}^{\left(
0\right) }|\pm \mu _{I_{N}}\cup \xi _{I_{m}}^{\left( 1\right) })},
\end{align}%
and%
\begin{align}
& \mathcal{R}_{n}^{+}(\xi _{\gamma _{+}}|\xi _{\gamma _{-}}|\mu _{\alpha
_{+}}|\mu _{a}|\mu _{p}|\mu _{\alpha _{-}})\left. =\right. \sum_{\sigma
_{\alpha _{+}},\sigma _{a},\sigma _{p}=\pm }R(\xi _{\gamma _{+}}|\xi
_{\gamma _{-}}|\mu _{\alpha _{+}}^{\sigma }|\pm \mu _{\alpha
_{-}})S_{n}^{+}(\xi _{\gamma _{+}}|\mu _{\alpha _{+}}^{\sigma }|\mu
_{a}^{\sigma }|\mu _{p}^{\sigma })  \notag \\
& \times \frac{H_{\sigma _{\alpha _{+}}}^{\mathcal{B}_{+}}(\mu _{\alpha
_{+}})H_{\sigma _{a}}^{\mathcal{B}_{+}}(\mu _{a})H_{\sigma _{p}}^{\mathcal{B}%
_{+}}(\mu _{p})}{H^{\mathcal{B}_{+}}(\xi _{\gamma _{+}}^{\left( 0\right)
})H^{\mathcal{B}_{+}}(\xi _{m+1}^{\left( 0\right) })}\frac{\eta ^{2}\tau
(\mu _{a}^{\sigma }|\pm \mu _{I_{N}-\{a\}}\cup \xi _{I_{m}}^{\left( 1\right)
})\tau (\mu _{p}^{\sigma }|\pm \mu _{I_{N+1}-\{a,p\}}\cup \xi
_{I_{m}}^{\left( 1\right) })}{(\mu _{a}^{\sigma }-\xi _{m+1}^{\left(
1\right) })(\xi _{m+1}^{\left( 0\right) }-\mu _{p}^{\sigma }+\eta )\tau (\xi
_{m+1}^{\left( 0\right) }|\pm \mu _{I_{N}-\{a\}}\cup \xi _{I_{m}}^{\left(
1\right) })},
\end{align}%
where 
\begin{equation}
S_{n}^{x}=\frac{\xi _{\gamma _{+}}^{\left( 0\right) }-\mu _{\alpha
_{+}}+\eta }{\prod\limits_{\substack{ a>b  \\ a,b\in \gamma _{+}}}\xi
_{ab}\prod\limits_{\substack{ a>b  \\ a,b\in \alpha _{+}}}\mu _{ba}}\text{ \
\ }\underset{k\in \gamma _{+},j\in \alpha _{+}}{\text{det}}M_{jk}^{x}\text{ ,%
}
\end{equation}%
with the $n\times n$ matrix $M^{x}$ defined by: 
\begin{align}
M_{jk}^{2,2}& =r(\xi _{k}^{\left( 0\right) }|\mu _{j})-(1-\delta
_{j,N+1})r(\mu _{j}|\xi _{k}^{\left( 0\right) })\frac{f(\mu _{\alpha
_{+}-\{j\}}\cup \mu _{a}|\mu _{j})}{f(\mu _{j}|\mu _{\alpha _{+}-\{j\}}\cup
\mu _{a})}\frac{f(\mu _{j}|\xi _{\gamma _{+}}^{\left( 0\right) }\cup \xi
_{m+1}^{\left( 0\right) })}{f(\xi _{\gamma _{+}}^{\left( 0\right) }\cup \xi
_{m+1}^{\left( 0\right) }|\mu _{j})}, \\
M_{jk}^{-}& =r(\xi _{k}^{\left( 0\right) }|\mu _{j})-(1-\delta
_{j,N+1})r(\mu _{j}|\xi _{k}^{\left( 0\right) })\frac{f(\mu _{\alpha
_{+}-\{j\}}|\mu _{j})}{f(\mu _{j}|\mu _{\alpha _{+}-\{j\}})}\frac{f(\mu
_{j}|\xi _{\gamma _{+}}^{\left( 0\right) }\cup \xi _{m+1}^{\left( 0\right) })%
}{f(\xi _{\gamma _{+}}^{\left( 0\right) }\cup \xi _{m+1}^{\left( 0\right)
}|\mu _{j})}, \\
M_{jk}^{+}& =r(\xi _{k}^{\left( 0\right) }|\mu _{j})-(1-\delta
_{j,N+1})r(\mu _{j}|\xi _{k}^{\left( 0\right) })\frac{f(\mu _{\alpha
_{+}-\{j\}}\cup \mu _{a}\cup \mu _{p}|\mu _{j})}{f(\mu _{j}|\mu _{\alpha
_{+}-\{j\}}\cup \mu _{a}\cup \mu _{p})}\frac{f(\mu _{j}|\xi _{\gamma
_{+}}^{\left( 0\right) }\cup \xi _{m+1}^{\left( 0\right) })}{f(\xi _{\gamma
_{+}}^{\left( 0\right) }\cup \xi _{m+1}^{\left( 0\right) }|\mu _{j})}.
\end{align}%
Then, the following proposition holds:

\begin{proposition}
The action of the following local operators on boundary separate states
read: 
\begin{align}
E_{m+1}^{2,2}\mathcal{\bar{B}}_{+}(\mu _{I_{M}})|\,0\,\rangle &
=\sum_{a=1}^{M}\sum_{n=0}^{l(m,M)}\sum_{\substack{ I_{M+1}=\alpha _{+}\cup
\alpha _{-}\cup \{a\},  \\ \{a\}\cap (\alpha _{+}\cup \alpha _{-})=\emptyset
,\alpha _{+}\cap \alpha _{-}=\emptyset  \\ I_{m}=\gamma _{+}\cup \gamma
_{-},\gamma _{+}\cap \gamma _{-}=\emptyset  \\ |\gamma _{+}|=|\alpha _{+}|=n 
}}\mathcal{R}_{n}^{2,2}(\xi _{\gamma _{+}}|\xi _{\gamma _{-}}|\mu _{\alpha
_{+}}|\mu _{a}|\mu _{\alpha _{-}})  \notag \\
& \times \mathcal{\bar{B}}_{+}(\mu _{\alpha _{-}}\cup \xi _{\gamma
_{+}}^{\left( 0\right) })|\,0\,\rangle , \\
\sigma _{m+1}^{-}\mathcal{\bar{B}}_{+}(\mu _{I_{M}})|\,0\,\rangle &
=\sum_{n=0}^{l(m,M)}\sum_{\substack{ I_{M+1}=\alpha _{+}\cup \alpha
_{-},\alpha _{+}\cap \alpha _{-}=\emptyset  \\ I_{m}=\gamma _{+}\cup \gamma
_{-},\gamma _{+}\cap \gamma _{-}=\emptyset  \\ |\gamma _{+}|=|\alpha _{+}|=n 
}}\mathcal{R}_{n}^{-}(\xi _{\gamma _{+}}|\xi _{\gamma _{-}}|\mu _{\alpha
_{+}}|\mu _{\alpha _{-}})  \notag \\
& \times \mathcal{\bar{B}}_{+}(\mu _{\alpha _{-}}\cup \xi _{\gamma
_{+}}^{\left( 0\right) })|\,0\,\rangle ,
\end{align}%
and 
\begin{align}
\sigma _{m+1}^{+}\mathcal{\bar{B}}_{+}(\mu _{I_{M}})|\,0\,\rangle &
=\sum_{a=1}^{M}\sum_{\substack{ p=1  \\ p\neq a}}^{M+1}\sum_{n=0}^{l(m,M)}%
\sum _{\substack{ I_{M+1}=\alpha _{+}\cup \alpha _{-}\cup \{a\}\cup \{p\}, 
\\ \{a,p\}\cap (\alpha _{+}\cup \alpha _{-})=\emptyset ,\alpha _{+}\cap
\alpha _{-}=\emptyset  \\ I_{m}=\gamma _{+}\cup \gamma _{-},\gamma _{+}\cap
\gamma _{-}=\emptyset  \\ |\gamma _{+}|=|\alpha _{+}|=n}}\mathcal{R}%
_{n}^{+}(\xi _{\gamma _{+}}|\xi _{\gamma _{-}}|\mu _{\alpha _{+}}|\mu
_{a}|\mu _{p}|\mu _{\alpha _{-}})  \notag \\
& \times \mathcal{\bar{B}}_{+}(\mu _{\alpha _{-}}\cup \xi _{\gamma
_{+}}^{\left( 0\right) })|\,0\,\rangle ,
\end{align}%
where we have defined $\mu _{M+1}=\xi _{m+1}.$
\end{proposition}

\begin{proof}
The proof of this proposition is done exactly along the same lines of the
previous one, we have to use the boundary-bulk decomposition of the boundary
separate states and then using the bulk formulae of Proposition 4.2 of our
paper \cite{KitKMNST08} for the action of these local operators on bulk
states, the result is computed.
\end{proof}

These formulae allows the computation of all one-point functions and we can
use them to derive the action of a monomial of two local operators, for
example one at the site 1 and one at the generic site $m$ of the chain, here
we present one instance of this:

\begin{cor}
The previous proposition implies the following form of the right action of $%
\sigma _{1+m}^{+}\sigma _{1}^{-}$ on an arbitrary separate state: 
\begin{align}
\sigma _{1+m}^{+}\sigma _{1}^{-}\mathcal{\bar{B}}_{+}(\mu
_{I_{M}})|\,0\,\rangle & =\sum_{a=1}^{M}\sum_{\substack{ p=1  \\ p\neq a}}%
^{M+1}\sum_{n=0}^{m-1}\sum_{\substack{ I_{M+1}=\alpha _{+}\cup \alpha
_{-}\cup \{a\}\cup \{p\},  \\ \{\alpha _{+}\cup \alpha _{-}\}\cap
\{a,p\}=\emptyset ,\alpha _{+}\cap \alpha _{-}=\emptyset  \\ %
\{2,..,m\}=\gamma _{+}\cup \gamma _{-},\gamma _{+}\cap \gamma _{-}=\emptyset 
\\ |\gamma _{+}|=|\alpha _{+}|=n}}\mathcal{\bar{B}}_{+}(\mu _{\alpha
_{-}}\cup \xi _{1}^{\left( 0\right) }\cup \xi _{\gamma _{+}}^{\left(
0\right) })|\,0\,\rangle  \notag \\
& \times \mathcal{R}_{n}^{+-}(\xi _{1}|\xi _{\gamma _{+}}|\xi _{\gamma
_{-}}|\mu _{\alpha _{+}}|\mu _{a}|\mu _{p}|\mu _{\alpha _{-}}),
\end{align}%
where:%
\begin{align}
& \mathcal{R}_{n}^{+-}(\xi _{1}|\xi _{\gamma _{+}}|\xi _{\gamma _{-}}|\mu
_{\alpha _{+}}|\mu _{a}|\mu _{p}|\mu _{\alpha _{-}})\left. =\right.
\sum_{\sigma _{\alpha _{+}},\sigma _{a},\sigma _{p}=\pm }\frac{H_{\sigma
_{\alpha _{+}}}^{\mathcal{B}_{+}}(\mu _{\alpha _{+}})H_{\sigma _{a}}^{%
\mathcal{B}_{+}}(\mu _{a})H_{\sigma _{p}}^{\mathcal{B}_{+}}(\mu _{p})}{H^{%
\mathcal{B}_{+}}(\xi _{1}^{\left( 0\right) }\cup \xi _{\gamma _{+}}^{\left(
0\right) })H^{\mathcal{B}_{+}}(\xi _{m+1}^{\left( 0\right) })}  \notag \\
& \times \frac{R(\xi _{\gamma _{+}}|\xi _{\gamma _{-}}|\mu _{\alpha
_{+}}^{\sigma }|\pm \mu _{\alpha _{-}})S_{n}^{+}(\xi _{\gamma _{+}}|\mu
_{\alpha _{+}}^{\sigma }|\mu _{a}^{\sigma }|\mu _{p}^{\sigma })\eta ^{2}}{%
\tau (\xi _{1}^{\left( 0\right) }|\pm \mu _{\alpha _{-}}\cup \mu _{\alpha
_{+}}^{\sigma })(\mu _{a}^{\sigma }-\xi _{m+1}^{\left( 1\right) })(\xi
_{m+1}^{\left( 0\right) }-\mu _{p}^{\sigma }+\eta )}  \notag \\
& \times \frac{\tau (\mu _{a}^{\sigma }|\pm \mu _{I_{M}-\{a\}}\cup \xi
_{I_{m}}^{\left( 1\right) })\tau (\mu _{p}^{\sigma }|\pm \mu
_{I_{M+1}-\{a,p\}}\cup \xi _{I_{m}}^{\left( 1\right) })}{\tau (\xi
_{m+1}^{\left( 0\right) }|\pm \mu _{I_{M}-\{a\}}\cup \xi _{I_{m}}^{\left(
1\right) })},
\end{align}%
and $\mu _{M+1}=\xi _{m+1}$.
\end{cor}

\subsection{Action of the quasi-local generators of elementary blocks\label%
{Act-El-Block}}

Let us present now the action of the basis $\prod\limits_{j=1}^{m}E_{j}^{%
\varepsilon _{j},\varepsilon _{j}^{\prime }}$ of quasi-local operators from
site one to the generic site $m$ of the chain on the boundary separate
states associated to the the transfer matrix $\mathcal{\bar{T}}(\lambda )$.
Following the exposition of our previous paper \cite{KitKMNST07}, for any
element $\prod\limits_{j=1}^{m}E_{j}^{\varepsilon _{j},\varepsilon
_{j}^{\prime }}$ of this basis, we can introduce the following sets of
integers:%
\begin{eqnarray}
&&\{i_{p}\}_{p\in \{1,\ldots ,s\}},\text{ with }i_{k}<i_{h}\text{ for }%
0<k<h\leq s,  \label{i_p-Def0} \\
&&\{i_{p}\}_{p\in \{s+1,\ldots ,s+s^{\prime }\}},\text{ with }i_{k}>i_{h}%
\text{ for }s<k<h\leq s+s^{\prime },  \label{i_p-Def1}
\end{eqnarray}%
defined by the conditions $\,\varepsilon _{j}=1$ iff $j\in \{i_{p}\}_{p\in
\{s+1,\ldots ,s+s^{\prime }\}}$ and $\varepsilon _{j}^{\prime }=2$ iff $j\in
\{i_{p}\}_{p\in \{1,\ldots ,s\}}$. Then, we can formulate the following
proposition:

\begin{proposition}
\label{Action-ElBlock}Taken the generic element of the quasi-local basis $%
\prod\limits_{j=1}^{m}E_{j}^{\varepsilon _{j},\varepsilon _{j}^{\prime }}$,
then its action on a boundary separate state reads: 
\begin{equation}
\prod\limits_{j=1}^{m}E_{j}^{\varepsilon _{j},\varepsilon _{j}^{\prime }}%
\mathcal{\bar{B}}_{+}(\mu _{I_{M}})|\,0\,\rangle =\sum\limits_{\beta
_{s+s^{\prime }}}\mathcal{F}_{\beta _{s+s^{\prime }}}^{+}(\mu _{I_{M+m}})%
\mathcal{\bar{B}}_{+}(\mu _{I_{M+m}\backslash \beta _{s+s^{\prime
}}})|\,0\,\rangle ,
\end{equation}%
here, we have defined $\mu _{M+j}:=\xi _{m+1-j}^{\left( 0\right) }\text{ for 
}j\in \{1,\ldots ,m\}$, the sum run over all the possible sets of integers $%
\beta _{s+s^{\prime }}=\{b_{1},\ldots ,b_{s+s^{\prime }}\}$ whose elements
satisfy the conditions 
\begin{equation}
\begin{cases}
b_{p}\in \{1,\ldots ,M\}\setminus \{b_{1},\ldots ,b_{p-1}\}\qquad & \text{%
for }0<p\leq s, \\ 
b_{p}\in \{1,\ldots ,M+m+1-i_{p}\}\setminus \{b_{1},\ldots ,b_{p-1}\}\  & 
\text{for }s<p\leq s+s^{\prime },%
\end{cases}
\label{beta-cond}
\end{equation}%
and the coefficient reads: 
\begin{align}
\mathcal{F}_{\beta _{s+s^{\prime }}}^{+}(\mu _{I_{M+m}})&
=\sum\limits_{\sigma _{\alpha _{+}}=\pm }\frac{a(\mu _{\beta _{s+s^{\prime
}}}^{\sigma })}{a(\xi _{I_{M}}^{\left( 0\right) })}\ \frac{H_{\sigma
_{\alpha _{+}}}^{\mathcal{B}_{+}}(\mu _{\alpha _{+}})}{H^{\mathcal{B}%
_{+}}(\xi _{\gamma _{+}}^{\left( 0\right) })}\frac{(\xi _{\gamma
_{+}}^{\left( 0\right) }-\mu _{\alpha _{+}}^{\sigma })}{(\xi _{\gamma
_{+}}^{\left( 1\right) }-\mu _{\alpha _{+}}^{\sigma })}\prod\limits_{1\leq
i<j\leq s+s^{\prime }}\frac{\mu _{b_{i}b_{j}}^{\sigma }}{\mu
_{b_{i}b_{j}}^{\sigma }+\eta }  \notag \\
& \times \prod\limits_{\epsilon =\pm }\left( \frac{(\mu _{\alpha
_{+}}^{\sigma }+\epsilon \mu _{\alpha _{-}}-\eta )(\xi _{\gamma
_{+}}^{\left( 0\right) }+\varepsilon \mu _{\alpha _{-}})}{(\mu _{\alpha
_{+}}^{\sigma }+\epsilon \mu _{\alpha _{-}})(\xi _{\gamma _{+}}^{\left(
1\right) }+\varepsilon \mu _{\alpha _{-}})}\right) \prod\limits_{i\in \alpha
_{+}}\frac{\prod\limits_{j\in \alpha _{+}}(\mu _{ji}^{\sigma }+\eta )}{%
\prod\limits_{j\in \alpha _{+}-\{i\}}(\mu _{ji}^{\sigma })}  \notag \\
& \times \prod\limits_{p=1}^{s}\frac{\prod\limits_{k=i_{p}+1}^{m}(\mu
_{b_{p}}^{\sigma }-\xi _{k}^{\left( 1\right) })}{\prod\limits_{k=i_{p}}^{m}(%
\mu _{b_{p}}^{\sigma }-\xi _{k}^{\left( 0\right) })}\prod%
\limits_{p=s+1}^{s+s^{\prime }}\frac{\prod\limits_{k=i_{p}+1}^{m}(\xi
_{k}^{\left( 0\right) }-\mu _{b_{p}}^{\sigma }+\eta )}{\prod\limits 
_{\substack{ k=i_{p}  \\ \!\!\!\!k\neq N+m+1-b_{p}\!\!\!\!}}%
^{m}\!\!\!\!\!\!(\xi _{k}^{\left( 0\right) }-\mu _{b_{p}}^{\sigma })}.
\end{align}%
The sum is over all $\sigma _{h}\in \{+,-\}$ for $h\in \alpha _{+}$ and we
have defined $\mu _{j}^{\sigma }:=\sigma _{j}\mu _{j}$ for $j\in \beta
_{s+s^{\prime }}$, with $\sigma _{j}=1$ if $j>M$, and 
\begin{eqnarray}
\gamma _{-} &=&\{M+m+1-j\}_{j\in \beta _{s+s^{\prime }}\cap \{M+1,\ldots
,M+m\}},\quad \gamma _{+}=\{1,\ldots ,m\}\setminus \gamma _{-}, \\
\alpha _{+} &=&\beta _{s+s^{\prime }}\cap \{1,\ldots ,M\},\alpha
_{-}=\{1,\ldots ,M\}\setminus \alpha _{+}.
\end{eqnarray}
\end{proposition}

\begin{proof}
Here, we have to use the boundary-bulk decomposition of the boundary
separate states and then we have to use the following formula:%
\begin{equation}
\prod\limits_{j=1}^{m}E_{j}^{\varepsilon _{j},\varepsilon _{j}^{\prime
}}B(\mu _{I_{M}})|\,0\,\rangle =\sum\limits_{\beta _{s+s^{\prime }}}\mathcal{%
F}_{\beta _{s+s^{\prime }}}(\mu _{I_{M+m}})B(\mu _{I_{M+m}\backslash \beta
_{s+s^{\prime }}})|\,0\,\rangle ,
\end{equation}%
for these actions on bulk states, where the coefficient $\mathcal{F}_{\beta
_{m}}$ reads:%
\begin{align}
\mathcal{F}_{\beta _{s+s^{\prime }}}(\mu _{I_{M+m}})& =\frac{a(\mu _{\beta
_{s+s^{\prime }}}^{\sigma })}{a(\xi _{I_{M}}^{\left( 0\right) })}\
\prod\limits_{j=1}^{s+s^{\prime }}\frac{\prod\limits_{k=1}^{M}(\mu
_{k\,b_{j}}+\eta )}{\prod\limits_{\substack{ k=1  \\ k\neq b_{j}}}^{M}\mu
_{k\,b_{j}}}\ \frac{\mu _{I_{M}}-\xi _{I_{m}}^{\left( 0\right) }}{\mu
_{I_{M}}-\xi _{I_{m}}^{\left( 1\right) }}\prod_{1\leq i<j\leq m}\frac{\mu
_{b_{i}\,b_{j}}}{\mu _{b_{i}\,b_{j}}+\eta }  \notag \\
& \times \prod_{p=1}^{s}\frac{\prod\limits_{k=i_{p}+1}^{m}(\mu _{b_{p}}-\xi
_{k}+\eta )}{\prod\limits_{k=i_{p}}^{m}(\mu _{b_{p}}-\xi _{k})}%
\prod_{p=s+1}^{s+s^{\prime }}\frac{\prod\limits_{k=i_{p}+1}^{m}(\xi _{k}-\mu
_{b_{p}}+\eta )}{\prod\limits_{\substack{ k=i_{p}  \\ k\neq N+m+1-b_{p}}}%
^{m}(\xi _{k}-\mu _{b_{p}})},
\end{align}%
as proven in Proposition 5.1 of our paper \cite{KitKMNST07}. The proof of
this proposition is done exactly along the same lines of the Proposition \ref%
{Action-Q_m},
\end{proof}

\section{Scalar products of separate states}

Here, we present the scalar products of boundary separate states with the
transfer matrix eigenstates directly in the framework for which we want to
compute correlation functions, i.e. for the case%
\begin{equation}
\bar{\mathsf{c}}_{-}=0,\text{ \ }\bar{\mathsf{b}}_{+}\neq 0,
\end{equation}%
where our original transfer matrix $\mathcal{T}(\lambda )$ is associated to
unparallel boundary magnetic fields and it is isospectral to the transfer
matrix $\mathcal{\hat{T}}(\lambda )$ associated to parallel boundary
magnetic fields. Here, we first rewrite the known scalar products results,
derived in our previous paper \cite{KitMNT17},
and then we analyze them in the thermodynamic limit for the case of the
ground state thanks to the knowledge of its density root distribution
achieved thanks to the derived isospectrality.

\subsection{The scalar product of separate states with transfer matrix
eigenstates}

In order to compute correlation functions, we have to be able to compute the
following type of ratios of scalar products of separate states with the
transfer matrix eigenstates:%
\begin{equation}
\frac{\langle \,\bar{t}\,|\,\mathcal{\bar{B}}_{+}(\upsilon _{\gamma }\cup
\mu _{\alpha })|\,0\,\rangle }{\langle \,\bar{t}\,|\mathcal{\bar{B}}_{+}(\mu
_{I_{q}})|0\,\rangle },\text{ with }\alpha \subset I_{q},\text{ }\gamma
\subset I_{m},\text{\ }
\end{equation}%
where $\langle \,\bar{t}\,|(=\langle \,t\,|\,\Gamma _{W}^{-1})$ is the
unique (up normalization) eigencovector of $\mathcal{\bar{T}}(\lambda )$
associated to the eigenvalue $t(\lambda )$ solving with $Q_{t}(\lambda )$, $%
\lambda ^{2}$-polynomial of degree $q$ and roots $\lambda _{1}^{2},\ldots
,\lambda _{q}^{2}$, the homogeneous version 
\begin{equation}
t(\lambda )\,Q_{t}(\lambda )=\mathsf{A}_{\bar{\zeta}_{+},\bar{\zeta}%
_{-}}(\lambda )\,Q_{t}(\lambda -\eta )+\mathsf{A}_{\bar{\zeta}_{+},\bar{\zeta%
}_{-}}(-\lambda )\,Q_{t}(\lambda +\eta )  \label{hom-BAX-eq}
\end{equation}
of the Baxter $TQ$-equation \eqref{Inhom-BAX-eq}. Denoted $\bar{\mu}%
_{I_{p+n}}=\upsilon _{\gamma }\cup \mu _{\alpha }$ with%
\begin{equation}
\bar{\mu}_{a}=\mu _{\alpha _{a}}\text{ for }a\in \{1,...,p\},\text{ \ }\bar{%
\mu}_{p+a}=\mu _{\gamma _{a}}\text{ for }a\in \{p+1,...,p+n\},
\end{equation}%
$\gamma =\{\gamma _{1},...,\gamma _{n}\}$ and $\alpha =\{\alpha
_{1},...,\alpha _{p}\}$, then the following proposition holds:

\begin{proposition}
Let the inhomogeneity parameters $\xi _{1},\ldots ,\xi _{N}$ be generic %
\eqref{cond-inh}, 
\begin{equation}
\bar{\mathsf{c}}_{-}=0,\text{ \ }\bar{\mathsf{b}}_{+}\neq 0,
\end{equation}%
and let $t(\lambda )$ be a generic $\mathcal{\bar{T}}$ -eigenvalue and $%
\langle \,\bar{t}\,|$ the associated unique (up normalization) $\mathcal{%
\bar{
T}}$ -eigencovector$\,$, then, the following representations hold:%
\begin{align}
\frac{\langle \,\bar{t}\,|\,\mathcal{\bar{B}}_{+}(\upsilon _{\gamma }\cup
\mu _{\alpha })|\,0\,\rangle }{\langle \,\bar{t}\,|\mathcal{\bar{B}}_{+}(\mu
_{I_{q}})|0\,\rangle }& =0\text{,\ \ if }p+n<q, \\
\frac{\langle \,\bar{t}\,|\,\mathcal{\bar{B}}_{+}(\upsilon _{\gamma }\cup
\mu _{\alpha })|\,0\,\rangle }{\langle \,\bar{t}\,|\mathcal{\bar{B}}_{+}(\mu
_{I_{q}})|0\,\rangle }& =\frac{\Gamma ((\bar{\zeta}_{+}+\bar{\zeta}%
_{-})/\eta +\!N\!-\!(q\!+p+n))}{\Gamma ((\bar{\zeta}_{+}+\bar{\zeta}%
_{-})/\eta +\!N\!-\!2q\!)}\frac{Q_{t}(\upsilon _{\gamma })(4\mu _{\beta
}^{2}-\eta ^{2})}{(4\upsilon _{\gamma }^{2}-\eta ^{2})Q_{t}(\mu _{\beta })}\,
\notag \\
& \times \frac{\widehat{V}(\mu _{\beta })\,}{\widehat{V}(\upsilon _{\gamma })%
}\frac{\mu _{\beta }^{2}-\mu _{\alpha }^{2}}{\upsilon _{\gamma }^{2}-\mu
_{\alpha }^{2}}\frac{\det_{p+n}\mathcal{S}_{t}(\lambda _{I_{q}}|\bar{\mu}%
_{I_{p+n}})}{\det_{q}\mathcal{S}_{t}(\lambda _{I_{q}}|\mu _{I_{q}})}\text{,
\ if }p+n\left. \geq \right. q,
\end{align}%
where $\Gamma(\lambda)$ is the gamma function and we have defined the set $\beta =I_{q}\backslash \alpha $ and, if $%
p+n\geq q$, the $(p+n)\times (p+n)$ square matrix $\mathcal{S}_{t}(\lambda
_{I_{q}}|\omega _{I_{p+n}})$ is defined by:%
\begin{equation}
\mathcal{S}_{t}(\lambda _{I_{q}}|\omega _{I_{p+n}})_{i,k}=%
\begin{cases}
{\partial \,t(\omega _{i})/}\partial \lambda _{k}\quad & \text{if }k\leq q,%
\vspace{2mm} \\ 
{\sum_{\epsilon \in \{+,-\}}\!\!\!\ \epsilon \,\mathsf{A}_{\bar{\zeta}_{+},%
\bar{\zeta}_{-}}(-\epsilon \omega _{i})\,\frac{Q_{t}(\omega _{i}+\epsilon
\eta )}{Q_{t}(\omega _{i})}\left( \omega _{i}+\epsilon \frac{\eta }{2}%
\right) ^{\!2(k-q)-1}}\ \  & \text{if }k>q.%
\end{cases}
\label{Slav-SP}
\end{equation}%
Then, the following identity holds%
\begin{equation}
\frac{\langle \,\bar{t}\,|\,\mathcal{\bar{B}}_{+}(\upsilon _{\gamma }\cup
\mu _{\alpha })|\,0\,\rangle }{\langle \,\bar{t}\,|\mathcal{\bar{B}}_{+}(\mu
_{I_{q}})|0\,\rangle }=\frac{\langle \,\hat{t}\,|\,\mathcal{\hat{B}}%
_{+}(\upsilon _{\gamma }\cup \mu _{\alpha })|\,0\,\rangle }{\langle \,\hat{t}%
\,|\mathcal{\hat{B}}_{+}(\mu _{I_{q}})|0\,\rangle }\text{,\ \ if }p+n\leq q,
\end{equation}%
where $\langle \,\hat{t}\,|\,$\ is the unique (up normalization)
eigencovector associated to the eigenvalue $t(\lambda )$ of the transfer
matrix $\mathcal{\hat{T}}(\lambda )$ defined in $\left( \ref{Def-T-Hat}%
\right) $.
\end{proposition}

\begin{proof}[Proof]
This proposition is a direct corollary of Theorem 4.2 of our previous paper 
\cite{KitMNT17}, we have just to use the formula (4.59) there to compute
this ratio. Note that we do not need to specify the normalization of the
eigencovector $\langle \,\bar{t}\,|\,$ as it appear to numerator and
denominator simultaneously. The second part of the proposition is then a
corollary on the Slavnov's type formulae \cite{Sla89} for the scalar
products of Bethe's like states \cite{Wan02,KitKMNST07}. In particular, we
can use the formula (4.9) of Theorem 4.1 of our paper \cite{KitKMNST07} to
compute the following ratio and derive the identity:%
\begin{equation}
\frac{\langle \,\hat{t}\,|\,\mathcal{\hat{B}}_{+}(\upsilon _{\gamma }\cup
\mu _{\alpha })|\,0\,\rangle }{\langle \,\hat{t}\,|\mathcal{\hat{B}}_{+}(\mu
_{I_{q}})|0\,\rangle }=\delta _{q,M+n}\frac{\langle \,\bar{t}\,|\,\mathcal{%
\bar{B}}_{+}(\upsilon _{\gamma }\cup \mu _{\alpha })|\,0\,\rangle }{\langle
\,\bar{t}\,|\mathcal{\bar{B}}_{+}(\mu _{I_{q}})|0\,\rangle }.
\end{equation}
\end{proof}

The above proposition also point out that the scalar products for the model
associated to the diagonal boundary matrices $\hat{K}_{-}(\lambda )$ and $%
\hat{K}_{+}(\lambda )$ and those associated to the diagonal $\bar{K}%
_{-}(\lambda )=\hat{K}_{-}(\lambda )$ and triangular $\bar{K}_{+}(\lambda
)\neq \hat{K}_{+}(\lambda )$ boundary matrices do not coincide for $p+n>q$
being in general:%
\begin{equation}
\frac{\langle \,\bar{t}\,|\,\mathcal{\bar{B}}_{+}(\upsilon _{\gamma }\cup
\mu _{\alpha })|\,0\,\rangle }{\langle \,\bar{t}\,|\mathcal{\bar{B}}_{+}(\mu
_{I_{q}})|0\,\rangle }\neq 0\text{ \ for \ }p+n>q.
\end{equation}%
Here, our main results will be to prove that if $t(\lambda )$ is the ground
state for the Hamiltonian associated to $\mathcal{\bar{T}}(\lambda )$ these
scalar products for $\lambda _{I_{q}}=\mu _{I_{q}}$ go to zero quickly
enough to make their contribution to correlation functions zero in the
thermodynamic and homogeneous limit. More precisely, the following result
holds:

\begin{proposition}
\label{Sp-Thermo}Let the inhomogeneity parameters $\xi _{1},\ldots ,\xi _{N}$
be generic \eqref{cond-inh},%
\begin{equation}
\bar{\mathsf{c}}_{-}=0,\text{ \ }\bar{\mathsf{b}}_{+}\neq 0,
\end{equation}%
and let us fix $t(\lambda )$ ($\mathcal{\bar{T}}$ -eigenvalue) and $%
Q_{t}(\lambda )$ (associated solution of the homogeneous Baxter's equation %
\eqref{hom-BAX-eq}) such that the $Q_{t}$-roots $\{\lambda _{I_{q}}\}$ are
distributed on the positive real axis according to the ground state density:%
\begin{equation}
\rho (\lambda )=\frac{1}{\,\cosh \pi \lambda },
\end{equation}%
in the thermodynamic limits, then in this limit it holds:%
\begin{align}
\frac{\langle \,\bar{t}\,|\,\mathcal{\bar{B}}_{+}(\xi _{\gamma }^{\left(
0\right) }\cup \lambda _{\alpha })|\,0\,\rangle }{\langle \,\bar{t}\,|%
\mathcal{\bar{B}}_{+}(\lambda _{I_{q}})|0\,\rangle }& =0\text{,\ \ if }p+n<q,\label{lessB}
\\
\frac{\langle \,\bar{t}\,|\,\mathcal{\bar{B}}_{+}(\xi _{\gamma }^{\left(
0\right) }\cup \lambda _{\alpha })|\,0\,\rangle }{\langle \,\bar{t}\,|%
\mathcal{\bar{B}}_{+}(\lambda _{I_{q}})|0\,\rangle }& =o(1/N^{(q-p)})\text{%
,\ \ if }p+n\left. >\right. q,\label{moreB}
\end{align}%
if q-p finite in the thermodynamic limit, and finally if $p+n\left. =\right.
q$:%
\begin{align}
\frac{\langle \,\bar{t}\,|\,\mathcal{\bar{B}}_{+}(\xi _{\gamma }^{\left(
0\right) }\cup \lambda _{\alpha })|\,0\,\rangle }{\langle \,\bar{t}\,|%
\mathcal{\bar{B}}_{+}(\lambda _{I_{q}})|0\,\rangle }& =\frac{\lambda _{\beta
}\,\lambda _{\beta }^{\left( 1\right) }\,(\xi _{\gamma }+\eta )\,y(\xi
_{\gamma }^{\left( 0\right) };\{\lambda _{I_{q}}\};\bar{\zeta}_{\pm })}{\xi
_{\gamma }\xi _{\gamma }^{\left( 0\right) }\,\lambda _{\beta }^{\left(
0\right) }\,y(\lambda _{\beta }\,;\{\lambda _{I_{q}}\};\bar{\zeta}_{\pm })}%
\frac{\widehat{V}(\lambda _{\beta })\,}{\widehat{V}(\xi _{\gamma }^{\left(
0\right) })}\frac{\lambda _{\beta }^{2}-\lambda _{\alpha }^{2}}{\xi _{\gamma
}^{\left( 0\right) 2}-\lambda _{\alpha }^{2}}  \notag \\
& \times \det_{n}\left[ \frac{\rho (\lambda _{\beta _{l}}-\xi _{\gamma
_{k}})-\rho (\lambda _{\beta _{l}}+\xi _{\gamma _{k}})}{2N\rho (\lambda
_{\beta _{l}})}\right] +o(1/N^{n})\text{,}\label{=B}
\end{align}%
where, we have defined%
\begin{equation}
y(\mu ;\{\lambda \};\bar{\zeta}_{\pm })=a(\mu )\,d(-\mu )\,(\mu +\bar{\zeta}%
_{+}-\eta /2)\,(\mu +\bar{\zeta}_{-}-\eta /2)Q_{t}(\mu -\eta )\,
\end{equation}%
and $\beta =\{\beta _{1},...,\beta _{n}\}=I_{q}\backslash \alpha $.
\end{proposition}

\begin{proof}
Thanks to the previous proposition we already know the validity of the
statement for $p+n<q$ while for $p+n=q$ this can be proven exactly along the
same lines with which we derived it in the Bethe ansatz framework, see
Section 4.4 of our previous paper \cite{KitKMNST07}. So we are left with the
proof in the case $n+p>q$.

Let us observe that we can write the result of the previous proposition for
the case $n+p>q$ as it follows:%
\begin{align}
\frac{\langle \,\bar{t}\,|\,\mathcal{\bar{B}}_{+}(\xi _{\gamma }^{\left(
0\right) }\cup \lambda _{\alpha })|\,0\,\rangle }{\langle \,\bar{t}\,|%
\mathcal{\bar{B}}_{+}(\lambda _{I_{q}})|0\,\rangle }& =\frac{\lambda _{\beta
}\,\lambda _{\beta }^{\left( 1\right) }\,(\xi _{\gamma }+\eta )\,y(\xi
_{\gamma }^{\left( 0\right) };\{\lambda _{I_{q}}\};\bar{\zeta}_{\pm })}{\xi
_{\gamma }\xi _{\gamma }^{\left( 0\right) }\,\lambda _{\beta }^{\left(
0\right) }\,y(\lambda _{\beta }\,;\{\lambda _{I_{q}}\};\bar{\zeta}_{\pm })}%
\frac{\widehat{V}(\lambda _{\beta })\,}{\widehat{V}(\xi _{\gamma }^{\left(
0\right) })}\frac{\lambda _{\beta }^{2}-\lambda _{\alpha }^{2}}{\xi _{\gamma
}^{\left( 0\right) 2}-\lambda _{\alpha }^{2}}  \notag \\
& \times \frac{\Gamma ((\bar{\zeta}_{+}+\bar{\zeta}_{-})/\eta
+\!N\!-\!(q\!+p+n))}{\Gamma ((\bar{\zeta}_{+}+\bar{\zeta}_{-})/\eta
+\!N\!-\!2q\!)}\frac{\det_{n+p}\mathcal{M}_{t}(\lambda _{I_{q}}|\xi _{\gamma
}^{\left( 0\right) }\cup \lambda _{\alpha })}{\det_{q}\mathcal{N}%
_{t}(\lambda _{I_{q}})}\text{,}
\end{align}%
where $\mathcal{N}_{t}$ is the matrix related to the Gaudin norm\footnote{%
See formula (4.29) of \cite{KitKMNST07} for its explicit expression.} and
the $(n+p)\times (n+p)$ matrix $\mathcal{M}_{t}(\lambda _{I_{q}}|\xi
_{\gamma }^{\left( 0\right) }\cup \lambda _{\alpha })$ has the following
representation%
\begin{equation}
\mathcal{M}_{t}(\lambda _{I_{q}}|\xi _{\gamma }^{\left( 0\right) }\cup
\lambda _{\alpha })=\left( 
\begin{array}{cc}
\mathcal{M}_{q\times p}^{\left( 1,1\right) } & \mathcal{M}_{q\times
n}^{\left( 1,2\right) } \\ 
\mathcal{M}_{(n+p-q)\times p}^{\left( 2,1\right) } & \mathcal{M}%
_{(n+p-q)\times n}^{\left( 2,2\right) }%
\end{array}%
\right)
\end{equation}%
where%
\begin{align}
\mathcal{M}_{i,k}^{\left( 1,1\right) }& =\mathcal{N}_{i,k},\text{ \ \ \ \ \
\ \ \ \ \ \ \ \ \ \ \ \ \ \ \ \ \ \ \ \ \ \ \ \ \ \ \ \ \ \ \ \ \ \ \ \ \ \ }%
i\leq q,k\leq p,  \notag \\
\mathcal{M}_{i,k}^{\left( 1,2\right) }& =i[r(\lambda _{\alpha _{i}},\xi
_{\gamma _{k}}^{\left( 0\right) })-r(\lambda _{\alpha _{i}},-\xi _{\gamma
_{k}}^{\left( 1\right) })],\text{\ \ \ \ \ \ \ }i\leq q,k\leq n,  \notag \\
\mathcal{M}_{i,k}^{\left( 2,1\right) }& =\left( \lambda _{\alpha
_{k}}^{\left( 1\right) }\right) ^{2i-1}+\left( \lambda _{\alpha
_{k}}^{\left( 0\right) }\right) ^{2i-1},\text{\ \ }i\leq n+p-q,k\leq p, 
\notag \\
\mathcal{M}_{i,k}^{\left( 2,2\right) }& =\xi _{\gamma _{k}}^{2j-1},\text{\ \
\ \ \ \ \ \ \ \ \ \ \ \ \ \ \ \ \ \ \ \ \ \ \ \ \ \ \ \ \ }i\leq n+p-q,k\leq
n,
\end{align}%
with $r(\lambda ,\xi )$ defined in $\left( \ref{r-instead-t}\right) $. Let
us now rewrite this ratio of determinants as a single determinant:%
\begin{equation}
\frac{\det_{n+p}\mathcal{M}_{t}(\lambda _{I_{q}}|\xi _{\gamma }^{\left(
0\right) }\cup \lambda _{\alpha })}{\det_{q}\mathcal{N}_{t}(\lambda _{I_{q}})%
}=\det_{M+n}\mathcal{W}_{t}(\lambda _{I_{q}}|\xi _{\gamma }^{\left( 0\right)
}\cup \lambda _{\alpha }),
\end{equation}%
where%
\begin{equation}
\mathcal{W}_{t}(\lambda _{I_{q}}|\xi _{\gamma }^{\left( 0\right) }\cup
\lambda _{\alpha })=\left( 
\begin{array}{cc}
I_{p\times p} & \mathcal{W}_{p\times n}^{\left( 1,2\right) } \\ 
\mathcal{W}_{n\times p}^{\left( 2,1\right) } & \mathcal{W}_{n\times
n}^{\left( 2,2\right) }%
\end{array}%
\right)
\end{equation}%
where $I_{p\times p}$ is the $p\times p$ identity matrix and $\mathcal{W}%
_{x\times y}^{\left( a,b\right) }$ are matrices of size $x\times y$, defined
by:%
\begin{align}
\mathcal{W}_{i,k}^{\left( 1,2\right) }& =\left[ \mathcal{N}^{-1}\mathcal{M}%
^{\left( 1,2\right) }\right] _{i,k},\text{ \ \ \ \ \ \ \ \ \ \ \ \ \ \ \ \ \
\ \ \ \ \ \ \ \ \ \ }i\leq p,k\leq n,  \notag \\
\mathcal{W}_{i,k}^{\left( 2,1\right) }& =0,\text{\ \ \ \ \ \ \ \ \ \ \ \ \ \
\ \ \ \ \ \ \ \ \ \ \ \ \ \ \ \ \ \ \ \ \ \ \ \ \ \ \ \ \ }i\leq q-p,k\leq p,
\notag \\
\mathcal{W}_{i,k}^{\left( 2,1\right) }& =\mathcal{M}_{i-(q-p),k}^{\left(
2,1\right) },\text{\ \ \ \ \ \ \ \ \ \ \ \ \ \ \ \ }q-p+1\leq i\leq n,k\leq
p,  \notag \\
\mathcal{W}_{i,k}^{\left( 2,2\right) }& =\left[ \mathcal{N}^{-1}\mathcal{M}%
^{\left( 1,2\right) }\right] _{p+i,k},\text{ \ \ \ \ \ \ \ \ \ \ \ \ \ \ \ \
\ \ }i\leq q-p,k\leq n, \\
\mathcal{M}_{i,k}^{\left( 2,2\right) }& =\xi _{\gamma _{k}}^{2i-1},\text{\ \
\ \ \ \ \ \ \ \ \ \ \ \ \ \ \ \ \ \ \ \ \ \ \ \ \ \ \ \ \ }i\leq n+p-q,k\leq
n.
\end{align}

Let us recall now that from the analysis of the ground state and the
discussion in Section 4.4 of our previous paper \cite{KitKMNST07}, we have
that it holds:%
\begin{equation}
\left[ \mathcal{N}^{-1}\mathcal{M}^{(1,2)}\right] _{i,k}=\frac{\rho (\lambda
_{\alpha _{i}}-\xi _{\gamma _{k}})-\rho (\lambda _{\alpha _{i}}+\xi _{\gamma
_{k}})}{2N\,\rho (\lambda _{\alpha _{i}})}+o\!\left( \frac{1}{N}\right)
\quad \text{if}\ i\leq q-p,k\leq n
\end{equation}%
so that we can write:%
\begin{equation}
\det_{M+n}\mathcal{W}_{t}(\lambda _{I_{q}}|\xi _{\gamma }^{\left( 0\right)
}\cup \lambda _{\alpha })=\det_{n}\mathcal{S}_{t}^{\prime }(\xi _{\gamma
}^{\left( 0\right) },\lambda _{\alpha }),\label{ScalarP-det}
\end{equation}%
with $\mathcal{S}_{t}^{\prime }=\mathcal{W}^{(2,2)}-\mathcal{W}^{(2,1)}%
\mathcal{W}^{(1,2)}$ the $n\times n$ matrices of elements defined by the
following formulae up $o\!\left( 1/N\right) $ terms: 
\begin{align}
\mathcal{S}_{i,k}^{\prime }& =\frac{\rho (\lambda _{\beta _{i}}-\xi _{\gamma
_{k}})-\rho (\lambda _{\beta _{i}}+\xi _{\gamma _{k}})}{2N\,\rho (\lambda
_{\beta _{i}})}\ \ \ \ \ \ \ \ \ \ \ \ \ \ \ \ \ \ \ \ \text{if}\ \ \ i\leq
q-p,k\leq n, \label{Sp-Matrix-1}\\
\mathcal{S}_{q-p+i,k}^{\prime }& =\xi _{\gamma _{k}}^{2i-1}-\sum_{l=1}^{p}%
\mathcal{W}_{i,l}^{(2,1)}\left[ \mathcal{N}^{-1}\mathcal{M}^{(1,2)}\right]
_{l,k}\ \ \ \ \text{if}\ \ 1\leq i\leq n+p-q,k\leq p, \label{Sp-Matrix-2}\\
& =\xi _{\gamma _{k}}^{2i-1}-\sum_{a=1}^{p}\left( \left( \lambda _{\alpha
_{a}}^{\left( 1\right) }\right) ^{2i-1}+\left( \lambda _{\alpha
_{a}}^{\left( 0\right) }\right) ^{2i-1}\right) \frac{\rho (\lambda _{\alpha
_{i}}-\xi _{\gamma _{k}})-\rho (\lambda _{\alpha _{i}}+\xi _{\gamma _{k}})}{%
2N\,\rho (\lambda _{\alpha _{i}})}.\label{Sp-Matrix-3}
\end{align}

Now, let us observe that $p=|\alpha |$ is of the same order of $q$ by
assumption, i.e. it goes to infinity for $N$ going to infinity and we can
write:%
\begin{align}
& \sum_{a=1}^{p}\left( \left( \lambda _{\alpha _{a}}^{\left( 1\right)
}\right) ^{2i-1}+\left( \lambda _{\alpha _{a}}^{\left( 0\right) }\right)
^{2i-1}\right) \frac{\rho (\lambda _{\alpha _{i}}-\xi _{\gamma _{k}})-\rho
(\lambda _{\alpha _{i}}+\xi _{\gamma _{k}})}{2N\,\rho (\lambda _{\alpha
_{i}})}  \notag \\
& =\sum_{l=1}^{q}\left( \left( \lambda _{l}^{\left( 1\right) }\right)
^{2j-1}+\left( \lambda _{l}^{\left( 0\right) }\right) ^{2j-1}\right) \frac{%
\rho (\lambda _{l}-\xi _{\gamma _{k}})-\rho (\lambda _{l}+\xi _{\gamma _{k}})%
}{2N\,\rho (\lambda _{l})}+o(1/N),
\end{align}%
and these sums are finite for any $j\geq 1$. Here, we just compute the one
associated to $j=1$. Let us define:%
\begin{equation}
\tilde{\rho}(\lambda )=\frac{i}{\sinh \pi \lambda }
\end{equation}%
then it holds%
\begin{equation}
\rho (\lambda )=\tilde{\rho}(\lambda +i/2)
\end{equation}%
and we can write:%
\begin{align}
& \left( \lambda ^{\left( 1\right) }+\lambda ^{\left( 0\right) }\right) 
\frac{\rho (\lambda -\xi )-\rho (\lambda +\xi )}{\rho (\lambda )}\left.
=\right. 2\lambda \frac{\tilde{\rho}(\lambda +i/2-\xi )-\tilde{\rho}(\lambda
+i/2+\xi )}{\rho (\lambda )}  \notag \\
& =(\lambda +i/2)\frac{\tilde{\rho}(\lambda +i/2-\xi )}{\rho (\lambda )}%
+(\lambda -i/2)\frac{\tilde{\rho}(\lambda +i/2-\xi )}{\rho (\lambda )}%
-(\lambda +i/2)\frac{\tilde{\rho}(\lambda +i/2+\xi )}{\rho (\lambda )} 
\notag \\
& -(\lambda -i/2)\frac{\tilde{\rho}(\lambda +i/2+\xi )}{\rho (\lambda )} \\
& =\sum_{\sigma =\pm }(\sigma \lambda +i/2)\frac{\tilde{\rho}(\sigma \lambda
+i/2-\xi )}{\rho (\lambda )}-\sum_{\sigma \left. =\right. \pm }(\sigma
\lambda -i/2)\frac{\tilde{\rho}(\sigma \lambda -i/2-\xi )}{\rho (\lambda )}
\end{align}%
where we have used that%
\begin{equation}
\tilde{\rho}(-\lambda )=-\tilde{\rho}(\lambda ),\text{ \ }\tilde{\rho}%
(\lambda \pm i)=-\tilde{\rho}(\lambda ),
\end{equation}%
to get the identities:%
\begin{eqnarray}
(\lambda -i/2)\frac{\tilde{\rho}(\lambda +i/2-\xi )}{\rho (\lambda )}
&=&-(\lambda -i/2)\frac{\tilde{\rho}(\lambda -i/2-\xi )}{\rho (\lambda )} \\
-(\lambda +i/2)\frac{\tilde{\rho}(\lambda +i/2+\xi )}{\rho (\lambda )}
&=&-(-\lambda -i/2)\frac{\tilde{\rho}(-\lambda -i/2-\xi )}{\rho (\lambda )}.
\\
-(\lambda -i/2)\frac{\tilde{\rho}(\lambda +i/2+\xi )}{\rho (\lambda )}
&=&(-\lambda +i/2)\frac{\tilde{\rho}(-\lambda +i/2-\xi )}{\rho (-\lambda )}
\end{eqnarray}%
Now, we can use the following identity holding in the thermodynamic limit: 
\begin{equation}
\sum_{\sigma =\pm }\left( \frac{1}{N}\sum_{j=1}^{q}f(\sigma \lambda
_{j})\right) \underset{N\rightarrow \infty }{\longrightarrow }\sum_{\sigma ={%
\pm }}\int\limits_{0}^{\infty }\,\,f(\sigma \lambda )\rho (\lambda )d\lambda
=\int\limits_{-\infty }^{\infty }\,f(\lambda )\,\rho (\lambda )d\lambda ,
\end{equation}%
to get in this limit 
\begin{equation}
\frac{1}{N}\sum_{l=1}^{q}\lambda _{l}\frac{\rho (\lambda _{l}-\xi _{\gamma
_{k}})-\rho (\lambda _{l}+\xi _{\gamma _{k}})}{\rho (\lambda _{l})}\underset{%
N\rightarrow \infty }{\longrightarrow }\sum_{\sigma ={\pm }}\frac{\sigma }{2}%
\int\limits_{-\infty +\sigma i/2}^{\infty +\sigma i/2}\lambda \tilde{\rho}%
(\lambda -\xi _{\gamma _{k}})d\lambda ,
\end{equation}%
which can be computed by the Residue Theorem. Indeed, we can write:%
\begin{align}
& \int\limits_{-\infty +i/2}^{\infty +i/2}\lambda \tilde{\rho}(\lambda -\xi
_{\gamma _{k}})d\lambda -\int\limits_{-\infty -i/2}^{\infty -i/2}\lambda 
\tilde{\rho}(\lambda -\xi _{\gamma _{k}})d\lambda \left. =\right.
\lim_{R\rightarrow +\infty }\left[ \int\limits_{-R+i/2}^{R+i/2}\lambda 
\tilde{\rho}(\lambda -\xi _{\gamma _{k}})d\lambda \right.  \\
& \left. -\int\limits_{-R-i/2}^{R-i/2}\lambda \tilde{\rho}(\lambda -\xi
_{\gamma _{k}})d\lambda +\int\limits_{-R-i/2}^{-R+i/2}\lambda \tilde{\rho}%
(\lambda -\xi _{\gamma _{k}})d\lambda -\int\limits_{R-i/2}^{R+i/2}\lambda 
\tilde{\rho}(\lambda -\xi _{\gamma _{k}})d\lambda \right]  \\
& =-2i\pi \,\text{Res}\lambda \rho (\lambda ,\xi _{\gamma _{k}})_{\,\vrule %
height13ptdepth1pt\>{\lambda =\xi }\!}=2\xi _{\gamma _{k}},
\end{align}%
for \ $|\text{Im}(\xi _{\gamma _{k}})|\leq 1/2$. Indeed, observing that it holds%
\begin{equation}
|\sinh \pi (\lambda -\xi _{\gamma _{k}})|^{2}=\frac{\cosh 2\pi R-\cos 2\pi
(x+i\xi _{\gamma _{k}})}{2}\geq \frac{\cosh 2\pi R -e^{ 2\pi |\text{Re}(\xi _{\gamma_{k}})|}}{2},%
\text{ \ \ }\forall x\in \lbrack -1/2,1/2]\text{ }
\end{equation}%
for $\lambda =\pm R+ix\in \lbrack \pm R-i/2,\pm R+i/2]$, then, taken $R$ such that $\cosh 2\pi
R\ge e^{ 2\pi |\text{Re}(\xi _{\gamma _{k}})|}$, we get the
estimates:%
\begin{equation}
\left\vert \int\limits_{\pm R-i/2}^{\pm R+i/2}\lambda \tilde{\rho}(\lambda
,\xi _{\gamma _{k}})d\lambda \right\vert \leq \sup_{\lambda \in \lbrack \pm
R-i/2,\pm R+i/2]}\left\vert \lambda \tilde{\rho}(\lambda ,\xi _{\gamma
_{k}})\right\vert \leq \left( \frac{4R^{2}+1}{2\left[ \cosh 2\pi
R -e^{ 2\pi |\text{Re}(\xi _{\gamma _{k}})|}\right] }\right) ^{1/2},
\end{equation}%
which imply%
\begin{equation}
\lim_{R\rightarrow +\infty }\left[ \int\limits_{-R-i/2}^{-R+i/2}\lambda 
\tilde{\rho}(\lambda -\xi _{\gamma _{k}})d\lambda
-\int\limits_{R-i/2}^{R+i/2}\lambda \tilde{\rho}(\lambda -\xi _{\gamma
_{k}})d\lambda \right] =0.
\end{equation}%
So that we get our result:%
\begin{equation}
\frac{1}{N}\sum_{l=1}^{q}\lambda _{l}\left[ \frac{\rho (\lambda _{l}-\xi
_{\gamma _{k}})-\rho (\lambda _{l}+\xi _{\gamma _{k}})}{\rho (\lambda _{l})}%
\right] \underset{N\rightarrow \infty }{\longrightarrow }\xi _{\gamma _{k}},
\end{equation}%
which proves that the line $\mathcal{S}_{q-p+i,k}^{\prime }$ \ goes to zero
for $i=1$ in the thermodynamic limit and so the proposition is proven.
\end{proof}

\section{Correlation functions}

In the following we first develop the analysis of the correlation functions
in the case of diagonal $\bar{K}_{-}(\lambda )$ and upper triangular $\bar{K}%
_{+}(\lambda )$, by imposing%
\begin{equation}
\bar{\mathsf{c}}_{-}=0,\text{ \ }\bar{\mathsf{b}}_{+}\neq 0,
\end{equation}%
hereon when we refer to the transfer matrix $\mathcal{\bar{T}}(\lambda )
$, it is associated to these boundary conditions, as well as the associated
Hamiltonian reads:%
\begin{equation}
\bar{H}=\sum_{i=1}^{N-1}\left[ \sigma _{i}^{x}\sigma _{i+1}^{x}+\sigma
_{i}^{y}\sigma _{i+1}^{y}+\sigma _{i}^{z}\sigma _{i+1}^{z}\right] +\frac{%
\eta }{\bar{\zeta}_{-}}\sigma _{1}^{z}+\frac{\eta }{\bar{\zeta}_{+}}\left[
\sigma _{N}^{z}+\bar{\mathsf{b}}_{+}\sigma _{N}^{+}\right] .  \label{H-bar}
\end{equation}%
We have proven the isospectrality of the transfer matrix $\mathcal{\bar{T}}%
(\lambda )$ and Hamiltonian $\bar{H}$, associated to unparallel boundary
magnetic fields, with the transfer matrix $\mathcal{\hat{T}}(\lambda )$ and
Hamiltonian%
\begin{equation}
\hat{H}=\sum_{i=1}^{N-1}\left[ \sigma _{i}^{x}\sigma _{i+1}^{x}+\sigma
_{i}^{y}\sigma _{i+1}^{y}+\sigma _{i}^{z}\sigma _{i+1}^{z}\right] +\frac{%
\eta }{\bar{\zeta}_{-}}\sigma _{1}^{z}+\frac{\eta }{\bar{\zeta}_{+}}\sigma
_{N}^{z},  \label{H-hat}
\end{equation}%
associated to parallel boundary magnetic fields along the z-direction. Then,
in the next section, we show that in the thermodynamic limit they share the
same zero-temperature correlation functions. While in a subsequent section,
these results are used to state the results for correlation functions of the
original open chain with the following transfer matrix and Hamiltonian:%
\begin{equation}
\mathcal{T}(\lambda )=\Gamma _{W}^{-1}\,\mathcal{\bar{T}}(\lambda )\Gamma
_{W},\text{ \ }\bar{H}=\Gamma _{W}^{-1}\,\bar{H}\Gamma _{W},
\end{equation}
with $\Gamma _{W}$ defined in Section \ref{Similarity-def-Sec}.

\subsection{The case of diagonal and triangular boundary matrices}

The following theorem holds:

\begin{theorem}
\label{Corre-fun}Let us assume that the following boundary conditions are
satisfied:%
\begin{equation}
\bar{\mathsf{c}}_{-}=0,\text{ \ }\bar{\mathsf{b}}_{+}\neq 0,
\end{equation}%
and let us take the following reality conditions:%
\begin{equation*}
\eta =i,\text{ \ }i\bar{\zeta}_{\pm }\in \mathbb{R}
\end{equation*}%
then, in the thermodynamic limit, all the zero-temperature correlation
functions relative to the Hamiltonian $\left( \ref{H-bar}\right) $ with
unparallel magnetic fields coincide with the correlation functions
(associated to the same quasi-local operator) relative to the Hamiltonian $%
\left( \ref{H-hat}\right) $ with parallel magnetic fields along the
z-direction.
\end{theorem}

\begin{proof}
In order to prove the theorem it is enough to prove it for the correlation
functions of a basis of quasi-local operators, so that to prove it we can
use the so-called elementary blocks, i.e. the ground state average of the
basis of monomial of elementary matrices:%
\begin{equation}
\frac{\langle \bar{t}|\prod\limits_{j=1}^{m}E_{j}^{\varepsilon
_{j},\varepsilon _{j}^{\prime }}\mathcal{\bar{B}}_{+}(\lambda
_{I_{q}})|\,0\,\rangle }{\langle \bar{t}|\mathcal{\bar{B}}_{+}(\lambda
_{I_{q}})|\,0\,\rangle }.
\end{equation}%
Let us observe that the following commutation relations hold:%
\begin{equation}
\lbrack \prod\limits_{j=1}^{m}E_{j}^{\varepsilon _{j},\varepsilon
_{j}^{\prime }},S_{z}]=\theta _{\varepsilon ,\varepsilon ^{\prime
}}\prod\limits_{j=1}^{m}E_{j}^{\varepsilon _{j},\varepsilon _{j}^{\prime }}%
\text{, \ \ for }S_{z}=\sum\limits_{j=1}^{N}\sigma _{j}^{z},
\end{equation}%
where the parity $\theta _{\varepsilon ,\varepsilon ^{\prime }}$ is
associated to the integer $s+s^{\prime }$ defined in Section \ref%
{Act-El-Block} by:%
\begin{equation}
\theta _{\varepsilon ,\varepsilon ^{\prime }}=m-(s+s^{\prime }),
\end{equation}%
and by the Proposition \ref{Action-ElBlock}, we have that on the finite
lattice it holds:%
\begin{equation}
\frac{\langle \bar{t}|\prod\limits_{j=1}^{m}E_{j}^{\varepsilon
_{j},\varepsilon _{j}^{\prime }}\mathcal{\bar{B}}_{+}(\lambda
_{I_{q}})|\,0\,\rangle }{\langle \bar{t}|\mathcal{\bar{B}}_{+}(\lambda
_{I_{q}})|\,0\,\rangle }=\sum_{b_{1}=1}^{q}\ldots
\sum_{b_{s}=1}^{q}\sum_{b_{s+1}=1}^{q+m}\ldots \sum_{b_{s+s^{\prime
}}=1}^{q+m}\mathcal{F}_{\beta _{s+s^{\prime }}}^{+}(\lambda _{I_{q+m}})\frac{%
\langle t|\mathcal{\bar{B}}_{+}(\lambda _{\alpha _{-}}\cup \xi _{\gamma
_{+}}^{(0)})|\,0\,\rangle }{\langle t|\mathcal{\bar{B}}_{+}(\lambda
_{I_{q}})|\,0\,\rangle },
\end{equation}%
where $\beta _{s+s^{\prime }}=\{b_{1},\ldots ,b_{s+s^{\prime }}\}$, $\lambda
_{q+j}:=\xi _{m+1-j}^{\left( 0\right) }\text{ for }j\in \{1,\ldots ,m\}$, 
\begin{eqnarray}
\alpha _{-} &=&I_{q}\backslash \alpha _{+},\text{ \ }\alpha _{+}=\beta
_{s+s^{\prime }}\cap I_{q},\text{ \ }\gamma _{+}=\{1,\ldots ,m\}\setminus
\gamma _{-}, \\
\gamma _{-} &=&\{N+m+1-j\}_{j\in \beta _{s+s^{\prime }}\cap \{N+1,\ldots
,N+m\}},
\end{eqnarray}%
and the coefficients $\mathcal{F}_{\beta _{s+s^{\prime }}}^{+}(\lambda
_{I_{q+m}})$ are defined as in Proposition \ref{Action-ElBlock}. We\ can use
now the results on the scalar products of Proposition \ref{Sp-Thermo} to
state that:%
\begin{equation}
\frac{\langle \bar{t}|\prod\limits_{j=1}^{m}E_{j}^{\varepsilon
_{j},\varepsilon _{j}^{\prime }}\mathcal{\bar{B}}_{+}(\lambda
_{I_{q}})|\,0\,\rangle }{\langle \bar{t}|\mathcal{\bar{B}}_{+}(\lambda
_{I_{q}})|\,0\,\rangle }=0\text{ \ for }\theta _{\varepsilon ,\varepsilon
^{\prime }}<0,
\end{equation}%
already on the finite chains, being%
\begin{equation}
\frac{\langle \bar{t}|\mathcal{\bar{B}}_{+}(\lambda _{\alpha _{-}}\cup \xi
_{\gamma _{+}}^{(0)})|\,0\,\rangle }{\langle \bar{t}|\mathcal{\bar{B}}%
_{+}(\lambda _{I_{q}})|\,0\,\rangle }=0\text{ \ for }\theta _{\varepsilon
,\varepsilon ^{\prime }}<0,
\end{equation}%
as $|\alpha _{-}|+|\gamma _{+}|<q$, for any compatible choice of $\alpha
_{-} $\ and $\gamma _{+}$ with $\theta _{\varepsilon ,\varepsilon ^{\prime
}}<0$.

In the case $\theta _{\varepsilon ,\varepsilon ^{\prime }}=0$, here we have
just to point out that from Proposition \ref{Action-ElBlock}, it follows
that the action of an elementary monomial $\prod\limits_{j=1}^{m}E_{j}^{%
\varepsilon _{j},\varepsilon _{j}^{\prime }}$ on the eigenstate $\mathcal{%
\bar{B}}_{+}(\lambda _{I_{q}})|\,0\,\rangle $ of transfer matrix $\mathcal{%
\bar{T}}(\lambda )$ has identical form of the action of the same monomial on
the eigenstate $\mathcal{\hat{B}}_{+}(\lambda _{I_{q}})|\,0\,\rangle $ of
transfer matrix $\mathcal{\hat{T}}(\lambda )$. This observation implies that
for $\theta _{\varepsilon ,\varepsilon ^{\prime }}=0$ the elementary blocks
associated to the Hamiltonian $\bar{H}$\ and $\hat{H}$ coincidence already
for the finite chains, due to the following chain of identity:%
\begin{align}
\frac{\langle \bar{t}|\prod\limits_{j=1}^{m}E_{j}^{\varepsilon
_{j},\varepsilon _{j}^{\prime }}\mathcal{\bar{B}}_{+}(\lambda
_{I_{q}})|\,0\,\rangle }{\langle \bar{t}|\mathcal{\bar{B}}_{+}(\lambda
_{I_{q}})|\,0\,\rangle }& =\sum_{b_{1}=1}^{q}\ldots
\sum_{b_{s}=1}^{q}\sum_{b_{s+1}=1}^{q+m}\ldots \sum_{b_{s+s^{\prime
}}=1}^{q+m}\mathcal{F}_{\beta _{s+s^{\prime }}}^{+}(\lambda _{I_{q+m}}) 
\notag \\
& \times \frac{\langle \bar{t}|\mathcal{\bar{B}}_{+}(\lambda _{\alpha
_{-}}\cup \xi _{\gamma _{+}}^{(0)})|\,0\,\rangle }{\langle \bar{t}|\mathcal{%
\bar{B}}_{+}(\lambda _{I_{q}})|\,0\,\rangle } \\
& =\sum_{b_{1}=1}^{q}\ldots \sum_{b_{s}=1}^{q}\sum_{b_{s+1}=1}^{q+m}\ldots
\sum_{b_{s+s^{\prime }}=1}^{q+m}\mathcal{F}_{\beta _{s+s^{\prime
}}}^{+}(\lambda _{I_{q+m}})  \notag \\
& \times \frac{\langle \hat{t}|\mathcal{\hat{B}}_{+}(\lambda _{\alpha
_{-}}\cup \xi _{\gamma _{+}}^{(0)})|\,0\,\rangle }{\langle \hat{t}|\mathcal{%
\hat{B}}_{+}(\lambda _{I_{q}})|\,0\,\rangle } \\
& =\frac{\langle \hat{t}|\prod\limits_{j=1}^{m}E_{j}^{\varepsilon
_{j},\varepsilon _{j}^{\prime }}\mathcal{\hat{B}}_{+}(\lambda
_{I_{q}})|\,0\,\rangle }{\langle \hat{t}|\mathcal{\hat{B}}_{+}(\lambda
_{I_{q}})|\,0\,\rangle },
\end{align}%
where we have used that by the Proposition \ref{Sp-Thermo} the following
scalar products coincide:%
\begin{equation}
\frac{\langle \bar{t}|\mathcal{\bar{B}}_{+}(\lambda _{\alpha _{-}}\cup \xi
_{\gamma _{+}}^{(0)})|\,0\,\rangle }{\langle \bar{t}|\mathcal{\bar{B}}%
_{+}(\lambda _{I_{q}})|\,0\,\rangle }=\frac{\langle \hat{t}|\mathcal{\hat{B}}%
_{+}(\lambda _{\alpha _{-}}\cup \xi _{\gamma _{+}}^{(0)})|\,0\,\rangle }{%
\langle \hat{t}|\mathcal{\hat{B}}_{+}(\lambda _{I_{q}})|\,0\,\rangle }\text{
\ for }\theta _{\varepsilon ,\varepsilon ^{\prime }}=0.
\end{equation}%
Let us stress that up to here we have shown that the the elementary blocks
associated to the Hamiltonian $\bar{H}$\ and $\hat{H}$ coincidence already
for the finite chains both for $\theta _{\varepsilon ,\varepsilon ^{\prime
}}<0$ and for $\theta _{\varepsilon ,\varepsilon ^{\prime }}=0$.

Instead, in the remaining case $\theta _{\varepsilon ,\varepsilon ^{\prime
}}>0$, we will show that these elementary blocks coincide in the
thermodynamic limit as for the finite chains it holds: 
\begin{equation}
\frac{\langle \hat{t}|\prod\limits_{j=1}^{m}E_{j}^{\varepsilon
_{j},\varepsilon _{j}^{\prime }}\mathcal{\hat{B}}_{+}(\lambda
_{I_{q}})|\,0\,\rangle }{\langle \hat{t}|\mathcal{\hat{B}}_{+}(\lambda
_{I_{q}})|\,0\,\rangle }=0,\text{ \ \ being }\frac{\langle \hat{t}|\mathcal{%
\hat{B}}_{+}(\lambda _{\alpha _{-}}\cup \xi _{\gamma
_{+}}^{(0)})|\,0\,\rangle }{\langle \hat{t}|\mathcal{\hat{B}}_{+}(\lambda
_{I_{q}})|\,0\,\rangle }=0\text{\ \ \ }\theta _{\varepsilon ,\varepsilon
^{\prime }}>0,
\end{equation}

as $|\alpha _{-}|+|\gamma _{+}|>q$, for any compatible choice of $\alpha
_{-} $\ and $\gamma _{+}$ with $\theta _{\varepsilon ,\varepsilon ^{\prime
}}>0$, while it may hold%
\begin{equation}
\frac{\langle \bar{t}|\prod\limits_{j=1}^{m}E_{j}^{\varepsilon
_{j},\varepsilon _{j}^{\prime }}\mathcal{\bar{B}}_{+}(\lambda
_{I_{q}})|\,0\,\rangle }{\langle \bar{t}|\mathcal{\bar{B}}_{+}(\lambda
_{I_{q}})|\,0\,\rangle }\neq 0,\text{ \ \ being }\frac{\langle \bar{t}|%
\mathcal{\bar{B}}_{+}(\lambda _{\alpha _{-}}\cup \xi _{\gamma
_{+}}^{(0)})|\,0\,\rangle }{\langle \bar{t}|\mathcal{\bar{B}}_{+}(\lambda
_{I_{q}})|\,0\,\rangle }\neq 0\text{.}
\end{equation}%
Then, in the current case we have to show the coincidence just in the
thermodynamic limit, i.e.%
\begin{equation}
\lim_{N\rightarrow \infty }\frac{\langle \bar{t}|\prod%
\limits_{j=1}^{m}E_{j}^{\varepsilon _{j},\varepsilon _{j}^{\prime }}\mathcal{%
\bar{B}}_{+}(\lambda _{I_{q}})|\,0\,\rangle }{\langle \bar{t}|\mathcal{\bar{B%
}}_{+}(\lambda _{I_{q}})|\,0\,\rangle }=0\text{\ \ \ for }\theta
_{\varepsilon ,\varepsilon ^{\prime }}>0.
\end{equation}%
In order to prove this we reorder the sums in these elementary block as it
follows:%
\begin{align}
\frac{\langle \bar{t}|\prod\limits_{j=1}^{m}E_{j}^{\varepsilon
_{j},\varepsilon _{j}^{\prime }}\mathcal{\bar{B}}_{+}(\lambda
_{I_{q}})|\,0\,\rangle }{\langle \bar{t}|\mathcal{\bar{B}}_{+}(\lambda
_{I_{q}})|\,0\,\rangle }& =\sum_{g=s}^{s+s^{\prime
}}\sum_{b_{1}=1}^{q}\ldots \sum_{b_{g}=1}^{q}\sum_{b_{g+1}=q+1}^{q+m}\ldots
\sum_{b_{s+s^{\prime }}=q+1}^{q+m}\mathcal{F}_{\beta _{s+s^{\prime
}}}^{+}(\lambda _{I_{q+m}})  \notag \\
& \times \frac{\langle t|\mathcal{\bar{B}}_{+}(\lambda _{\alpha _{-}}\cup
\xi _{\gamma _{+}}^{(0)})|\,0\,\rangle }{\langle t|\mathcal{\bar{B}}%
_{+}(\lambda _{I_{q}})|\,0\,\rangle },\label{EB-sums}
\end{align}%
now in the thermodynamic limit the $q$ diverges as $N$, so that for any
fixed $s\leq g\leq s+s^{\prime }$ the sums%
\begin{equation}
\sum_{b_{1}=1}^{q}\ldots \sum_{b_{g}=1}^{q}
\end{equation}%
leads to a finite $g$ multiple integrals in the thermodynamic limit provided
the integrand is of order $O(1/N^{g})$, while the other sums contribute to
order 1 to the thermodynamic limit. It is now enough to recall that in
Proposition \ref{Sp-Thermo} we have shown 
\begin{equation*}
\frac{\langle \bar{t}|\mathcal{\bar{B}}_{+}(\lambda _{\alpha _{-}}\cup \xi
_{\gamma _{+}}^{(0)})|\,0\,\rangle }{\langle \bar{t}|\mathcal{\bar{B}}%
_{+}(\lambda _{I_{q}})|\,0\,\rangle }=o(1/N^{(g=q-|\alpha _{-}|)}),\ \ \ 
\text{if \ }|\alpha _{-}|+|\gamma _{+}|\left. >\right. q,\label{EB-sums}
\end{equation*}%
to conclude that these sums goes to zero and to prove our statement.
\end{proof}

\bigskip

{\bf Remark:} Here, we want to argue that the possible presence of boundary roots only influence the correlation functions for $s+s’=m$ by a integration contour encircling the relative boundary poles, while leaving unchanged, i.e. zero, the others correlation functions for $s+s’\neq m$ in the thermodynamic limit. The same elementary blocks behaviour has been first derived for the open XXZ spin $1/2$ quantum chains with parallel ($Z$-oriented) boundary magnetic fields \cite{KitKMNST07} and, more recently, it has been extended for the class of elementary blocks computed in \cite{NicT22,NicT23,NicT24,NicT25} for the open XXZ and XYZ spin $1/2$ quantum chains with general unparallel boundary magnetic fields.

To this aim, we have to deduce the type of modifications that the presence of boundary roots can have on the scalar product behavior in the thermodynamic limit w.r.t. that derived in our Proposition \ref{Sp-Thermo}, where these roots were not considered. Here, we follow and adapt to the current setting the analysis developed in \cite{NicT22}. So, if a boundary root is present and there is a $\check i\leq q-p$ such that $\lambda_{\beta_{\check i}}=\check\lambda_-=-\bar{\zeta}_{-}-i/2+\check\epsilon$, with $\check\epsilon$ being an exponentially small correction in $N$, the first $q-p$ rows of the matrix $\mathcal{S}_{t}^{\prime }$ in the determinant of the scalar product (\ref{ScalarP-det}) rewrite as it follows: 
\begin{align}\label{S-divergent}
    \mathcal{S}_{a,b}^{\prime } \underset{N\to\infty}{\sim} \begin{cases}
    \displaystyle i\pi\,\check\epsilon\big[\rho(\lambda_{\beta_{a}}-\xi_{\gamma_b})-\rho(\lambda_{\beta_{a}}+\xi_{\gamma_b})\big] \quad 
          &\text{if }a=\check i,\vspace{1mm}\\
    \displaystyle\frac{\rho(\lambda_{\beta_{a}}-\xi_{\gamma_b})-\rho(\lambda_{\beta_{a}}+\xi_{\gamma_b})}{2N\rho(\lambda_{\beta_{a}})} \quad
          &\text{otherwise}.
    \end{cases}
\end{align}
i.e. for any such $\check i$ there is a row in (\ref{Sp-Matrix-1}) which is now exponentially small as $\check\epsilon$, while the remaining rows of (\ref{ScalarP-det}) still admit the evaluation (\ref{Sp-Matrix-3}) up to one term of order $\check\epsilon$.
So, when we consider elementary blocks with $s+s’=m$, if the boundary root $\check\lambda_-$ belongs to the set of roots, the exponentially small contribution from the corresponding row can be compensated by the prefactor
\begin{equation}\label{divergent}
   \frac{1}{-\check\lambda_- -\bar{\zeta}_{-} -i/2}\underset{N\to\infty}{\sim}-\frac 1{\check\epsilon_-},
\end{equation}
so that the final contribution is of order 1 and can be written as a contour integral around the point $\bar{\zeta}_{-}+i/2$. 

Similarly, when we consider elementary blocks with $s+s’ < m$, if the boundary root $\check\lambda_-$ belongs to the set of roots, the exponentially small contribution from the corresponding row in (\ref{S-divergent}) once again can be compensated singling out the only contribution in the sum over the roots for the elementary block (\ref{EB-sums}) that generate the same divergent prefactor (\ref{divergent}). So the number of sum over the roots will be reduced of one to g-1 in (\ref{EB-sums}) with $s\leq g\leq s+s^{\prime }$ while the scalar product term will go as $o(1/N^{(g-1)})$, taking into account the described compensation of the exponentially small contributions, from which follows our statement that these elementary blocks are zeros in the thermodynamic limits. Finally, in the case  $s+s’ > m$ the associated scalar products are zeros implying that the same is true for the corresponding elementary blocks already for the finite chain.

\bigskip

Then, taking into account the results of the previous Theorem \ref{Corre-fun} and the previous Remark, these elementary blocks can be computed as done in our previous paper \cite{KitKMNST07}, leading to
the following multiple integral representations in the thermodynamic and homogeneous limit: 
\begin{align}
& \frac{\langle \bar{t}|\prod_{j=1}^{m}E_{j}^{\varepsilon _{j},\varepsilon
_{j}^{\prime }}|\bar{t}\rangle }{\langle \bar{t}|\bar{t}\rangle }\left.
=\right. \delta _{\theta _{\varepsilon ,\varepsilon ^{\prime }},0}(-1)^{m-s+%
\frac{m(m-1)}{2}}\bar{\zeta}_{-}^{m}\pi ^{m(m+1)}\int\limits_{\mathcal{C}}\prod\limits_{j=1}^{s}\frac{d\lambda _{j}}{2}\int\limits_{\mathcal{\tilde C}}\prod\limits_{j=s+1}^{m}\frac{d\lambda _{j}}{2}  \notag \\
& \times \prod\limits_{p=1}^{s}\frac{(\lambda _{p}+i/2)^{m+i_{p}-1}(\lambda
_{p}-i/2)^{m-i_{p}}}{\cosh ^{2m}\left( \pi \lambda _{p}\right) }%
\prod\limits_{p=s+1}^{m}\frac{(\lambda _{p}+i/2)^{m+i_{p}-1}(\lambda
_{p}+3i/2)^{m-i_{p}}}{\cosh ^{2m}\left( \pi \lambda _{p}\right) }  \notag \\
& \times \!\!\prod\limits_{k<l}\frac{\sinh \left( \pi \lambda _{kl}\right)
\sinh \left( \pi \bar{\lambda}_{kl}\right) }{(\lambda _{kl}-i)\,(\bar{\lambda%
}_{kl}+i)}\prod\limits_{k=1}^{m}\frac{\sinh \left( \pi \lambda _{k}\right) }{%
\lambda _{k}+i/2+\bar{\zeta}_{-}},
\end{align}%
where the $\{i_{p}\}$ have been defined in $\left( \ref{i_p-Def0}\right) $-$\left( \ref{i_p-Def1}\right) $ and the contour
\begin{equation}
\mathcal{C}=\left\{ 
\begin{array}{l}
\mathbb{R}\text{ \ \ if the boundary root is not contained in the Bethe roots} \\ 
\mathbb{R}\cup \Gamma ^{+}(-i/2-\bar{\zeta}_{-})\text{ \ \ if the boundary root is contained in the Bethe roots\ \ }%
\end{array}%
\right. 
\end{equation}%
and the contour $\mathcal{\tilde C}$ is defined as
\begin{equation}
   \mathcal{\tilde C} =\mathcal{C}\cup \Gamma^+ (-i/2),
\end{equation}
where $\Gamma^+ (x)$ surrounds the point $x$ with index $+1$, all other poles being outside.
\subsection{Non-diagonal case isospectral to the diagonal and triangular one}

Let us consider here our original open XXX spin 1/2 quantum chain with
non-diagonal and non-commutative boundary matrices $K_{\pm }(\lambda )$,
satisfying the following boundary condition:%
\begin{equation}
e^{\tau _{+}}=e^{\tau _{-}}\frac{(\epsilon _{-}\sqrt{1+4\kappa _{-}^{2}}%
+1)(\epsilon _{+}\sqrt{1+4\kappa _{+}^{2}}-1)}{4\kappa _{+}\kappa _{-}},
\label{Tri-Diag-Cond}
\end{equation}%
fix a couple $(\epsilon _{+},\epsilon _{-})\in \{-1,1\}^{2}$, with otherwise
general boundary parameters%
\begin{equation}
\kappa _{+}\neq \pm \kappa _{-},
\end{equation}%
satisfying the following reality condition:%
\begin{equation}
\eta =i,\,i\bar{\zeta}_{\pm }\in \mathbb{R},\,i\xi _{a}\in \mathbb{R},  \label{Reality}
\end{equation}%
then by the Lemma \ref{Lem-iso-cond} its transfer matrix $\mathcal{T}%
(\lambda )$ and Hamiltonian $H$, defined in $\left( \ref{NDtransfer}\right) $
and $\left( \ref{H|XXX-ND}\right) $, are isospectral to the transfer matrix $%
\mathcal{\hat{T}}(\lambda )$ and Hamiltonian $\hat{H} $, both self-adjoint,
respectively. While, by tensor product similarity transformation $\Gamma _{W}
$, $\mathcal{T}(\lambda )$ and $H$ reduce to the transfer matrix $\mathcal{%
\bar{T}}(\lambda )$ and Hamiltonian $\bar{H}$.

Then these similarity transformations and the previous Theorem \ref%
{Corre-fun} allow us to compute the correlation functions/elementary blocks
in the original model associated to the non-diagonal and non-commuting $%
K_{\pm }(\lambda )$ boundary matrices as simple linear combinations of those
of the model associated to the $\bar{K}_{\pm }(\lambda )$ ones.

More in detail, the gauge transformation can be explicitly written, only in
terms of the $K_{-}(\lambda )$ boundary parameters, as it follows:%
\begin{equation}
W=\left( 
\begin{array}{cc}
1 & \frac{-1+\sqrt{1+4\kappa _{-}^{2}}}{2\kappa _{-}e^{-\tau _{-}}} \\ 
\frac{1-\sqrt{1+4\kappa _{-}^{2}}}{2\kappa _{-}e^{\tau _{-}}} & 1%
\end{array}%
\right) ,
\end{equation}%
and so defined%
\begin{equation}
\check{E}_{1,m}^{\left( \{(\varepsilon _{1},\varepsilon _{1}^{\prime
}),...,(\varepsilon _{m},\varepsilon _{m}^{\prime })\}\right) }\equiv \Gamma
_{W}E_{1,m}^{\left( \{(\varepsilon _{1},\varepsilon _{1}^{\prime
}),...,(\varepsilon _{m},\varepsilon _{m}^{\prime })\}\right) }\Gamma
_{W}^{-1}=\prod\limits_{a=1}^{m}\check{E}_{a}^{\left( \varepsilon
_{a},\varepsilon _{a}^{\prime }\right) },
\end{equation}%
with:%
\begin{equation}
\check{E}_{a}^{\left( \varepsilon _{a},\varepsilon _{a}^{\prime }\right)
}=W_{a}E_{a}^{\left( \varepsilon _{a},\varepsilon _{a}^{\prime }\right)
}W_{a}^{-1},
\end{equation}%
the generic $m$-site elementary block in the original model%
\begin{equation}
\langle E_{1,m}^{\left( \{(\varepsilon _{1},\varepsilon _{1}^{\prime
}),...,(\varepsilon _{m},\varepsilon _{m}^{\prime })\}\right) }\rangle
_{ND}\equiv \frac{\langle \,t\,|\prod\limits_{j=1}^{n}E_{j}^{\left(
\varepsilon _{a},\varepsilon _{a}^{\prime }\right) }|t\,\rangle }{\langle
t\,|t\rangle },
\end{equation}%
coincides with%
\begin{align}
\frac{\langle \,\bar{t}\,|\Gamma _{W}E_{1,m}^{\left( \{(\varepsilon
_{1},\varepsilon _{1}^{\prime }),...,(\varepsilon _{m},\varepsilon
_{m}^{\prime })\}\right) }\Gamma _{W}^{-1}|\bar{t}\,\rangle }{\langle \bar{t}%
\,|\bar{t}\rangle }& =\frac{\langle \,\hat{t}\,|\Gamma _{W}E_{1,m}^{\left(
\{(\varepsilon _{1},\varepsilon _{1}^{\prime }),...,(\varepsilon
_{m},\varepsilon _{m}^{\prime })\}\right) }\Gamma _{W}^{-1}|\hat{t}\,\rangle 
}{\langle \hat{t}\,|\hat{t}\rangle }  \notag \\
& \equiv \langle \check{E}_{1,m}^{\left( \{(\varepsilon _{1},\varepsilon
_{1}^{\prime }),...,(\varepsilon _{m},\varepsilon _{m}^{\prime })\}\right)
}\rangle _{D},
\end{align}%
i.e. a sum of up to $2^{2m}$ elementary blocks of the open XXX spin 1/2
chain with $\hat{K}_{\pm }(\lambda )$ diagonal boundary matrices. Let us
write explicitly these decompositions of the one-point and two-point
correlation functions of the original model in terms of those of the
diagonal one:

\begin{cor}
The following one-point functions of the original model:%
\begin{align}
\langle E_{m}^{(1,1)}\rangle _{ND}& =\frac{2\kappa _{-}^{2}\langle
E_{m}^{(1,1)}\rangle _{D}}{1+4\kappa _{-}^{2}-\sqrt{1+4\kappa _{-}^{2}}}+%
\frac{\sqrt{1+4\kappa _{-}^{2}}-1}{2\sqrt{1+4\kappa _{-}^{2}}}\langle
E_{m}^{(2,2)}\rangle _{D}, \\
\langle E_{m}^{(2,1)}\rangle _{ND}& =\frac{e^{\tau _{-}}\kappa _{-}\langle
\sigma _{m}^{z}\rangle _{D}}{\sqrt{1+4\kappa _{-}^{2}}},\text{ \ }\langle
E_{m}^{(1,2)}\rangle _{ND}=\frac{\kappa _{-}\langle \sigma _{m}^{z}\rangle
_{D}}{e^{\tau _{-}}\sqrt{1+4\kappa _{-}^{2}}}, \\
\langle E_{m}^{(2,2)}\rangle _{ND}& =\frac{2\kappa _{-}^{2}\langle
E_{m}^{(2,2)}\rangle _{D}}{1+4\kappa _{-}^{2}-\sqrt{1+4\kappa _{-}^{2}}}+%
\frac{\sqrt{1+4\kappa _{-}^{2}}-1}{2\sqrt{1+4\kappa _{-}^{2}}}\langle
E_{m}^{(1,1)}\rangle _{D},
\end{align}%
admit the above representations in terms of the nonzero one-point
functions of the diagonal model, so that the magnetization profile reads:%
\begin{equation}
\langle \sigma _{m}^{z}\rangle _{ND}=\frac{\langle \sigma _{m}^{z}\rangle
_{D}}{\sqrt{1+4\kappa _{-}^{2}}}\text{.}
\end{equation}%
Moreover, for the following two-point functions we have\footnote{%
Here, we have written only the standard three two point functions but the
other possible five are also easily written as
linear combinations.}:%
\begin{align}
\langle \sigma _{1}^{z}\sigma _{m+1}^{z}\rangle _{ND}& =\frac{\langle \sigma
_{1}^{z}\sigma _{m+1}^{z}\rangle _{D}}{1+4\kappa _{-}^{2}}+\frac{4\kappa
_{-}^{2}\left( \langle \sigma _{1}^{+}\sigma _{m+1}^{-}\rangle _{D}+\langle
\sigma _{1}^{-}\sigma _{m+1}^{+}\rangle _{D}\right) }{1+4\kappa _{-}^{2}}, \\
\langle \sigma _{1}^{-}\sigma _{m+1}^{+}\rangle _{ND}& =\frac{\kappa
_{-}^{2}\langle \sigma _{1}^{z}\sigma _{m+1}^{z}\rangle _{D}}{1+4\kappa
_{-}^{2}}+\frac{2\kappa _{-}^{4}\langle \sigma _{1}^{+}\sigma
_{m+1}^{-}\rangle _{D}}{1+4\kappa _{-}^{2}-\sqrt{1+4\kappa _{-}^{2}}}+\frac{%
\sqrt{1+4\kappa _{-}^{2}}-1}{\sqrt{1+4\kappa _{-}^{2}}}\langle \sigma
_{1}^{-}\sigma _{m+1}^{+}\rangle _{D}, \\
\langle \sigma _{1}^{+}\sigma _{m+1}^{-}\rangle _{ND}& =\frac{\kappa
_{-}^{2}\langle \sigma _{1}^{z}\sigma _{m+1}^{z}\rangle _{D}}{1+4\kappa
_{-}^{2}}+\frac{2\kappa _{-}^{4}\langle \sigma _{1}^{-}\sigma
_{m+1}^{+}\rangle _{D}}{1+4\kappa _{-}^{2}-\sqrt{1+4\kappa _{-}^{2}}}+\frac{%
\sqrt{1+4\kappa _{-}^{2}}-1}{\sqrt{1+4\kappa _{-}^{2}}}\langle \sigma
_{1}^{+}\sigma _{m+1}^{-}\rangle _{D}.
\end{align}
\end{cor}

It is worth commenting that the elementary blocks of correlation functions
can be used as a basis to decompose any correlation function in terms of
them. From this point of view the natural elementary block basis to be used
in our original open XXX spin 1/2 quantum chain with non-diagonal and
non-commuting boundary matrices $K_{\pm }(\lambda )$ satisfying $\left( \ref%
{Tri-Diag-Cond}\right) $ and the reality conditions $\left( \ref{Reality}%
\right) $ is defined in the following:

\begin{cor}
Let us define the following quasi-local operators:%
\begin{equation}
\bar{E}_{1,m}^{\left( \{(\varepsilon _{1},\varepsilon _{1}^{\prime
}),...,(\varepsilon _{m},\varepsilon _{m}^{\prime })\}\right) }=\Gamma
_{W}^{-1}E_{1,m}^{\left( \{(\varepsilon _{1},\varepsilon _{1}^{\prime
}),...,(\varepsilon _{m},\varepsilon _{m}^{\prime })\}\right) }\Gamma
_{W}=\prod\limits_{a=1}^{m}\bar{E}_{a}^{\left( \varepsilon _{a},\varepsilon
_{a}^{\prime }\right) },
\end{equation}%
where:%
\begin{equation}
\bar{E}_{a}^{\left( \varepsilon _{a},\varepsilon _{a}^{\prime }\right)
}=W_{a}^{-1}E_{a}^{\left( \varepsilon _{a},\varepsilon _{a}^{\prime }\right)
}W_{a}.
\end{equation}%
Then the following identities holds:%
\begin{equation}
\langle \bar{E}_{1,m}^{\left( \{(\varepsilon _{1},\varepsilon _{1}^{\prime
}),...,(\varepsilon _{m},\varepsilon _{m}^{\prime })\}\right) }\rangle
_{ND}=\langle E_{1,m}^{\left( \{(\varepsilon _{1},\varepsilon _{1}^{\prime
}),...,(\varepsilon _{m},\varepsilon _{m}^{\prime })\}\right) }\rangle _{D},
\end{equation}%
which, in particular means, that the non-zero elementary blocks are only
those associated to the $\bar{E}_{1,m}^{\left( \{(\varepsilon
_{1},\varepsilon _{1}^{\prime }),...,(\varepsilon _{m},\varepsilon
_{m}^{\prime })\}\right) }$ commuting with:%
\begin{equation}
\bar{E}_{z}=\sum\limits_{a=1}^{N}\bar{\sigma}_{a}^{z},
\end{equation}%
with 
\begin{equation}
\bar{\sigma}_{a}^{z}=W_{a}^{-1}E_{a}^{\left( \varepsilon _{a},\varepsilon
_{a}^{\prime }\right) }W_{a}=\frac{\sigma _{a}^{z}+2\kappa _{-}(e^{\tau
_{-}}\sigma _{a}^{+}+e^{-\tau _{-}}\sigma _{a}^{-})}{\sqrt{1+4\kappa _{-}^{2}%
}}.
\end{equation}
\end{cor}

\section{Conclusion}

We have shown that the correlation functions for the open XXX spin 1/2 chain
with unparallel boundary magnetic fields are written as linear combination
of those of the open XXX spin 1/2 chain with parallel boundary magnetic
fields, whose multiple integral formulae were derived in \cite%
{KitKMNST07,KitKMNST08}.

The main technical novelties here developed are the computations of the
boundary-bulk decomposition of the boundary separate states together with
the computation of the action of local operators on separate states for
these open XXX quantum chains with unparallel magnetic fields. Let us
comment that the boundary-bulk decomposition of boundary separate states is
required as we have at our disposal only a reconstruction of local operators
in terms of generators of the Yang-Baxter algebra and not in terms of the
boundary generators of the reflection algebra \cite{KitKMNST07,KitKMNST08}.
Then, the boundary-bulk decomposition allows us to compute the action of
local operators over boundary separate states by acting on analogue bulk
states. Here, the main difficulty is to rearrange this action as a linear
combination of purely boundary separate states. These fundamental steps to computed correlation functions have been here solved for the
XXX open spin 1/2 quantum chains but they have served as simplified examples and starting point to develop the required analysis for the more involved XXZ/XYZ open quantum spin 1/2 chains then developed in \cite{NicT22,NicT23,NicT24,NicT25}.

Let us comment that the results of Section 6.1 provide a proof of the
physical expectation that correlation functions of a quasi-local operator,
on finite number of sites, should coincide in the thermodynamic limit for
two open chains that share the same boundary magnetic fields on the site 1.
In fact, we have proven it here in the special constrained case, for which
the two open chains share the same boundary matrix in site 1 but one
diagonal and the other triangular in the site $N$. It is worth commenting
that the main technical reason for our choice to compute correlation
functions under this constrain (allowing for an homogeneous $TQ$-equation
formulation of the transfer matrix spectrum) is due to the reduced knowledge
of the thermodynamic limit of the ground state for the most general
unconstrained non-diagonal boundary matrices. To achieve the full control of
this ground state distribution can then make possible the computation of
correlation functions in this most general unconstrained case and will be of
central interest for future research in the open chain framework.

\section*{Acknowledgments}

G. N. is supported by CNRS and Laboratoire de Physique, ENS-Lyon. G. N.
would like to thanks V. Terras for interesting discussions on topics related
to the current paper.

\end{document}